%% file: 0_0_main.tex
\documentclass[11pt]{article}

\input{packages}

\input{macros_letters}

\input{macros}
\usepackage{comment}
\usepackage[T1]{fontenc}

\usepackage[hidelinks]{hyperref}  % load this after the definitions in macros.tex -NEAL 

\algnewcommand\algorithmicinput{\textbf{Input:}}
\algnewcommand\algorithmicoutput{\textbf{Output:}}
\algnewcommand\Input{\item[\algorithmicinput]}%
\algnewcommand\Output{\item[\algorithmicoutput]}%

\newcommand{\eat}[1]{}

\title{Online Paging with Heterogeneous Cache Slots\thanks
  {Conference and journal versions of this paper appear in STACS and Algorithmica~\cite{chrobak+hlprsy:heterogeneous:conference,chrobak+hlprsy:heterogeneous:journal}.}}
\author[1]{Marek Chrobak\thanks{Research partially supported by National Science Foundation grants CCF-1536026 and CCF-2153723.}}
\author[2]{Samuel Haney\thanks{Research partially supported by National Science Foundation grants CCF-1527084 and CCF-1535972.}}
\author[3]{Mehraneh Liaee\thanks{Research partially supported by National Science Foundation grants CCF-1535929 and CCF-1909363.}}
\author[4]{Debmalya Panigrahi\thanks{Research partially supported by National Science Foundation grants CCF-1527084, CCF-1535972, CCF-1750140, CCF-1955703, an Army Research Office grant W911NF2110230, and the Indo-US Joint Center for Algorithms under Uncertainty.}}
\author[3]{Rajmohan Rajaraman\thanks{Research partially supported by National Science Foundation grants CCF-1535929 and CCF-1909363.}}
\author[3]{Ravi Sundaram\thanks{Research partially supported by National Science Foundation grants CCF-1535929 and IIS-2039945.}}
\author[1]{Neal E. Young\thanks{Research partially supported by National Science Foundation grant CCF-1619463.}}
\affil[1]{University of California Riverside}
\affil[2]{Tumult Labs}
\affil[3]{Northeastern University}
\affil[4]{Duke University}

\date{}

\begin{document}
	
\maketitle
\begin{abstract}
  It is natural to generalize the online \kServer problem
  by allowing each request to specify not only a point $p$, but also a subset $S$ of servers that may serve it.
  % \mareksrevision{
  To date, only a few special cases of this problem have been studied.
 % \reporttag{(R2.b)}                                                                            % <------------ report tag
  The objective of the work presented in this paper has been to more systematically
  explore this generalization in the case of uniform and star metrics.
  % }
  For uniform metrics, the problem is equivalent to a generalization of \Paging
  in which each request specifies not only a page $p$, but also a subset $S$ of cache slots,
  and is satisfied by having a copy of $p$ in some slot in $S$. We call this problem \emph{\SlotHeteroPaging}.

  In realistic settings only certain subsets of cache slots or servers would appear in requests. Therefore
  we parameterize the problem by specifying a family $\slotsetfamily \subseteq 2^{[k]}$ of requestable slot sets,
  and we establish bounds on the competitive ratio as a function of the cache size $k$ and family $\slotsetfamily$:
  \begin{itemize}
  \item  If all request sets are allowed ($\slotsetfamily=2^{[k]}\setminus\{\emptyset\}$),
    the optimal deterministic and randomized competitive ratios
    are exponentially worse than for standard \Paging ($\slotsetfamily=\{[k]\}$).
  \item
    As a function of $|\slotsetfamily|$ and $k$, the optimal deterministic ratio is polynomial:
    at most $O(k^2|\slotsetfamily|)$ and at least $\Omega(\sqrt{|\slotsetfamily|})$.
  \item  
    For any laminar family $\slotsetfamily$ of height $h$, the optimal ratios
    are $O(hk)$ (deterministic) and $O(h^2\log k)$ (randomized).  
  \item The special case of laminar $\slotsetfamily$ that we call \emph{\AllOrOnePaging}
    extends standard \Paging by allowing each request to specify a specific slot to put the requested page in.
    The optimal deterministic ratio for \emph{weighted} \AllOrOnePaging is $\Theta(k)$.
    Offline \AllOrOnePaging is \NP-hard.
  \end{itemize}
  Some results for the laminar case are shown via a reduction to the generalization of \Paging
  in which each request specifies a set $\pageset$ of \emph{pages},
  and is satisfied by fetching any page from $\pageset$ into the cache.  
  The optimal ratios for the latter problem (with laminar family of height $h$)
  are at most $hk$ (deterministic) and $hH_k$ (randomized).
  \\[1em]
\end{abstract}

\thispagestyle{empty}
\newpage 

\pagenumbering{arabic}

%%%%%%%%%%%%%%%%%%%%%%%%%%%%%%%%%%%%%%%%%%%%%%%%%%%%%%%%%%%%%%%%%%%%%%%%%%%%%%%%%%%%%%
%%%%%%%%%%%%%%%%%%%%%%%%%%%%%%%%%%%%%%%%%%%%%%%%%%%%%%%%%%%%%%%%%%%%%%%%%%%%%%%%%%%%%%

\section{Introduction}%
\label{sec: introduction}
\input{1_introduction}

%%%%%%%%%%%%%%%%%%%%%%%%%%%%%%%%%%%%%%%%%%%%%%%%%%%%%%%%%%%%%%%%%%%%%%%%%%%%%%%%%%%%%%
%%%%%%%%%%%%%%%%%%%%%%%%%%%%%%%%%%%%%%%%%%%%%%%%%%%%%%%%%%%%%%%%%%%%%%%%%%%%%%%%%%%%%%

\section{Formal Definitions}%
\label{sec: preliminaries}
\input{2_preliminaries}

%%%%%%%%%%%%%%%%%%%%%%%%%%%%%%%%%%%%%%%%%%%%%%%%%%%%%%%%%%%%%%%%%%%%%%%%%%%%%%%%%%%%%
%%%%%%%%%%%%%%%%%%%%%%%%%%%%%%%%%%%%%%%%%%%%%%%%%%%%%%%%%%%%%%%%%%%%%%%%%%%%%%%%%%%%%

\section{\SlotHeteroPaging}% 
\label{sec: slot hetero paging}
\input{3_0_slot_hetero_paging}

\subsection{Upper bounds for deterministic  \SlotHeteroPaging}%
\label{sec: slot hetero upper bound}
\input{3_1_sh_paging_upper_bounds}

\subsection{Lower bounds for deterministic \SlotHeteroPaging}%
\label{sec: deterministic slot hetero lower bounds}
\input{3_2_deterministic_sh_paging_lower_bounds}

\subsection{Lower bound for randomized \SlotHeteroPaging}%
\label{sec: randomized slot hetero lower bounds}
\input{3_3_randomized_sh_paging_lower_bounds}

%%%%%%%%%%%%%%%%%%%%%%%%%%%%%%%%%%%%%%%%%%%%%%%%%%%%%%%%%%%%%%%%%%%%%%%%%%%%%%%%%%%%%
%%%%%%%%%%%%%%%%%%%%%%%%%%%%%%%%%%%%%%%%%%%%%%%%%%%%%%%%%%%%%%%%%%%%%%%%%%%%%%%%%%%%%

\section{Upper Bounds for \PageLaminarPaging}%
\label{sec: page laminar paging}

\input{4_page_laminar_paging}

%%%%%%%%%%%%%%%%%%%%%%%%%%%%%%%%%%%%%%%%%%%%%%%%%%%%%%%%%%%%%%%%%%%%%%%%%%%%%%%%%%%%%
%%%%%%%%%%%%%%%%%%%%%%%%%%%%%%%%%%%%%%%%%%%%%%%%%%%%%%%%%%%%%%%%%%%%%%%%%%%%%%%%%%%%%

\section{\SlotLaminarPaging}%
\label{sec: slot laminar paging}

\input{5_0_laminar_sc_paging}

\subsection{Upper bounds for randomized and offline  \SlotLaminarPaging}%
\label{sec: reduction of slot laminar to paging}

\input{5_1_laminar_reduction}

\subsection{Improved upper bound for deterministic \SlotLaminarPaging}%
\label{sec: laminar k-server}

\input{5_2_laminar_deterministic_upper_bound}

%%%%%%%%%%%%%%%%%%%%%%%%%%%%%%%%%%%%%%%%%%%%%%%%%%%%%%%%%%%%%%%%%%%%%%%%%%%%%%%%%%%%%
%%%%%%%%%%%%%%%%%%%%%%%%%%%%%%%%%%%%%%%%%%%%%%%%%%%%%%%%%%%%%%%%%%%%%%%%%%%%%%%%%%%%%

\section{\AllOrOnePaging}%
\label{sec: all or one paging}

\input{6_0_all_or_one_paging}

\subsection{Lower bound for randomized \AllOrOnePaging}%
\label{sec: all or one paging improved lower}

\input{6_1_all_or_one_lower}

\subsection{NP-completeness of offline \AllOrOnePaging}%
\label{sec: np-completeness of all or one paging}

\input{6_2_all_or_one_np_completeness}

%%%%%%%%%%%%%%%%%%%%%%%%%%%%%%%%%%%%%%%%%%%%%%%%%%%%%%%%%%%%%%%%%%%%%%%%%%%%%%%%%%%%%
%%%%%%%%%%%%%%%%%%%%%%%%%%%%%%%%%%%%%%%%%%%%%%%%%%%%%%%%%%%%%%%%%%%%%%%%%%%%%%%%%%%%%

\section{\WeightedAllOrOnePaging}%
\label{sec: weighted all or one paging}
\input{7_weighted_all_or_one}

%%%%%%%%%%%%%%%%%%%%%%%%%%%%%%%%%%%%%%%%%%%%%%%%%%%%%%%%%%%%%%%%%%%%%%%%%%%%%%%%%%%%%
%%%%%%%%%%%%%%%%%%%%%%%%%%%%%%%%%%%%%%%%%%%%%%%%%%%%%%%%%%%%%%%%%%%%%%%%%%%%%%%%%%%%%

\section{Open Problems}%
\label{sec: open problems}

\input{8_open_problems}

%%%%%%%%%%%%%%%%%%%%%%%%%%%%%%%%%%%%%%%%%%%%%%%%%%%%%%%%%%%%%%%%%%%%%%%%%%%%%%%%%%%%%
%%%%%%%%%%%%%%%%%%%%%%%%%%%%%%%%%%%%%%%%%%%%%%%%%%%%%%%%%%%%%%%%%%%%%%%%%%%%%%%%%%%%%

\bigskip
\myparagraph{Conflict of interest statement.}
The authors have no relevant financial or non-financial interests to disclose.

%%%%%%%%%%%%%%%%%%%%%%%%%%%%%%%%%%%%%%%%%%%%%%%%%%%%%%%%%%%%%%%%%%%%%%%%%%%%%%%%%%%%%
%%%%%%%%%%%%%%%%%%%%%%%%%%%%%%%%%%%%%%%%%%%%%%%%%%%%%%%%%%%%%%%%%%%%%%%%%%%%%%%%%%%%%

\bibliographystyle{plainurl}
\bibliography{references.bib}

\end{document}

%% file: packages.tex
\usepackage[margin=2.5cm]{geometry}
\usepackage{authblk}

\usepackage{amssymb,amsmath,amsthm,amsfonts}
\usepackage{thmtools}
\usepackage{color}

\usepackage{yfonts}
\usepackage{mathrsfs}

\usepackage{graphicx}
\usepackage{xspace}

\usepackage{times}
\usepackage{verbatim}
\usepackage{enumitem}
\usepackage{graphicx,wrapfig,lipsum}
\usepackage[font=small]{caption}
\usepackage{subcaption}
\usepackage[english]{babel}
\usepackage[utf8]{inputenc}
\usepackage{algorithm}
\usepackage[noend]{algpseudocode}
\usepackage{url}
\usepackage[all]{xy}
\usepackage[noadjust]{cite}
\usepackage{rotating}
\usepackage{ifpdf}
\usepackage{tcolorbox}
\usepackage{lipsum}

\usepackage{algorithm}
\usepackage{algpseudocode}

\usepackage{mdframed}             % \begin{framed}
\usepackage{xifthen} % \ifthenelse \isempty
\usepackage{mathtools}  % \coloneq

\DeclareMathAlphabet\mathbfcal{OMS}{cmsy}{b}{n}

%% file: macros_letters.tex
\newcommand{\complB}{{\overline{B}}} 
\newcommand{\complZ}{{\overline{Z}}}  
\newcommand{\complX}{{\overline{X}}}  
\newcommand{\complY}{{\overline{Y}}}  
\newcommand{\complS}{{\overline{S}}}

\newcommand{\calC}{\ensuremath{\mathcal C}\xspace}

\newcommand{\calG}{\ensuremath{\mathcal G}\xspace}
\newcommand{\calH}{\ensuremath{\mathcal H}\xspace}

\newcommand{\calM}{\ensuremath{\mathcal M}\xspace}

\newcommand{\calZ}{\ensuremath{\mathcal Z}\xspace}

\newcommand{\eulere}{{\textup{\textrm{e}}}}

%% file: macros.tex
\colorlet{darkgreen}{green!45!black}

\newcommand{\slotset}{{S}} 
\newcommand{\slotsetfamily}{{\mathcal{S}}} 
\newcommand{\slotsetfamilyclosure}{{{\slotsetfamily}^\ast}}

\newcommand{\slotrequestsequence}{\sigma}
\newcommand{\slotrequest}{\sigma}
\newcommand{\slotrequestpair}[2]{{ \langle {#1}, {#2} \rangle }}
\newcommand{\requestpair}[2]{{ \langle {#1}, {#2} \rangle }}

\newcommand{\pageset}{P} 
\newcommand{\pagesetfamily}{{\mathcal P}}
\newcommand{\pagerequestsequence}{\pi}
\newcommand{\pagerequest}{\pageset} 
\newcommand{\catratrequestsequence}{\mu}

\newcommand{\requestsequence}{\sigma}
\newcommand{\twopageconfig}{{\textit{Q}}}
\newcommand{\forcingseq}{{\psi}}

\newcommand{\solutionC}{{C}}

\newcommand{\kServer}{$k$-Server\xspace}
\newcommand{\Paging}{Paging\xspace}

\newcommand{\HeteroServer}{Heterogenous \kServer}

\newcommand{\GenerServer}{Generalized \kServer}
\newcommand{\AllOrOnePaging}{All-or-One \Paging}
\newcommand{\WeightedPaging}{Weighted \Paging}
\newcommand{\WeightedAllOrOnePaging}{Weighted All-Or-One \Paging}

\newcommand{\SlotHeteroPaging}{Slot-Heterogenous \Paging}
\newcommand{\OneOutOfmPaging}[1]{One-of-$#1$ \Paging}
\newcommand{\PageSubsetPaging}{Page-Subset \Paging}

\newcommand{\SlotLaminarPaging}{Slot-Laminar \Paging}
\newcommand{\PageLaminarPaging}{Page-Laminar \Paging}

\newcommand{\prevc}[2]{c_{#2}({#1})}
\newcommand{\prevp}[2]{p_{#2}({#1})}

\newcommand{\algA}{\ensuremath{\mathbb{A}}\xspace}
\newcommand{\algB}{\ensuremath{\mathbb{B}}\xspace}

\newcommand{\algExhSearch}{\text{\sc{ExhSearch}}\xspace}
\newcommand{\algRefSearch}{\text{\sc{RefSearch}}\xspace}

\newcommand{\NAME}[1]{\ensuremath{\mathsf{#1}}\xspace}

\newcommand{\sumofsizes}{\NAME{mass}}

\newcommand{\cost}{\NAME{cost}}
\newcommand{\opt}{\NAME{opt}}

\newcommand{\wt}[1]{\NAME{wt}(#1)}

\newcommand{\capacity}{\NAME{cap}} 
\newcommand{\credit}{\NAME{credit}}
\newcommand{\artificialpage}{{\zeta}}

\newcommand{\rep}{\NAME{rep}}
\newcommand{\dep}{\NAME{anc}}

\newcommand{\edgegadget}{{\boldsymbol{\omega}}}

\newcommand{\C}{C}                % given solution
\newcommand{\defn}[1]{\margincomment{def}{\emph{\color{blue}#1}}\emph{#1}}
\renewcommand{\defn}[1]{\emph{#1}}

\newtheorem{theorem}{Theorem}[section]
\newtheorem{lemma}{Lemma}[section]
\newtheorem{corollary}[lemma]{Corollary}
\newtheorem{claim}[lemma]{Claim}

\newtheorem{observation}[lemma]{Observation}

\newcommand{\ignore}[1]{}

\newcommand{\margincomment}[2]%
{\marginpar{\footnotesize\raggedright {\color{red}#1}: #2}}

\newcommand{\reporttag}[1]%
{\marginpar{\footnotesize\raggedright {\bf\color{red}#1}}}

\newcommand{\myparagraph}[1]{{\medskip\noindent\textbf{#1}}}
\newcommand{\myitparagraph}[1]{{\medskip\noindent\textit{#1}}}

\newcommand{\braced}[1]{{ \left\{ {#1} \right\} }}
\newcommand{\smallbraced}[1]{{ \{ {#1} \} }}

\newcommand{\suchthat}{{\;:\;}}
\newcommand{\NP}{\ensuremath{\mathbb{NP}}\xspace}

\newcommand{\R}{{\mathbb{R}}}

\newcommand{\poly}{{\rm poly}}

\newcommand{\reals}{{\mathbb R}}

\newcommand{\razy}{\,} 
\newcommand{\mytimes}{\,{\cdot}\,}
\newcommand{\polylog}{\operatorname{polylog}}
\newcommand\floor[1]{\lfloor#1\rfloor}

\newcommand{\ee}{{\textrm{e}}}

\newcommand{\complements}[1]{{\widetilde{#1}}}

\newcommand{\two}{2\,}
\newcommand{\three}{3\,}

\newcommand{\six}{6\,}

\newcommand{\half}{\textstyle{\frac{1}{2}}}

\newcommand{\ListLengths}{\setlength{\itemsep}{0ex}\setlength{\topsep}{1ex}\setlength{\partopsep}{0ex}}

\newlist{steps}{enumerate}{7}
\setlist[steps]{wide,
  topsep=1pt,
  parsep=2pt,
  partopsep=1pt,
  itemsep=0pt,
  align=left,
  labelindent=0pt,
  itemindent=0pt,
  leftmargin=1pt, 
  rightmargin=0pt,
  labelsep=4pt,
  labelwidth=0pt,
}
\setlist[steps,1]{
  leftmargin=6pt, 
  label=\small\arabic*.,
  ref=\arabic*
}
\setlist[steps,2]{
  labelsep=5pt,
  label=\small\arabic{stepsi}.\arabic*.,
  ref=\arabic{stepsi}.\arabic*
}
\setlist[steps,3]{
  labelsep=6pt, 
  label=\small\arabic{stepsi}.\arabic{stepsii}.\arabic*.,
  ref=\arabic{stepsi}.\arabic{stepsii}.\arabic*
}
\setlist[steps,4]{
  labelsep=7pt, 
  label=\small\arabic{stepsi}.\arabic{stepsii}.\arabic{stepsiii}.\arabic*.,
  ref=\arabic{stepsi}.\arabic{stepsii}.\arabic{stepsiii}.\arabic*
}
\setlist[steps,5]{
  labelsep=8pt, 
  label=\small\arabic{stepsi}.\arabic{stepsii}.\arabic{stepsiii}.\arabic{stepsiv}.\arabic*.,
  ref=\arabic{stepsi}.\arabic{stepsii}.\arabic{stepsiii}.\arabic{stepsiv}.\arabic*
}
\setlist[steps,6]{
  labelsep=9pt,  
label=\small\arabic{stepsi}.\arabic{stepsii}.\arabic{stepsiii}.\arabic{stepsiv}.\arabic{stepsv}.\arabic*.,
  ref=\arabic{stepsi}.\arabic{stepsii}.\arabic{stepsiii}.\arabic{stepsiv}.\arabic{stepsv}.\arabic*
}
\setlist[steps,7]{
  labelsep=10pt, 
  label=\small\arabic{stepsi}.\arabic{stepsii}.\arabic{stepsiii}.\arabic{stepsiv}. \arabic{stepsv}. \arabic{stepsvi}.\arabic*.,
  ref=\arabic{stepsi}.\arabic{stepsii}.\arabic{stepsiii}.\arabic{stepsiv}.\arabic{stepsv}. \arabic{stepsvi}.\arabic*
}

\newcommand{\step}[1][]{\item\maybeLabel{#1}}
\newcommand{\interval}[2]{[#1, #2]}

\makeatletter
\newcommand{\myskipitem}{\refstepcounter{\@enumctr}}
\makeatother

\newcommand{\maybeLabel}[1]{\ifthenelse{\isempty{#1}}{}{\label{#1}}}

\newcommand{\algcomment}[1]{\hfill{\footnotesize\emph{#1}}}

\newif\ifshort
\shorttrue

%% file: 1_introduction.tex
%%%%%%%%%%%%%%%%

The standard \kServer and \Paging models assume homogenous (interchangeable) servers and cache slots.
% \mareksrevision{
They don't model applications where servers have different capabilities e.g., content-delivery networks such as Akamai with multi-level distributed caches~\cite{li2017collaborative,sundarrajan2017footprint}, or modern cache systems that partition the slots, sometimes dynamically, with differential accessibility by specific processors, cores, processes, threads, or page sets (e.g.,~\cite
{DBLP:journals/dafes/ZangG16,Domnitser_etal_non-monopolizable_caches_2012,Liu_etal_defeating_cache_side_channels_2016,Wang_etal_new_cache_designs_2007,ye+wcl:cache_partitioning,xiang_dcaps_2018}).
% }

This motivates us to generalize the online \kServer problem to allow each request to specify not only a point $p$, but also a subset $S$ of
servers that may serve it.  We call this generalization \emph{Heterogenous} \kServer.  To date, only a few special cases of
this problem have been studied~\cite{jignesh_restricted_servers_2004,castenow+fkmd:server}.
% \mareksrevision{
Significant insights into \kServer and its extensions~\cite {ChrobakL91,ChrobakKPV91,feuerstein_uniform_1998, DBLP:conf/focs/BansalEK17,
DBLP:journals/talg/BansalEKN23, DBLP:conf/isaac/BienkowskiJS19,DBLP:journals/tcs/KoutsoupiasT04, Chiplunkar2015Metrical} have
 % \reporttag{(R2.b)}                                                                            % <------------ report tag
been obtained by first examining simpler versions of the problem.
Following this strategy, we embark on a systematic study of the Heterogenous \kServer problem
when the underlying metric space is uniform or has a star topology.  
% }
For uniform metrics, the problem is equivalent to a variant of \Paging in which each request
specifies a page $p$ and a subset $S$ of $k$ cache slots, to be
satisfied by having a copy of $p$ in some slot in $S$. We call this
problem \emph{\SlotHeteroPaging}.  For star metrics the problem
reduces to a weighted variant where the cost of retrieving a page is
the weight of the page.  For reasons discussed below, we parameterize
these problems by allowing the requestable sets $S$ to be restricted
to an arbitrary but pre-specified family $\slotsetfamily \subseteq
2^{[k]}$.  (Restricting to $\slotsetfamily=\{[k]\}$ gives standard \Paging and \kServer.)

% \mareksrevision{
As common in the study of online optimization problems, we employ competitive analysis.
We use the standard definifion of the competitive ratio of an online algorithm 
(see Section~\ref{sec: preliminaries}). In the discussion below, by the
\emph{optimal competitive ratio} for a given problem we mean the smallest competitive
ratio achievable by an online algorithm.
% }

Next is a summary of our results, followed by a summary of related work.

%%%%%%%%%%%%%

\myparagraph{\SlotHeteroPaging (Section~\ref{sec: slot hetero paging}).}
As we point out, \SlotHeteroPaging easily reduces (preserving the competitive ratio)
to the \GenerServer problem in uniform metrics, for which upper bounds of $k2^k$ and $O(k^2\log k)$
 % \reporttag{(R2.c)}                                                                            % <------------ report tag
on the deterministic and randomized ratios are known~\cite{DBLP:journals/talg/BansalEKN23,DBLP:conf/isaac/BienkowskiJS19}. 
\begin{itemize}
\item
\ignore{Our Theorems~\ref{thm: slot hetero lower bound}\,(i)
  and~\ref{thm: lower bound power set randomized} show that the
  optimal deterministic and randomized competitive ratios for
  \SlotHeteroPaging are at least $\Omega(2^k/\sqrt{k})$ and
  $\Omega(k)$.}  We show that the optimal deterministic and randomized
competitive ratios for \SlotHeteroPaging are at least
$\Omega(2^k/\sqrt{k})$ and $\Omega(k)$, respectively
(Theorems~\ref{thm: slot hetero lower bound}\,(i) and~\ref{thm: lower bound power set randomized}).
\end{itemize}
Hence, the optimal ratios for \SlotHeteroPaging are exponentially
worse than for standard \Paging.  The proofs of Theorems~\ref{thm:
  slot hetero lower bound} and~\ref{thm: lower bound power set
  randomized} employ some novel ideas that may be useful for other
problems: the lower bound in Theorems~\ref{thm: slot hetero lower
  bound}\,(i) uses an adversary argument that requires the
construction of a set family not yet studied in the literature, while the proof of
Theorem~\ref{thm: lower bound power set randomized} is carried out via
a reduction \emph{from} standard \Paging with a cache of size $\exp(\Theta(k))$.

The large competitive ratios in these lower bounds occur only for
instances that use exponentially many distinct request sets
$\slotset$.  In realistic settings only certain subsets of cache
slots or servers can appear in requests, namely those that represent
capabilities or functionalities relevant in a given setting.  This
motivates us to study the optimal ratios as a function of the cache
size $k$ and the family $\slotsetfamily$ of requestable slot sets, and
to try to identify natural families that admit more reasonable ratios.
\begin{itemize}
  \item
\ignore{Theorem~\ref{thm: subset server upper bound} shows that the optimal deterministic ratio is at most $k^2 |\slotsetfamily|$ 
for any family $\slotsetfamily$. Theorem~\ref{thm: slot hetero lower bound}\,(ii) shows a complementary lower bound:
for infinitely many families $\slotsetfamily$, every deterministic online algorithm has competitive ratio $\Omega(\sqrt{|\slotsetfamily|})$.}
We show that the optimal deterministic ratio is at most $k^2 |\slotsetfamily|$ 
 % \reporttag{(R2.b)}                                                                            % <------------ report tag
for any family $\slotsetfamily$ (Theorem~\ref{thm: subset server upper bound}). 
% \mareksrevision{
 Theorem~\ref{thm: slot hetero lower bound}\,(ii) shows a complementary lower bound:
for infinitely many values of $k$ there is a family $\slotsetfamily \subseteq 2^{[k]}$,
 for which every deterministic online algorithm has competitive ratio $\Omega(\sqrt{|\slotsetfamily|})$.
% }
\end{itemize}
Together Theorems~\ref{thm: subset server upper bound} and~\ref{thm: slot hetero lower bound}\,(ii) imply that,
 % \reporttag{(R1.1)}                                                                            % <------------ report tag
as a function of $|\slotsetfamily|$ and $k$, the optimal deterministic ratio for \SlotHeteroPaging is polynomial.

%%%%%%%%%%%%%%%%

\myparagraph{\PageLaminarPaging (Section~\ref{sec: page laminar paging}).} 
We take a brief detour to consider \emph{\PageSubsetPaging}, a natural generalization of \Paging
in which each request is a set $\pageset$ of \emph{pages} from an arbitrary
but fixed family $\pagesetfamily$, and is satisfiable 
by fetching any page from $\pageset$ into any slot in the cache.
We focus on its special case of \emph{\PageLaminarPaging}, where this set family $\pagesetfamily$ is laminar.
\begin{itemize}
\item
  We show that the optimal deterministic and randomized ratios for \PageLaminarPaging
  are at most $hk$ and $h H_k$, where $h$ is the height of the laminar family and $H_k=\sum_{i=1}^k 1/i = \ln k + O(1)$ (Theorem~\ref{thm: page laminar}).
  \end{itemize}
The proof is by a reduction that replaces each set request
$\pageset$ by a request to one carefully chosen page in $\pageset$,
yielding an instance of \Paging, while increasing the optimal cost by
at most a factor of $h$.

%%%%%%%%%%%%%%%%

\myparagraph{\SlotLaminarPaging (Section~\ref{sec: slot laminar paging}).} 
We then return to \SlotHeteroPaging,
now considering the specific structure of $\slotsetfamily$, showing better bounds when $\slotsetfamily$ is laminar.
% \mareksrevision{
This case, which we call \defn{\SlotLaminarPaging}, is intended to model scenarios where server capabilities are hierarchical.
% And
For example, 
some proposed cache partitioning systems (see~\cite{DBLP:journals/dafes/ZangG16} and  
the references therein) divide the cache into parts, some exclusive to certain processes
and other fully or partially shared. 
Such a % Their
cache partitioning strategy can be modeled as a simple laminar structure.
% }

Laminarity implies that $|\slotsetfamily|<2k$, so (per Theorem~\ref{thm: subset server upper bound} above)
the optimal deterministic ratio is $O(k^3)$.
\begin{itemize}
\item We show that the optimal deterministic and randomized ratios for
\SlotLaminarPaging are $O(h^2k)$ and $O(h^2\log k)$, where $h\le k$ is
the height of $\slotsetfamily$ (Theorem~\ref{thm: slot laminar}).  We
next tighten the deterministic bound to $O(hk)$ (Theorem~\ref{thm:
  laminar upper bound}).
\end{itemize}
The proof of Theorem~\ref{thm: slot laminar} is via a reduction to
\PageLaminarPaging (discussed above),
while the proof of Theorem~\ref{thm: laminar upper
  bound} refines the generic algorithm from Theorem~\ref{thm: subset
  server upper bound}.  The dependence on $k$ in these bounds is
asymptotically tight, as \SlotLaminarPaging generalizes standard
\Paging.

%%%%%%%%%%%%%%%%%%%%%%%%

\myitparagraph{Reducing \SlotLaminarPaging to \PageLaminarPaging.}
The reduction of \SlotLaminarPaging to \PageLaminarPaging in Theorem~\ref{thm: slot laminar} is achieved
via a relaxation of \SlotLaminarPaging that drops the constraint that each slot holds at most one page,
while still enforcing the cache-capacity constraint of $k$.
This relaxed instance is naturally equivalent to an instance of \PageLaminarPaging.
The proof then shows how any solution for the relaxed instance
can be ``rounded''  back to a solution for the original \SlotLaminarPaging instance,
losing an $O(h)$ factor in the cost and competitive ratio.

%%%%%%%%%%%%%%%%%%%%%%%%

\myparagraph{\AllOrOnePaging (Section~\ref{sec: all or one paging}).}
\defn{\AllOrOnePaging} is the restriction of \SlotLaminarPaging (with
height $h=2$) to $\slotsetfamily=\{[k]\}\cup\{\{j\}\}_{j\in[k]}$. That
is, only two types of requests are allowed: \emph{general requests}
(allowing the requested page to be anywhere in the cache), and
\emph{specific requests} (requiring the page to be in a specified
slot).  Specific requests don't give the algorithm any choice, so may
appear easy to handle, but in fact make the problem substantially
harder than standard \Paging.  Recent independent
work on \AllOrOnePaging~\cite{castenow+fkmd:server} has shown that the optimal
deterministic ratio is twice that of \Paging, to within an additive constant.
\begin{itemize}
\item
  We show that the optimal randomized ratio of \AllOrOnePaging is also
  at least twice that for \Paging (Theorem~\ref{thm: all or one deterministic 2k-1 lower bound}), 
  while Theorem~\ref{thm: slot laminar} upper bounds the optimal randomized ratio to within a
  constant factor of that for \Paging.  We also show that the offline
  problem is $\NP$-hard (Theorem~\ref{thm: all or one np-hard}), in
  sharp contrast to even \kServer, which can be solved in polynomial time for arbitrary metrics.
\end{itemize}
%

%%%%%%%%%%%%%%%%%%%%%%%%

\myparagraph{\WeightedAllOrOnePaging (Section~\ref{sec: weighted all or one paging}).} 
% \mareksrevision{
To gain insight into \HeteroServer in non-uniform metrics we investigate
 % \reporttag{(R2.b)}                                                                            % <------------ report tag
\defn{\WeightedAllOrOnePaging}, which extends \AllOrOnePaging so that
each page has a non-negative weight and the cost of each retrieval is the weight of the page instead of 1.
% }
%
\begin{itemize}
  \item
% \mareksrevision{
We show that the optimal deterministic ratio for \WeightedAllOrOnePaging is $O(k)$, matching the ratio for standard
\WeightedPaging up to a constant factor (Theorem~\ref{thm: wtd alg}).
% }
\end{itemize}
The algorithm in the proof is implicitly a linear-programming
primal-dual algorithm. With this approach the crucial obstacle to
overcome is that the standard linear program for standard
\WeightedPaging does not force pages into specific slots.  Indeed,
doing so makes the natural integer linear program an \NP-hard
multicommodity-flow problem.  (Section~\ref{sec: weighted all or one
  paging} has an example that illustrates the challenge.)  We show how
to augment the linear program to partially model the slot constraints.

\begin{table}
    \centering\newcommand{\SETSIZE}{\footnotesize}
    \newcommand{\ditto}[1]{\hspace*{#1em}''}
    \newcommand{\nest}[1]{#1}
    \renewcommand{\arraystretch}{1.1}
    {\SETSIZE
        \begin{tabular}{|@{\,} l @{\,\,} l @{\,} l @{\!\!} l @{\,} l @{\,}|}
            \multicolumn{1}{l}{problem} & set family $\slotsetfamily$ (or $\pagesetfamily$) & deterministic & randomized & \multicolumn{1}{l}{where}\\
            \hline\SETSIZE
            \SlotHeteroPaging & $2^{[k]}\setminus\{\emptyset\}$ & ${}\le k 2^k$ & ${}\le O(k^2\log k)$
            & via~\cite{DBLP:journals/talg/BansalEKN23, DBLP:conf/isaac/BienkowskiJS19} % generkserver
            \\
            \nest{\ditto 3} & arbitrary $\slotsetfamily$
            & ${}\le k \min(|\slotsetfamily^\ast|, \sumofsizes(\slotsetfamily))$ &
            & Thm.~\ref{thm: subset server upper bound} % 3.1 arb S det upper bound k min(S*, mass(S))
            \\[1ex]
            \nest{\OneOutOfmPaging{m}, $m\approx k/2$} & $\binom {[k]} m$
            & ${}\ge \Omega(2^k/\sqrt k)$ & ${}\ge \Omega(k)$
            & Thms.~\ref{thm: slot hetero lower bound}(i),~\ref{thm: lower bound power set randomized}
            \\[1ex]
            \nest{\OneOutOfmPaging{m}, any $m$}
            & $\binom {[k]} m$ & ${}\gtrsim \Omega((4k/m)^{m/2}/\sqrt m)$ &
            & Thm.~\ref{thm: slot hetero lower bound}(ii) % 3.2  one-of-m det lower bound (4k/m)^{m/2}/\sqrt m
            \\[1ex]
            \nest{\SlotLaminarPaging} & laminar $\slotsetfamily$, height $h$ & ${} \le (2h-1)k$
            & ${}\le 3 h^2 H_k$
            & Thms.~\ref{thm: slot laminar},~\ref{thm: laminar upper bound}
            \\[1ex]
            \nest{\AllOrOnePaging} & \footnotesize $\{[k]\}\cup\{\{s\} : s\in [k]\}$
            & ${}\ge 2k-1$ & ${}\ge 2H_k-1$
            & ~\cite{castenow+fkmd:server,haney:thesis},
            Thm.~\ref{thm: all or one deterministic 2k-1 lower bound}
            \\
            \nest{\ditto{3}} & \ditto{3}
            & ${}\le 2k+14$ &
            & ~\cite{castenow+fkmd:server}
            \\[0.5ex]\hline
            \rule{0pt}{2.5ex}%
            \SETSIZE
            \WeightedAllOrOnePaging & \footnotesize $\{[k]\}\cup\{\{s\} : s\in [k]\}$ & ${}\le O(k)$ &
            & Thm.~\ref{thm: wtd alg} % 7.1 wtd all-or-one det O(k)
            \\[0.6ex] \hline
            \multicolumn{1}{|@{~}r}{\SETSIZE\PageSubsetPaging restr.~to\!\!} & $\pagesetfamily=\binom{\textrm{all pages}} m$ &
            ${}\ge \binom {k+m} k - 1$ && \cite{feuerstein_uniform_1998}
            \\
            \multicolumn{2}{|l}{\ditto 3} &
            ${}\le k(\binom {k+m} k - 1)$ & ${}\le O(k^3\log m)$& \cite{Chiplunkar2015Metrical}
            \\[1ex]
            \nest{\PageLaminarPaging} & $\pagesetfamily$ laminar, height $h$ & ${}\le hk$ & ${}\le h H_k$
            & Thm.~\ref{thm: page laminar}  % 4.1 page-laminar upper bounds: hk, h H_k, h (offline)

            \\ \hline
        \end{tabular}
    }
    \caption{Summary of upper ($\le$) and lower ($\ge$) bounds on optimal competitive ratios.
    Here $\sumofsizes(\slotsetfamily)=\sum_{\slotset\in\slotsetfamily} |\slotset|$
        and $\slotsetfamily^\ast = \bigcup_{\slotset\in\slotsetfamily} 2^\slotset$.
        % \mareksrevision{
        By $\binom{X}{m}$ we denote the family of all $m$-element subsets of a set $X$.
        % }
        The lower bound for \OneOutOfmPaging{m}
        holds for some but not all $m$ and $k$---see Theorem~\ref{thm: slot hetero lower bound}(ii).
        The upper bound for {\SlotLaminarPaging} in the deterministic case (Theorem~\ref{thm: slot laminar})
        is in fact $2\mytimes\sumofsizes(\slotsetfamily)-k$, which is at most $(2h-1)k$.
        Also, offline \AllOrOnePaging and its generalizations are \NP-hard
%~\cite{chrobak+hlprsy:heterogeneous},
        (Theorem~\ref{thm: all or one np-hard}),
        as is offline \PageSubsetPaging~(\cite{Chiplunkar2015Metrical}).}\label{table: results}
\end{table}

% -------------------------------- END REPLACEMENT FOR TABLE 1

%%%%%%%%%%%%%%%%%%%%%%%%%%%%%%

\myparagraph{Related work.}
\Paging and \kServer have played a central role in the theory of online computation since their introduction in the 1980s~\cite{DBLP:journals/cacm/SleatorT85,DBLP:journals/jal/ManasseMS90,DBLP:books/daglib/0097013}.
For \kServer, the optimal deterministic ratio is between $k$ and $2k-1$~\cite{DBLP:journals/jacm/KoutsoupiasP95}.
 % \reporttag{(R2.c)}                                                                            % <------------ report tag
Recent work~\cite{coester_koutsoupias_k-servers_2021} offers hope for closing this gap,
and substantial progress towards resolving the randomized case has been reported 
in~\cite{bansal_etal_polylogarithmic_2015,DBLP:conf/stoc/BubeckCLLM18,DBLP:conf/stoc/BubeckCR23}.
For Weighted \Paging the optimal ratios are $k$ (deterministic) and $\Theta(\log k)$ (randomized)~\cite
{DBLP:journals/cacm/SleatorT85,
  fiat_etal_competitive_paging_1991,
  DBLP:journals/algorithmica/McGeochS91,
  DBLP:journals/tcs/AchlioptasCN00,
  DBLP:journals/jacm/BansalBN12,Lee18,leeErratum}.

\emph{Restricted Caching} is one previously studied model with \emph{heterogenous} cache slots.
It is the restriction of \SlotHeteroPaging in which each page $p$ has one \emph{fixed} set $S_p\subseteq [k]$ of slots,
and each request to $p$ requires $p$ to be in some slot in $S_p$.
For this problem the optimal randomized ratio is $O(\log^2 k)$~\cite{buchbinder_matroid_caching_2014}.
Better bounds are possible given further restrictions on the sets,
as in \emph{Companion Caching}, which models a cache partitioned into set-associative and fully associative parts~\cite
{brehob_restricted_caching_2003,
  mendel_seiden_companion_caching_2004,
  brehob_np-hard_nonstandard_caches_2004}.
It is natural to ask whether \emph{Restricted \kServer}---%
the restriction of \HeteroServer that requires each point $p$ to be served by a server in a \emph{fixed} set $S_p$---%
is easier than \HeteroServer.  While the two problems are different for many metric classes, they can be shown to be equivalent in metric spaces with no isolated points, such as Euclidean spaces.
The $\NP$-hardness result for Restricted Caching from~\cite{brehob_np-hard_nonstandard_caches_2004}
implies that offline \SlotHeteroPaging
with $\slotsetfamily = \braced{\braced{s, k}\suchthat s\in [k-1]}$ is \NP-hard.

Other sophisticated online caching models include \emph{Snoopy Caching},
in which multiple processors each have their own cache and coordinate to maintain consistency across writes~\cite
{karlin_snoopy_caching_1988}, \emph{Multi-Level Caching},
where the cost to access a slot depends on the slot~\cite{DBLP:journals/jal/ChrobakN00}, and \emph{Writeback-Aware Caching},
where each page has multiple copies, each with a distinct level and weight,
and each request specifies a page and a level, and can be satisfied by
fetching a copy of this page at the given or a higher level~\cite{DBLP:conf/apocs/BeckmannGHM20,DBLP:conf/spaa/BansalNT21}.
(This is a special case of weighted \PageLaminarPaging where $\pagesetfamily$ consists of pairwise-disjoint chains.)
\emph{Multi-Core Caching} models the fact that faults can change the request sequence (e.g.~\cite{kamali_multicore_paging_2020}).

Patel's master thesis~\cite{jignesh_restricted_servers_2004} studies \HeteroServer with just two types of requests---general requests (to
be served by any server) and specialized requests (to be served by any server in a fixed subset $S'$ of ``specialized'' servers)---
and bounds the optimal ratios for uniform metrics and the line.  Recent independent work on deterministic algorithms for online
\AllOrOnePaging establishes a $2k-1$ lower bound and a $2k + 14$ upper bound~\cite{castenow+fkmd:server}.  Earlier
work in~\cite{haney:thesis} presents a $2k-1$ lower bound and a $3k$ upper bound on deterministic algorithms.
% \ignore{We note that an
%   earlier version of our results on online \AllOrOnePaging, including  the $2k-1$ lower bound and a $3k$ upper bound on deterministic
%   algorithms, appeared in~\cite{haney:thesis}.}

\HeteroServer reduces (see Section~\ref{sec: slot hetero paging}) to the \defn{\GenerServer} problem,
in which each server moves in its own metric space, each request specifies one point in each space,
and the request is satisfied by moving any one server to the requested point in its space~\cite{DBLP:journals/tcs/KoutsoupiasT04}.
For uniform metrics, the optimal competitive ratios for this problem are between $2^k$ and $k2^k$ (deterministic)
and between $\Omega(k)$ and $O(k^2 \log k)$ (randomized)~\cite{DBLP:journals/talg/BansalEKN23,DBLP:conf/isaac/BienkowskiJS19}.
These ratios are exponentially worse than the ratios for standard \kServer.
\HeteroServer, parameterized by $\slotsetfamily$, provides a spectrum of problems that bridges the two extremes.

\emph{Weighted} \kServer is a restriction of \GenerServer
in which servers move in the same metric space but have different weights,
 % \reporttag{(R2.3)}                                                                            % <------------ report tag
and the cost is the weighted distance~\cite{DBLP:journals/tcs/FiatR94}.
% \mareksrevision{
For this problem the deterministic and randomized ratios are at least (respectively) doubly exponential~\cite{DBLP:conf/focs/BansalEK17,DBLP:journals/talg/BansalEKN23}
and exponential~\cite{chiplunkar_etal_generalized_server_2020,DBLP:conf/esa/AyyadevaraC21}, even in uniform metrics.
% }

For \PageSubsetPaging restricted to $m$-element sets of pages,
the optimal ratios are between $\binom{k+m} k -1$ and $k(\binom{k+m} k - 1)$
(deterministic)
and between $\Omega(\log km)$ and $O(k^3 \log m)$ (randomized)~\cite{feuerstein_uniform_1998,Chiplunkar2015Metrical}.
This problem has been studied as 
uniform \emph{Metrical Service Systems with Multiple Servers (MSSMS)}.
MSSMS is the generalization of \kServer
where each request is a \emph{set} of points, one of which needs to be
covered by some server.

The \defn{$k$-Chasing} problem extends \kServer by having each request $\pageset$ be a convex subset of $\R^d$,
to be satisfied by moving any server to any point in $\pageset$~\cite{bubeck2021online}. 
For $k$-Chasing, no online algorithm is competitive even for $d= k=2$~\cite{bubeck2021online},
while for $k=1$ the ratios grow with $d$~\cite{argue_chasing_2021,sellke_chasing_2020}. 

In the \emph{$k$-Taxi problem} each request $(p,q)$ requires any server to move to $p$ then (for free) to $q$.
For this problem the optimal ratios are exponentially worse than for standard \kServer~\cite{DBLP:conf/stoc/CoesterK19,DBLP:journals/mp/BuchbinderCN23}.

%%% Local Variables:
%%% mode: latex
%%% TeX-master: "0_0_main__short"
%%% End:

%% file: 2_preliminaries.tex
%%%%%%%%%%%%%%%%

% \mareksrevision{
Most of the set-theoretic notation and terminology used in our paper is standard.
By $\reals$ we denote the set of real numbers.
For a non-negative integer $j$, we use notation $[j]$ for the set of first $j$ positive
integers: $[j] = \braced{1,2,\ldots,j}$.
If $X$ is any set then  $2^X$ denotes the power set of $X$, that is
$2^X = \braced{Y: Y\subseteq X}$, and (adapting the notation for binomial coefficients) we
use notation $\binom{X}{m}$ for the family of all $m$-element subsets of $X$.
% }

\newcommand{\mydefn}[1]{\par\smallskip\noindent\defn{\textbf{#1}}.}

% \mareksrevision{
\mydefn{Heterogenous \kServer}
In this natural generalization of the well-known \kServer problem,
each request 
% may be served only by one of specified servers.  
specifies a subset of servers that may serve the request.
As in \kServer, we are given $k$ servers, numbered $1,2,\ldots,k$, that reside in a metric space $\calM$.
We are also given a family $\slotsetfamily \subseteq 2^{[k]} \setminus \braced{\emptyset}$ of requestable
sets of servers. The objective is to serve a given
 request sequence $\braced{\rho_t}_{t= 1}^T$,
where each request is specified as a pair $\rho_t = \slotrequestpair{r_t}{\slotset_t}$,
for some point $r_t\in \calM$ and set  $\slotset_t \in \slotsetfamily$.
To serve the request $\rho_t = \slotrequestpair{r_t}{\slotset_t}$,
one of the servers in $\slotset_t$ must be moved to $r_t$.
The cost of this move is the distance from its previous location to $r_t$.
The overall cost of serving the request sequence $\braced{\rho_t}_{t= 1}^T$
is the total movement cost of all servers.
% }

\mydefn{\SlotHeteroPaging}
A problem instance consists of a set $[k] = \braced{1,2,\ldots,k}$ of cache slots,
a family $\slotsetfamily \subseteq 2^{[k]} \setminus \braced{\emptyset}$ of requestable slot sets,
and a request sequence $\slotrequestsequence = \braced{\slotrequest_t}_{t= 1}^T$,
where each request has the form $\slotrequest_t = \slotrequestpair{p_t}{\slotset_t}$
for some \emph{page} $p_t$ and set $\slotset_t \in \slotsetfamily$.
A \emph{cache configuration} $C$ is a function that specifies the content of each slot $s\in [k]$;
this content is either a single page (said to be assigned to the slot) or empty. 
Configuration $C$ is said to \emph{satisfy} a request $\slotrequestpair{p}{S}$ if it assigns page $p$
 % \reporttag{(R2.4)}                                                                            % <------------ report tag  
to at least one slot in $S$. A \emph{solution} for a given request sequence $\slotrequestsequence$
is a sequence $\braced{C_t}_{t= 1}^T$ of cache configurations such that, for each $t\in [T]$,
$C_t$ satisfies request $\slotrequest_t$.  The objective is to minimize the number of 
\emph{retrievals},  where a page $p$ is retrieved in slot $s$ at time $t$ 
if $C_t$ assigns $p$ to $s$, but $C_{t-1}$ does not (or $t=1$).
An event when $C_{t-1}$ does not assign $p_t$ to any slot in $\slotset_t$ is called a \emph{fault}.
Obviously a fault triggers a retrieval but, while this is not strictly necessary, it is
convenient to also allow an algorithm to make spontaneous retrievals,
either by fetching into the cache a non-requested page or by moving pages within the cache.

\mydefn{\SlotLaminarPaging}
This is the restriction of \SlotHeteroPaging to instances where $\slotsetfamily$ is \emph{laminar}:
every pair $R, R'\in \slotsetfamily$ of sets is either disjoint or nested.  (This implies $|\slotsetfamily| \le 2k$.) 
% \mareksrevision{
A laminar family $\slotsetfamily$ can be naturally represented by a forest (a collection of disjoint rooted trees),
with a set $R$ being a descendant of $R'$ if $R\subseteq R'$. 
When discussing
\SlotLaminarPaging we routinely use tree-related terminology. For example, we refer to
some sets in $\slotsetfamily$ as leaves, roots, or internal nodes, and we also use other
common relations between the nodes of a rooted tree: of being a child, parent, or ancestor.
% }
The height $h$ of a laminar family $\slotsetfamily$
is one more than the maximum height of a tree in $\slotsetfamily$, 
that is the maximum $h$ for which $\slotsetfamily$
contains a sequence of $h$ strictly nested sets: $R_1 \subsetneq R_2 \subsetneq \ldots \subsetneq R_h$.

\mydefn{\AllOrOnePaging}
This is the restriction of \SlotLaminarPaging to instances with $\slotsetfamily = \braced{[k]} \cup \braced{\braced{j}}_{j\in [k]}$.
That is, there are two types of requests:
\emph{general}, of the form $\slotrequestpair p {[k]}$, requiring page $p$ to be in at least one slot of the cache, 
and \emph{specific}, of the form $\slotrequestpair p {\{j\}}$, $j\in [k]$, requiring page $p$
to be in slot $j$. For convenience,  $\slotrequestpair p \ast$ is a synonym for $\slotrequestpair p {[k]}$,
while $\slotrequestpair p j$ is a synonym for $\slotrequestpair p {\{j\}}$.

\mydefn{\WeightedAllOrOnePaging}
This is the natural extension of \AllOrOnePaging in which each page $p$ is assigned a non-negative weight $\wt p$,
and the cost of retrieving $p$ is $\wt p$ instead of 1.

\mydefn{\OneOutOfmPaging{m}}
This is the restriction of \SlotHeteroPaging to instances
with $\slotsetfamily = \binom {[k]}{m} = \{\slotset \subseteq [k] : |\slotset|=m\}$,
 % \reporttag{(R2.5)}                                                                            % <------------ report tag
that is, every request specifies a set of $m$ slots.

\mydefn{\PageSubsetPaging}
An instance consists of $k$ cache slots,  a collection $\pagesetfamily$ of requestable sets of pages,
and a request sequence $\pagerequestsequence = \braced{\pagerequest_t}_{t=1}^T$, where each $\pagerequest_t$ is drawn from $\pagesetfamily$.
A solution is a sequence  $\braced{C_t}_{t= 1}^T$ of cache configurations (as previously defined)
such that, at each time $t\in[T]$, $C_t$ assigns at least one page in $\pageset_t$ to at least one slot.
The objective is to minimize the number of retrievals. (Slots are interchangeable here, so a cache configuration could be defined
as a multiset of at most $k$ pages, but using slot assignments is technically more convenient.)

\mydefn{\PageLaminarPaging}
This is the restriction of \PageSubsetPaging to instances where $\pagesetfamily$ is \emph{laminar}.

\mydefn{\GenerServer} In this variant of \kServer, each server moves
in its own metric space; each request specifies one point in each
space, and the request is satisfied by moving any one server to the
requested point in its space~\cite{DBLP:journals/tcs/KoutsoupiasT04}.
%%%%%%%%%%%%%%%% 

\myparagraph{Approximation algorithms}.
An algorithm $\algA$ for a given cost minimization problem  is called a \emph{$c$-approximation algorithm}
if, for each instance $\slotrequestsequence$, $\algA$ satisfies
$\cost_{\algA}(\slotrequestsequence) \;\le\; c \cdot \opt(\slotrequestsequence) + b$, where $\cost_{\algA}(\slotrequestsequence)$ is the cost of $\algA$ on $\slotrequestsequence$, 
$\opt(\slotrequestsequence)$ is the optimum cost of $\slotrequestsequence$, and $b$ is a constant independent of $\slotrequestsequence$.  
We follow the standard convention that when we are considering $\algA$ as an offline algorithm,
the constant $b$ must be $0$.

%%%%%%%%%%%%%%%% 

\myparagraph{Online algorithms and competitive ratio}.
 % \reporttag{(R1.2)}                                                              % <------------ report tag
% \mareksrevision{                                                              
In the online variants of the paging and server problems 
the requests arrive online, one per time step, and an online algorithm must satisfy each request before
subsequent requests are revealed.
% }
An online algorithm $\algA$ is called \emph{$c$-competitive} if $\algA$ is a $c$-approximation algorithm.
As common in the literature, we use the term
``optimal deterministic (resp.~randomized) competitive ratio'' to refer to the
optimal ratio of of a deterministic (resp.~randomized) online algorithm for the given problem.

% \mareksrevision{
For {\SlotHeteroPaging} (or {\PageSubsetPaging}) problems that we study,
we assume that the algorithm knows in advance the underlying set family $\slotsetfamily$ (or $\pagesetfamily$).
 % \reporttag{(R2.a)}                                                              % <------------ report tag
This assumption is natural, given that each set family defines a different computational problem, and it's
needed for algorithms customized to a specific family (for example, for {\AllOrOnePaging}).
On the other hand, the generic Algorithm~{\algExhSearch} for \SlotHeteroPaging in Section~\ref{sec: slot hetero upper bound}
does not use any information about the family $\slotsetfamily$.
Its adaptation to \SlotLaminarPaging in Section~\ref{sec: laminar k-server}, called \algRefSearch, is 
presented under the assumption that the laminar family $\slotsetfamily$ is known. With some care
this assumption can be removed, although at the cost of introducing distracting complications in the proof.
% }

%%% Local Variables:
%%% mode: latex
%%% TeX-master: "0_0_main__short"
%%% End:

%% file: 3_0_slot_hetero_paging.tex
%\section{\SlotHeteroPaging}%
%\label{sec: slot hetero paging}
%\input{3_0_slot_constrained_paging}

%%%%%%%

Any instance of \SlotHeteroPaging can be reduced to an instance of \GenerServer in uniform spaces, as follows.  Represent each cache slot
 % \reporttag{(R2.6)}                                                                            % <------------ report tag
by a server in a uniform metric space whose points are the pages.
% RR: added soem explanation for the reduction.
% \mareksrevision{
Each request $\requestpair{p}{\slotset}$ in the instance of \SlotHeteroPaging
is simulated by a long sequence $\rho_{p,\slotset}$ of requests in the instance of \GenerServer.
(Recall that a request in \GenerServer is given as a vector of points, one for each server.)
In $\rho_{p,\slotset}$, each request specifies point $p$ for each server in $\slotset$. 
For each other server, in $[k]\setminus \slotset$, we choose two arbitrary different
points in its space, and the requests in $\rho_{p,\slotset}$ alternate between these two points.  
If we make $\rho_{p,\slotset}$ sufficiently long, these alternating requests force
any competitive online algorithm serving
$\rho_{p,\slotset}$ to move one of the servers in $\slotset$ to $p$, which specifies the slot in $\slotset$ to use for serving request $\requestpair{p}{\slotset}$.
This reduction works both for deterministic and randomized algorithms.
%(In fact, it also works for arbitrary metrics.)
% }

Composing this reduction with the upper bounds from~\cite{DBLP:journals/talg/BansalEKN23}
yields immediate upper bounds of $O(k2^k)$ and $O(k^3\log k)$ on the
deterministic and randomized ratios for unrestricted \SlotHeteroPaging
(that is, with $\slotsetfamily = 2^{[k]}\setminus\braced{\emptyset}$).

Theorem~\ref{thm: slot hetero lower bound}\,(i) (Section~\ref{sec: deterministic slot hetero lower bounds}) 
and Theorem~\ref{thm: lower bound power set randomized} (Section~\ref{sec: randomized slot hetero lower bounds})
show that these bounds are tight within $\poly(k)$ factors: the optimal ratios are at least $\Omega(2^k/\sqrt{k})$ and $\Omega(k)$, respectively.
But restricting $\slotsetfamily$ allows better ratios: Theorem~\ref{thm: subset server upper bound}
(Section~\ref{sec: slot hetero upper bound}) shows an upper bound of $k^2 |\slotsetfamily|$
on the optimal deterministic ratio for any family $\slotsetfamily$.
For \OneOutOfmPaging{m}, Theorem~\ref{thm: subset server upper bound} and Theorem~\ref{thm: slot hetero lower bound}\,(ii)
imply that the optimal deterministic ratio is $O(k^{m+1})$ and $\Omega((4k/m)^{m/2}/\sqrt m)$.

%% file: 3_1_sh_paging_upper_bounds.tex
%%%%%%%%%%%%%%%%%%%%%%%%%%%%%%%%%%%%%%%%%%%%%%%%%%%%%%%%%%%%%%%%%%%%%%%%%%%%%% 

% \subsection{Upper Bounds for Deterministic  \SlotHeteroPaging}%
% \label{sec: slot hetero upper bound}
% \input{3_1_sc_paging_upper_bounds}

%%%%%%%%%%%%%%%%%%%%%%%%%%%%%%%%%%%%%%%%%%%%%%%%%%%%%%%%%%%%%%%%%%%%%%%%%%%%%% 

% \mareksrevision{
This section gives upper bounds on the optimal deterministic competitive ratios for \SlotHeteroPaging
with any slot-set family $\slotsetfamily$,
 % \reporttag{(R1.3)}                                                  % <------------ report tag
as a function of $\sumofsizes(\slotsetfamily) = \sum_{\slotset\in\slotsetfamily}|\slotset|$
and $|\slotsetfamilyclosure|$,
where $\slotsetfamilyclosure = \bigcup_{\slotset\in \slotsetfamily}2^\slotset$.
These two quantities capture the ``size'' or the ``complexity'' of $\slotsetfamily$, although
in a different way, and their mutual relation depends on the structure of $\slotsetfamily$.
In general we have $\sumofsizes(\slotsetfamily) \le k| \slotsetfamily | \le k |\slotsetfamilyclosure|$.
But $|\slotsetfamilyclosure|$ could be as small as $O(\sumofsizes(\slotsetfamily)/k)$
(say, when $\slotsetfamily = 2^{[k]}\setminus \braced{\emptyset}$)
or it could be exponentially larger (say, when $\slotsetfamily = \braced{[k]}$).
% }

% \mareksrevision{
In the theorem below, the bound using $\sumofsizes(\slotsetfamily)$ follows from an easy counting argument.
 % \reporttag{(R1.4)}                                                    % <------------ report tag
The proof of the second bound, in terms of $|\slotsetfamilyclosure|$,
uses the rank method of~\cite{DBLP:journals/talg/BansalEKN23}.
This method estimates the number of steps in one phase of a natural exhaustive-search
algorithm by the rank of a certain upper-triangular matrix.
Our argument is a natural refinement of this approach---in essence, the original
proof from~\cite{DBLP:journals/talg/BansalEKN23} directly applies to the case when
$\slotsetfamily$ is the power set of $[k]$, and we show how to
customize it to an arbitrary given set family $\slotsetfamily$.
% }

%%%%%%%%%%%%%%% 

\begin{theorem}\label{thm: subset server upper bound}
  Fix any $\slotsetfamily\subseteq 2^{[k]}\setminus\braced{\emptyset}$.
  The competitive ratio of Algorithm~{\algExhSearch} in Figure~\ref{fig: subset server upper bound}
  for \SlotHeteroPaging with requestable sets from $\slotsetfamily$
  is at most $k\cdot \min \braced{|\slotsetfamilyclosure|,\sumofsizes(\slotsetfamily)}$.
\end{theorem}

%%%%%%%%%%%%%%% 
\begin{figure}[t]
  \begin{mdframed}[userdefinedwidth=\linewidth]\hspace*{-0.02\linewidth}
    \begin{minipage}{\linewidth}
      \begin{steps}
      \item[] {\bf input:} \SlotHeteroPaging instance $(k, \slotsetfamily, \slotrequestsequence=(\slotrequest_1, \ldots, \slotrequest_T))$
        \step let the initial cache configuration $C_0$ be arbitrary;
        let $\ell \gets 1$ \algcomment{--- $\ell$ is the start of the current phase}
        \step for each time $t \gets 1,2,\ldots, T$:   
          \begin{steps}
            \step if current configuration $C_{t-1}$ satisfies request $\slotrequest_t$:
            ignore the request (set $C_{t} = C_{t-1}$)
            \step else:
              \begin{steps}
	        \step if any configuration satisfies all requests
                $\slotrequest_{\ell}, \slotrequest_{\ell+1}, \ldots, \slotrequest_{t}$: let $C_t$ be any such configuration
                \step else: let $\ell \gets t$; let $C_t$ be any configuration satisfying $\slotrequest_t$
                \algcomment{--- start the next phase}
              \end{steps}
          \end{steps}
      \end{steps}
    \end{minipage}
  \end{mdframed}
  \caption{Online algorithm~\algExhSearch for \SlotHeteroPaging.}\label{fig: subset server upper bound}
\end{figure}

%%%%%%%%%%% 

The theorem implies that the competitive ratio of \OneOutOfmPaging{m} is polynomial in $k$ when $m$ is constant.

\begin{proof}[Proof of Theorem~\ref{thm: subset server upper bound}.]
  Assume without loss of generality that the algorithm faults in each step $t$,
  that is $C_{t-1}$ does not satisfy $\slotrequest_t = \requestpair{p_t}{\slotset_t}$.
  (Otherwise first remove such requests; this doesn't change the algorithm's cost or increase the optimal cost.)

  We first bound the maximum length of any phase.
  The argument is the same for each phase.
  To ease notation assume the phase is the first (with $\ell =1$).
  Let $L$ be the length of the phase.
  By the initial assumption, the following holds:
  \begin{description}
  \item{(UT)} \emph{For each time $t\in [L]$, configuration $C_{t-1}$ satisfies requests $\slotrequest_1, \slotrequest_2, \ldots, \slotrequest_{t-1}$,
    but not $\slotrequest_t$.}
  \end{description}

The final configuration $C_L$ in the phase satisfies all requests in the phase.  
% \mareksrevision{
In particular, for each given set $\slotset\in\slotsetfamily$, if the phase has a
request $\slotrequestpair{p}{\slotset}$ for some page $p$, then $C_L$ has $p$ in some slot in $\slotset$.
Associate request $\slotrequestpair{p}{\slotset}$ with this slot. Then different requests
involving $\slotset$ will be associated with different slots.
Therefore, \emph{(i) there are at most
  $|\slotset|$ distinct requests in the phase that use any given set
 % \reporttag{(R2.7)}                                                                            % <------------ report tag
  $\slotset\in\slotsetfamily$}.  
 % }
  Property (UT) implies that \emph{(ii)
  every request $\slotrequest_t$ in this phase is distinct} (indeed,
for any $t' < t$, $C_{t-1}$ satisfies $\slotrequest_{t'}$ but not
$\slotrequest_t$).  
% \mareksrevision{
Observations (i) and (ii) imply the  bound
% }
$ L \le \sum_{\slotset\in\slotsetfamily} |\slotset| = \sumofsizes(\slotsetfamily)$.

  (As an aside, the above argument uses only that every request in the phase is distinct, a weaker condition than~(UT).
Given only that property, the above bound on $L$
is tight for every $\slotsetfamily$ in the following sense:
  consider any configuration $C$ that puts a distinct page in each slot $s\in[k]$,
  and a request sequence $\slotrequest$ that requests in any order every pair $\requestpair{p}{\slotset}$
 % \reporttag{(R1.5,6,7)}                                                                            % <------------ report tag
  such that $\slotset\in\slotsetfamily$ and $C$ assigns $p$ to a slot in $\slotset$.
 % \mareksrevision{
  Then $\slotrequest$ is satisfied by the single configuration $C$,
  % }
  while having $\sumofsizes(\slotsetfamily)$ distinct requests.)
  
  \smallskip

  Next we give a second bound on $L$ that is tighter for some families $\slotsetfamily$. 
  Identify each page $p$ with a distinct but arbitrary real number\footnote{%
  % \mareksrevision{
  We use real numbers for convenience. The proof works for elements from any sufficiently large field.
  % }
  }.
  % \mareksrevision{
  For each cache configuration $C_t$,
   let $C_t^i\in\R$ denote (the real number associated with) the page in slot $i$, if any, else $0$.   
   % }
   Define matrix $M \in \R^{L\times L}$ by
 % \reporttag{(R2.8)}                                                                            % <------------ report tag
  % 
  \begin{equation*}
    M_{st} \;=\; \prod_{i\in \slotset_t} (C_{s-1}^i - p_t),
  \end{equation*}
  %
  % \mareksrevision{
  so that $M_{st} = 0$ if and only if $C_{s-1}$ satisfies  $\slotrequest_t = \slotrequestpair{p_t}{\slotset_t}$,
  that is $C_{s-1}^i = p_t$ for some $i\in \slotset_t$.
  % }
  So Property~(UT) implies that $M$ is upper-triangular and non-zero on the diagonal.  
  % \mareksrevision{
  Thus $M$ has rank $L$.
	% }
	
  Expanding the formula for $M_{st}$, we obtain
  \begin{align*}
    M_{st}
    \;=\; \sum_{S\subseteq \slotset_t} \Big(\prod_{i\in S} C_{s-1}^i\Big) \mytimes \Big(\prod_{i\in \slotset_t\setminus S} -p_t\Big)
    \;=\; \sum_{S\subseteq \slotset_t} \Big(\prod_{i\in S} C_{s-1}^i\Big) \mytimes (-p_t)^{|\slotset_t| - |S|}
    \;=\; \sum_{S\in\slotsetfamilyclosure} A_{sS}\mytimes B_{St} \,,
  \end{align*}
  where matrices $A\in \R^{L\times \slotsetfamilyclosure}$  and $B\in \R^{\slotsetfamilyclosure\times L}$ are  defined by
  \begin{equation*}
    A_{sS} \;=\; \prod_{i\in S} C_{s-1}^i
    \quad\textrm{and}\quad
    B_{St} \;=\; \begin{cases}
      (-p_t)^{|\slotset_t| - |S|} & \quad\textrm{if}\; S\subseteq \slotset_t
      \\
      0				& \quad \textrm{otherwise}.
    \end{cases}
  \end{equation*}

  That is, $M=A B$; $A$ and $B$ (and $M$) have rank at most $|\slotsetfamilyclosure|$.
  And $M$ has rank $L$, so $L \le |\slotsetfamilyclosure|$.
  \eat{
    \begin{equation}
    L \;\le\; |\slotsetfamilyclosure|.
    \label{eqn: second bound on L}
  \end{equation}
  }
  To bound the optimum cost, consider any phase other than the last. Let $t'$ and $t''$ be the start and end times.
  Suppose for contradiction that the optimal solution incurs no cost (has no retrievals) during $[t'+1, t''+1]$.
  Then its configuration at time $t'$ satisfies all requests in $[t', t''+1]$,
  contradicting the algorithm's condition for terminating the phase.
  So the optimal solution pays at least 1 per phase (other than the last).
  In any phase of length $L$ the algorithm pays at most $kL$ (at most $k$ per step).
  This and the two upper bounds on $L$
  %~(\ref{eqn: second bound on L}) and~(\ref{eqn: first bound on L})
  imply Theorem~\ref{thm: subset server upper bound}.
\end{proof}

%%% Local Variables:
%%% mode: latex
%%% TeX-master: "0_0_main__short"
%%% End:

%% file: 3_2_deterministic_sh_paging_lower_bounds.tex
%%%%%%%%%%%%%%%%%%%%%%%%%%%%%%%%%%%%%%%%%%%%%%%%%%%%%%%%%%%%%%%%%%%%%%%%%%%%%% 

% \subsection{Lower Bounds for Deterministic \SlotHeteroPaging}%
% \label{sec: deterministic slot hetero lower bounds}
% \input{3_2_deterministic_sc_paging_lower_bounds}

%%%%%%%%%%%%%%%%%%%%%%%%%%%%%%%%%%%%%%%%%%%%%%%%%%%%%%%%%%%%%%%%%%%%%%%%%%%%%% 

We establish our lower bounds for \SlotHeteroPaging and \OneOutOfmPaging{m} given in Table~\ref{table: results}.

\begin{theorem}\label{thm: slot hetero lower bound}~{} 
  (i)\ For all odd $k$, the optimal deterministic ratio for \OneOutOfmPaging m with $m=(k+1)/2$
  is at least $\binom k m = \Omega(2^k/\sqrt{k})$.  For all $k$,
  the optimal ratio  with $m=\lfloor (k+1)/2\rfloor$  is $\Omega(2^k/\sqrt{k})$.
  (ii)\ For any even $m\ge 2$ and any $k > m$ that is an odd multiple
  of $m-1$, the optimal deterministic ratio for \OneOutOfmPaging{m} is at least
  $\binom{m-1}{m/2} \big( \frac{k}{m-1} \big)^{m/2} =
  \Theta((4k/m)^{m/2}/\sqrt m) = \Omega(\sqrt{|\slotsetfamily|})$,  where $\slotsetfamily =\binom{[k]}{m}$.
\end{theorem}

% \mareksrevision{
  Before proving Theorem~\ref{thm: slot hetero lower bound}
  we prove Lemma~\ref{lem: lower bound via eluders}, % .
 % \reporttag{(R1.8)}                                                                            % <------------ report tag
% Thislemma
which establishes a general lower bound on the competitive ratio for \SlotHeteroPaging for some
requestable set families $\slotsetfamily$. Namely, if, for a given $\slotsetfamily$, we can
identify a family $\calZ \subseteq 2^{[k]}$ with certain properties then no online algorithm
can have competitive ratio smaller than $|\calZ|$.  
The proof of Theorem~\ref{thm: slot hetero lower bound} then constructs such families $\calZ$ for appropriate families
$\slotsetfamily$ of requestable sets.
% }

Throughout this section $\complX$ denotes the complement of set $X \subseteq [k]$, that is
$\complX = [k] \setminus X$.

%%%%%%%%%%%%%% 

\begin{lemma}\label{lem: lower bound via eluders}
  For some $\slotsetfamily\subseteq 2^{[k]}$,
  suppose there are two set families $\calG\subseteq\slotsetfamily$ and $\calZ\subseteq 2^{[k]}$   such that
  \vspace{-0.075in}
  \begin{description} \setlength{\itemsep}{-0.04in}
  \item{(gz0)}
    For each $X\subseteq [k]$ there is $\slotset\in\calG$
    such that $\slotset\subseteq X$ or  $\slotset\subseteq \complX$.
    % 
 % \reporttag{(R1.9)}                                                                            % <------------ report tag
  \item{(gz1)} If $Z\in\calZ$ then $\complZ\notin\calZ$.
  \item{(gz2)}
    For each $\slotset\in\calG$ there is  $Y\in\calZ$
    such that $\slotset\not\subseteq Z$  and $\slotset\not\subseteq \complZ$ for all $Z\in \calZ \setminus \braced{Y}$.       
  \end{description}
  \vspace{-0.075in}
  Then the optimal deterministic ratio for \SlotHeteroPaging with family $\slotsetfamily$    is at least $|\calZ|$.
\end{lemma}

\begin{proof}
  The proof is an adversary argument based on the following idea.
  % \mareksrevision{
  We consider a suitably chosen fixed set of $2|\calZ|$ solutions. 
  At each step, the adversary chooses a request, using one of the sets in $\calG$, which forces the algorithm to fault
  but among our chosen solutions only at most two will fault. 
 % \reporttag{(R1.10, R2.9)}                                                                            % <------------ report tag
   At the end, the algorithm's total cost
  is at least $|\calZ|$ times the average cost of these chosen solutions, so its competitive ratio is at least $|\calZ|$.
  % }
  This general approach is common for lower bounds on deterministic online algorithms
  (see e.g.\ lower bounds on the optimal ratios for \kServer~\cite{DBLP:journals/jal/ManasseMS90},
  for Metrical Task Systems~\cite{DBLP:journals/jacm/BorodinLS92}
  and for \GenerServer on uniform metrics~\cite{DBLP:journals/tcs/KoutsoupiasT04}).
  
  Here are the details. Let $\algA$ be any deterministic online algorithm for \SlotHeteroPaging with slot-set family $\slotsetfamily$.
 % \reporttag{(R1.11)}                                                                            % <------------ report tag
  The adversary will request just two pages, $p_0$ and $p_1$. 
  For a set $X\subseteq [k]$, let $Q_X$ denote the cache configuration where
  % \mareksmargincomment{changed $C_X\rightarrow Q_X$, for consistency with Thm 3.3} % <------------------------------ margincomment
  the slots in $X$ contain $p_0$ and the slots in $\complX$ contain $p_1$.
  Without loss of generality assume that each slot of $\algA$'s cache
  always holds $p_0$ or $p_1$---its cache configuration is $Q_X$ for some $X$.

  At each step, if the current configuration of $\algA$ is $Q_X$,
  the adversary chooses $\slotset\in\calG$ such that either $\slotset\subseteq X$ or $\slotset\subseteq \complX$.
  (Such an $\slotset$ exists by Property~(gz0).)  If $\slotset\subseteq X$, then all slots in $\slotset$ hold $p_0$,
  and the adversary requests $\requestpair{p_1}{\slotset}$, causing a fault.
  Otherwise, $\slotset\subseteq \complX$, so all slots in $\slotset$ hold $p_1$.
  In this case the adversary requests $\requestpair{p_0}{\slotset}$, causing a fault.
  The adversary repeats this $K$ times, where $K$ is arbitrarily large.
  Since  $\algA$ faults at each step, the overall cost of $\algA$ is at least $K$.

  It remains to bound the optimal cost.  Let $\complements{\calZ} = \braced{\complZ\suchthat Z\in \calZ}$.
  By~(gz1), we have $\complements{\calZ} \cap \calZ = \emptyset$.
  For each $Z\in\calZ\cup \complements{\calZ}$ define a solution
  called the \emph{$Z$-strategy}, as follows.  The solution starts in configuration $Q_Z$. 
  It stays in $Q_Z$ for the whole computation, except that on requests $\requestpair{p_a}{\slotset}$
  that are not served by $Q_Z$ 
  % \mareksrevision{
  (that is, when in configuration $Q_Z$ all slots of $\slotset$ contain $p_{1-a}$), 
  % }
 % \reporttag{(R1.12)}                                                                            % <------------ report tag
  it retrieves $p_a$ to any slot $j\in \slotset$, serves the request,
  then retrieves $p_{1-a}$ back into slot $j$,  restoring configuration $Q_Z$.  This costs 2.

  We next observe that in each step at most one $Z$-strategy faults (and pays $2$).
  Assume that the request at a given step is to $p_0$  (the case of a request to $p_1$ is symmetric).
  Let this request be $\requestpair{p_0}{\slotset}$, where $\slotset\in\calG$.  Let $Y\subseteq [k]$ be the set from Property~(gz2).
  % \mareksrevision{
 Then, for all $Z\in (\calZ \cup \complements{\calZ})\setminus\smallbraced{Y,\complY}$
  we have $\slotset\cap Z \neq\emptyset$,  
  % }
  implying that configuration $Q_Z$ has a slot in $\slotset$
 % \reporttag{(R1.13)}                                                                            % <------------ report tag
  that contains $p_0$---in other words, configuration $Q_Z$ satisfies $\slotset$. 
  Also, either  $\slotset\cap Y \neq\emptyset$ or $\slotset\cap \complY \neq\emptyset$,
  so one of the configurations $Q_{Y}$ or $Q_{\complY}$ also satisfies $\slotset$.
  Therefore only one $Z$-strategy ($Y$ or $\complY$)  might not satisfy $\slotset$.
  So, in each step, at most one $Z$-strategy faults (and pays $2$).

  Thus the combined total cost for all $Z$-strategies (not counting the cost of at most $k$ for moving to $Z$ at the beginning)    is at most $2K$.
  There are $2|\calZ|$ such strategies, so their average cost is at most $(2K+k)/2|\calZ|$.
  The cost of $\algA$ is at least $K$,    so the ratio is at least 
	  $\frac{K}{(2K+k)/2|\calZ|} = \frac{|\calZ|}{1 + k/2K}$.
  Taking $K$ arbitrarily large, the lemma follows.
\end{proof}

%%%%%%%%%%%% 

\begin{proof}[Proof of Theorem~\ref{thm: slot hetero lower bound}]
  \emph{Part (i).}
  Recall that $m=\lfloor (k+1)/2\rfloor$. First consider the case when $k$ is odd. Apply Lemma~\ref{lem: lower bound via eluders},
  taking both $\calG$ and $\calZ$ to be $\binom{[k]}{m}$.
  Properties~(gz0) and~(gz1) follow directly from $k$ being odd and the definitions of $\calG$ and $\calZ$.
 % \reporttag{(R2.10)}                                                                            % <------------ report tag
  Property~(gz2) also holds with $Y=\slotset$. (For any $\slotset\in\calG$, every $Z\in\calZ$
  satisfies $|Z| = |\slotset| > |\complZ|$, so $\slotset\not\subseteq \complZ$,
  while $\slotset\subseteq Z$ implies $Z=\slotset$.)  Thus, by Lemma~\ref{lem: lower bound via eluders}, the ratio is at least
  $|\calZ| = \binom{k}{(k+1)/2} = \Omega(2^k/\sqrt{k})$.  This proves Part~(i) for odd $k$.
  
  For even $k$, let $k' = k-1$. Then apply Part~(i) to $k'$ using just cache slots in $[k']$,
  that is, using slot-set family $\slotsetfamily' = \binom{[k']} m \subseteq \binom {[k]} m = \slotsetfamily$,
  with slot $k$ playing no role as it is never requested.
  This proves Part (i).
  
  \smallskip
  \emph{Part (ii).}
  Fix such an $m$ and $k$.  Let $\ell = k/(m-1)$ so $\ell\ge 3$ is odd.
  Recall that $\slotsetfamily = \binom{[k]} m$ is the family of requestable slot sets.
  Partition $[k]$ arbitrarily into $m-1$ disjoint subsets $B^1,B^2,\ldots,B^{m-1}$, each of cardinality $\ell$.
  For each $B^e$, order its slots arbitrarily as $B^e = \braced{b^e_1,b^e_2,\ldots,b^e_\ell}$.
  For any index $c\in \braced{1,2,\ldots,\ell}$ and an integer $i$,
  let $c\oplus i$ denote $((c+i-1)\bmod\ell) + 1$.
  In other words, we view each $B^e$ as an odd-length cycle,
  and this cyclic structure is important in the proof.
  Any consecutive pair $\smallbraced{b^e_c,b^e_{c\oplus 1}}$ of slots on this cycle
  is called an \emph{edge}. Thus each cycle $B^e$ has $\ell$ edges.

  First we define $\calG \subseteq \slotsetfamily$ for Lemma~\ref{lem: lower bound via eluders}.
  The sets $\slotset$ in $\calG$ are those obtainable as follows:
  choose any $m/2$ edges, no two from the same cycle,
  then let $S$ contain the $m$ slots in those $m/2$ chosen edges.
  (The six slots inside the three dashed ovals in Figure~\ref{fig: one-out-of-m lower bound}
  show one $\slotset$ in $\calG$.)  This set of $m/2$ edges uniquely determines $\slotset$, and vice versa.

  We verify that $\calG$ has Property~(gz0) from Lemma~\ref{lem: lower bound via eluders}.
  Indeed, consider any $X\subseteq [k]$.
  Call the slots in $X$ \emph{white} and the slots in $\complX$ \emph{black}.
  Each cycle $B^e$ has odd length, so has an edge $\{b^e_c, b^e_{c\oplus 1}\}$
  that is white (with two white slots) or black (with two black slots).
  So either (i) at least half the cycles have a white edge, or (ii) at least half have a black edge.
  Consider the first case (the other is symmetric).
  There are $m-1$ cycles, and $m$ is even, so at least $m/2$ cycles have a white edge.
  So there are $m/2$ white edges with no two in the same cycle.
  The set $\slotset$ comprised of the $m$ white slots from those edges
  is in $\calG$, and is contained in $X$ (because its slots are white).  So $\calG$ has Property~(gz0).
  
  %%%%%%%% 

  \begin{figure}[t]
    \begin{center}
      \includegraphics[width = 6.2in]{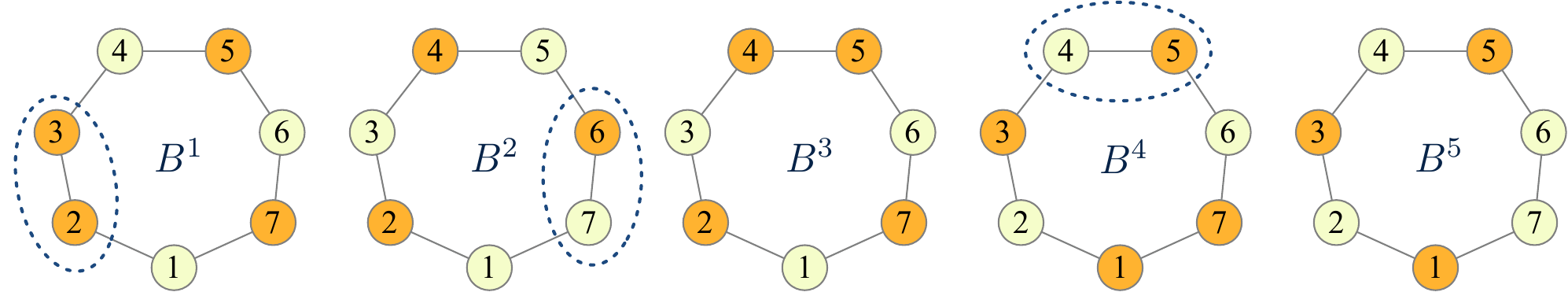}
    \end{center}
    \caption{Illustration of the proof of Theorem~\ref{thm: slot hetero lower bound} Part (ii) for $k = 35$, $m = 6$, and $\ell = 7$.
      The figure shows the partition of all slots into $m-1 = 5$ sets $B^1, \ldots, B^5$, each represented by a cycle.
      To avoid clutter, each slot $b^e_c$ is represented by its index $c$ within $B^e$.
      The picture shows set $\slotset=\{b^1_2, b^1_3, b^2_6, b^2_7, b^4_4, b^4_5\}\in\calG$, marked by dashed ovals.
      It also shows $Z_{\slotset'}\in\calZ$, represented by orange/shaded circles, for $\slotset'=\{b^1_2, b^1_3, b^3_4, b^3_5, b^4_7, b^4_1\}$.
    }\label{fig: one-out-of-m lower bound}
  \end{figure}

  %%%%%%%% 

  Next we define $\calZ \subseteq 2^{[k]}$ for Lemma~\ref{lem: lower bound via eluders}.
  The set $\calZ$ contains,  for each set $\slotset'\in \calG$,  one set $Z_{\slotset'}$, defined as follows.
  For each of the $m/2$ cycles $B^e$ having an edge $\{b^e_c, b^e_{c\oplus 1}\}$ in $\slotset'$,
  add to $Z_{\slotset'}$ the two slots on that edge, together with the $(\ell-3)/2$ slots
  $b^e_{c\oplus 3}, b^e_{c\oplus 5}, \ldots, b^e_{c \oplus (\ell-2)}$.
  For each of the $m/2-1$ remaining cycles $B^e$,  add to $Z_{\slotset'}$ the $(\ell-1)/2$ slots  $b^e_{1}, b^e_{3}, \ldots, b^e_{\ell-2}$.
  (The orange/shaded slots in Figure~\ref{fig: one-out-of-m lower bound} show one set $Z_{\slotset'}$ in $\calZ$.)
  Then $Z_{\slotset'}$ contains exactly $m/2$ edges (the ones in $\slotset'$)
  while its complement $\complZ_{\slotset'}$ contains exactly $m/2-1$ edges (one from each cycle with no edge in $\slotset'$).
  This implies Property~(gz1).  Note that $Z_{\slotset'} \neq Z_{\slotset''}$ for different sets $S',S''\in \calG$.

  Next we show Property~(gz2).  Given any set $\slotset\in \calG$, let $Y = Z_{\slotset} \in \calZ$.
  Consider any $Z_{\slotset'} \in \calZ$
  such that $\slotset \subseteq Z_{\slotset'}$ or $\slotset\subseteq \complZ_{\slotset'}$.
  We need to show $Z_{\slotset'} = Z_{\slotset}$, i.e., $\slotset'=\slotset$.
  It cannot be that $\slotset\subseteq \complZ_{\slotset'}$, because
  $\slotset$ contains $m/2$ edges,  whereas $\complZ_{\slotset'}$ contains $m/2-1$ edges.
  So $\slotset\subseteq Z_{\slotset'}$.  But $\slotset$ and $Z_{\slotset'}$ each contain exactly $m/2$ edges,
  which therefore must be the same.  It follows from the definition of $Z_{\slotset'}$
  that $\slotset'=\slotset$.  So Property~(gz2) holds.

  So $\calG$ and $\calZ$ have Properties~(gz0)-(gz2) from Lemma~\ref{lem: lower bound via eluders}.
 Directly from definition we have \(|\calZ| = |\calG|\),  while \(|\calG|= \binom{m-1}{m/2} \ell^{m/2}\)
  because there are $\binom{m-1}{m/2}$ ways to choose $m/2$ distinct cycles,
  and then for each of these $m/2$ cycles there are $\ell$ ways to choose one edge. 
  Lemma~\ref{lem: lower bound via eluders} and $\ell=k/(m-1)$
  imply that the optimal deterministic ratio is at least $f(m, k) = \binom{m-1}{m/2} (k/(m-1))^{m/2}$.  To complete the proof of part~(ii) we lower-bound $f(m,k)$. We observe that
  \begin{equation}
    4^m
    \;=\;
    \Omega(\sqrt m\, (k/(k-m))^{k-m+1/2}).
    \label{eqn: aux claim for one-of-m paging}
  \end{equation}
  This can be verified by considering two cases:
  If $k\ge m+2$ then, using $1+z\le \ee^z$, we have
  $\sqrt m \, (k/(k-m))^{k-m+1/2}
        = \sqrt m (1+m/(k-m))^{k-m+1/2} 
        \le \sqrt m\cdot \ee^{5m/4}
        \le 2\cdot 4^m$, for all $m\ge 1$.
  In the remaining case, for $k = m+1$, we have  $ \sqrt m (k/(k-m))^{k-m+1/2}
  = \sqrt m (1+m)^{3/2}  \le  2\cdot 4^m$.
  Thus~(\ref{eqn: aux claim for one-of-m paging}) indeed holds.  Now, recalling that $f(m, k) = \binom{m-1}{m/2} (k/(m-1))^{m/2}$, we derive
  \begin{align}
    f(m, k) &{} \;=\;\Theta \big ((2^m/\sqrt m) \cdot (k/(m-1))^{m/2}\big) && (\textit{Stirling's approximation}) 
	\notag\\
    &{} = \;\Theta \big ((4k/m)^{m/2} \cdot (1+1/(m-1))^{m/2}/\sqrt m\big) && (\textit{rewriting})
	\notag\\
    &{} =\; \Theta \big ((4k/m)^{m/2}/\sqrt m\big) && ((1+1/(m-1))^{m/2} \le \ee) &&
	\label{eqn: one-of-m paging first bound}
\end{align}
This gives us one estimate on the competitive ratio in Theorem~\ref{thm: slot hetero lower bound}(ii).
To obtain a second estimate, squaring both sides of~(\ref{eqn: one-of-m paging first bound}), we obtain
\begin{align*}
    f(m, k)^2 &{} \;= \; \Omega \big((4k/m)^m/m\big) = \Omega \big ((k/m)^{m}\cdot 4^m / m\big) &&  
	\notag\\
    	&{} = \;\Omega\big((k/m)^{m}\cdot (k/(k-m))^{k-m+1/2}/\sqrt m\big) && (\textit{using~\eqref{eqn: aux claim for one-of-m paging}})
	 \notag\\
    &\textstyle{} = \;\Omega(\binom k m) = \Omega(|\slotsetfamily|) && (\textit{Stirling's approximation})
  \end{align*}
 Therefore $f(m, k) = \Omega(\sqrt{|\slotsetfamily|})$, as claimed,
 completing the proof of Theorem~\ref{thm: slot hetero lower bound}(ii).
\end{proof}

% \mareksrevision{
  It is worth noting
  % adding
  that, as can be seen from the proofs in this section, both Theorem~\ref{thm: slot hetero lower bound}
and Lemma~\ref{lem: lower bound via eluders} hold even for instances with just two pages.
% }

%%% Local Variables:
%%% mode: latex
%%% TeX-master: "0_0_main__short"
%%% End:

%% file: 3_3_randomized_sh_paging_lower_bounds.tex
% \subsection{Lower Bound for Randomized \SlotHeteroPaging}%
% \label{sec: randomized slot hetero lower bounds}
% \input{3_3_randomized_sc_paging_lower_bounds}

%%%%%%%%%%%%% 

Next we present a lower bound on the optimal competitive ratio for randomized algorithms:
 % \reporttag{(R1.14)}                                                                            % <------------ report tag

\begin{theorem}\label{thm: lower bound power set randomized}
  The optimal randomized ratio for \OneOutOfmPaging{m} with
  % \textcolor{blue}{cache size $k$ and} 
  cache size $k$ and
  $m = \floor{k/2}$ is $\Omega(k)$,
  even for inputs that use only two pages.
  % \textcolor{blue}{even for inputs that use only two pages}. 
\end{theorem}

The proof is by a reduction from standard \Paging with some $N$ pages
and a cache of size $N-1$.  For any $N$, this problem has optimal
randomized competitive ratio $H_{N-1} = \Theta(\log
N)$~\cite{fiat_etal_competitive_paging_1991}. This and the next lemma
imply the theorem.

%%%%%%%%%%% 

\begin{lemma}\label{lemma: lower bound power set randomized}
  Every $f(k)$-competitive (randomized) online algorithm $\algA$ for \OneOutOfmPaging{m}  with $m=\floor{k/2}$
  can be converted into an $O(f(k))$-competitive (randomized) online algorithm $\algB$
  for standard \Paging with $N$ pages and a cache of size $N-1$,  where $N=2^{\Theta(k)}$.
\end{lemma}

\begin{proof}
  Fix a sufficiently large $k$.  Assume without loss of generality that $k$ is even
  (otherwise apply the construction below to slots in $[k-1]$, ignoring slot $k$ as it is never requested).
  Take $N=\floor{\eulere^{k/16}}$.

  To ease exposition,  view the \Paging problem with $N$ pages and a cache of size $N-1$
  as the following equivalent online \emph{Cat and Rat} game
  on any set $\calH$ of $N$ \emph{holes} (see e.g.~\cite[\S 11.3]{borodin_el-yaniv_on_randomization_1999}).
  The input is a sequence $\catratrequestsequence=(R_0,C_1, \ldots, C_T)$ of holes (i.e.,  $R_0\in \calH$ and
  $C_t\in \calH$ for all $t$).  A solution is any sequence $(R_1, \ldots, R_T)$ of holes
  such that $R_t \ne C_t$ for all $t\in [T]$.  Informally, at each time $t\in [T]$, the cat inspects hole $C_t$,
  and if the rat's hole $R_{t-1}$ at time $t-1$ was $C_t$,
  the rat is required to move to some other hole $R_t \in \calH \setminus\{C_t\}$.
  (For the solution to be online, $R_t$ must be independent of $C_{t+1}, C_{t+2}, \ldots, C_T$ for all $t \ge 1$.)
  The goal is to minimize the number of times the rat moves, that is $|\{ t\in[T] : R_t \ne R_{t-1}\}|$. 

The claimed algorithm~$\algB$ for \Paging will work by 
reducing a given instance $\catratrequestsequence$ on a set $\calH$ of $N$ holes
to an instance $\slotrequestsequence$ of \OneOutOfmPaging{(k/2)}, simulating $\algA$ on $\slotrequestsequence$,
and converting its solution to a solution for $\catratrequestsequence$.
This instance $\slotrequestsequence$ uses just two pages, $p_0$ and $p_1$.  
To describe the reduction, we need a few more definitions and observations.

For two disjoint sets $S_0, S_1 \subseteq [k]$, by $\twopageconfig(S_0,S_1)$ we denote 
the cache configuration that assigns $p_0$ to slots in $S_0$ and $p_1$ to slots in $S_1$, with the remaining slots empty.
A configuration $\twopageconfig(S_0,S_1)$ is called \emph{balanced} if $|S_0| = |S_1| = k/2$. (This obviously implies that
$\complS_0 = S_1$, where $\complS_0 = [k]\setminus S_0$.)  
% \mareksrevision{
Note that the Hamming distance between any two balanced configurations (that is, the number of
slots whose contents differ in these configurations) is always even.
% }

The following easy observation will be useful:

%%%%%

\begin{observation}\label{obs: balanced}
Let $S\subseteq [k]$ with $|S| = k/2$. Any request $\requestpair{p_0}{S}$ is satisfied by 
every balanced configuration except $\twopageconfig(\complS,S)$, and any request $\requestpair{p_1}{S}$
is satisfied by every balanced configuration except $\twopageconfig(S,\complS)$.
\end{observation}

For any set $\calC$ of cache configurations,  a \emph{forcing sequence for $\calC$} is a
request sequence such that the cache configurations that satisfy all requests in this sequence without cost are exactly those in $\calC$.

%%%%%

\begin{claim}\label{claim: sh paging lower bound forcing}
% \mareksrevision{
Let $\calC$ be a set of balanced cache configurations such that any two configurations in $\calC$ are at Hamming distance 
strictly greater than $2$.
 % \reporttag{(R1.15)}                                                                            % <------------ report tag
Then there is a request sequence $\forcingseq(\calC)$ that is forcing for $\calC$.
% }
\end{claim}

  To verify the claim, take $\forcingseq(\calC)$ to be the sequence formed by (any ordering of) all those allowed requests
  that are satisfied by all configurations in $\calC$. Specifically, $\forcingseq(\calC)$
  consists of all requests $\requestpair {p_i} S$ such that $i\in\{0,1\}$ and $|S|=k/2$,
  with $S \cap S_i \ne \emptyset$ for all $\twopageconfig(S_0, S_1)\in \calC$. By definition,
  each configuration in $\calC$ satisfies all requests in $\forcingseq(\calC)$.

  It remains to show that for any configuration $\twopageconfig(S'_0, S'_1)\notin \calC$ there is
  a request in $\forcingseq(\calC)$ not satisfied by $\twopageconfig(S'_0, S'_1)$. 
  In the case that $\twopageconfig(S'_0, S'_1)$ is balanced, we can take
  this request to be $\requestpair {p_1} {S'_0}$, because, by Observation~\ref{obs: balanced},
  it is included in $\forcingseq(\calC)$ but it is not satisfied by $\twopageconfig(S'_0, S'_1)$.
  
  Next consider the case that $\twopageconfig(S'_0, S'_1)$ is not balanced.
  % \mareksrevision{
  Without loss of generality, assume that $\twopageconfig(S'_0, S'_1)$ has no
  empty slots, since filling empty slots with $p_0$ or $p_1$ can only increase the number of
  requests in $\forcingseq(\calC)$ satisfied by $\forcingseq(\calC)$.
  Likewise assume by symmetry that $|S'_0| > k/2$.
  % }
  Let $B_0$ and $B'_0$ be any two size-$k/2$ subsets of $S'_0$
  such that $|B_0\setminus B'_0| = |B'_0\setminus B_0| = 1$.
  Then $\twopageconfig(B_0, \complB_0)$ and $\twopageconfig(B'_0, \complB'_0)$ are at Hamming distance 2 so both cannot be in $\calC$.
  Assume without loss of generality that $\twopageconfig(B_0, \complB_0)$ is not in $\calC$.
  Then request $\requestpair {p_1} {B_0}$ is satisfied by every configuration in $\calC$ because, by Observation~\ref{obs: balanced},
  the only balanced configuration that doesn't satisfy this request is $\twopageconfig(B_0, \complB_0)$,
  which is not in $\calC$. So $\requestpair {p_1} {B_0}$ is in $\forcingseq(\calC)$.
  On the other hand, since $B_0 \subseteq S'_0$, request $\requestpair {p_1} {B_0}$ is not satisfied by $\twopageconfig(S'_0, S'_1)$.
  This proves Claim~\ref{claim: sh paging lower bound forcing}.

\smallskip

We now describe how Algorithm~$\algB$ computes its solution $R =(R_1, \ldots, R_T)$ for a given request sequence $\catratrequestsequence$.
To streamline presentation, we present it first as an offline algorithm. At the beginning $\algB$ chooses
any collection $\calC$ of $N$ balanced configurations such that the Hamming distance between every two distinct configurations in $\calC$ is at least $k/16$.
Such a collection $\calC$ can be constructed using a greedy method, as in~\cite{DBLP:conf/isaac/BienkowskiJS19}. (Here we only need existence, which
can be also established by a probabilistic proof:  if one forms $\calC$ by randomly and uniformly sampling $N$ times with replacement from the balanced configurations,
  then by a standard Chernoff bound and the na{\"\i}ve union bound $\calC$ has the required property with positive probability.)
Algorithm~$\algB$ lets $\calH=\calC$, identifying holes with cache configurations.
% \mareksrevision{
Then, for each time $t\in [T]$, it replaces the request $C_t$ in $\catratrequestsequence$
by $\forcingseq(\calC\setminus \braced{C_t})^{k}$, that is, $k$ repetitions of the forcing sequence for $\calC\setminus\braced {C_t}$.
(This sequence exists if $k$ is large enough, by Claim~\ref{claim: sh paging lower bound forcing} and the choice of $\calC$.)
This produces sequence $\slotrequestsequence$, an instance of \OneOutOfmPaging{m}, with $m = k/2$.

Next, $\algB$ simulates $\algA$ on $\slotrequestsequence$. Let $D$ be a solution for $\slotrequestsequence$ produced by $\algA$.
Without loss of generality we can assume that, for each $t\in[T]$,  as $D$ responds to $\forcingseq(\calC\setminus \braced{C_t})^{k}$
it uses at least one configuration $P\in \calC\setminus\braced{C_t}$. (This is because otherwise $D$ would incur cost
at least $k$ on each such $\forcingseq(\calC\setminus \braced{C_t})^{k}$. 
 % \reporttag{(R1.16)}                                                                            % <------------ report tag
We could then replace $D$ by a solution $D'$ that at the end of each such forcing sequence
moves into any configuration in $\calC\setminus \braced{C_t}$. Being in a different configuration could increase the future cost of $D'$
by at most $k$, so the total cost of $D'$ will be at most twice that of $D$.)
% Then $\algB$ lets $R_t = P$.
Finally $\algB$ takes $R_t$ to be $P$.
The produced sequence $R=(R_1, \ldots, R_T)$ is a valid solution to $\catratrequestsequence$,  because $R_t \in \calC\setminus\braced{C_t}$ for $t\in[T]$.  
% }

To complete the description of $\algB$, it remains to observe that $R$ can indeed be produced in an online fashion,
because $R_t$ does not depend on any future requests in $\catratrequestsequence$. Also, $\algB$ is deterministic if $\algA$ is. 

%%%%%

\begin{claim}\label{claim: sh paging lower bound opt}
$\opt(\slotrequestsequence) \le k\, \opt(\catratrequestsequence)$.
\end{claim}

To prove this claim,
  given an optimal solution $(R^\ast_1, \ldots, R^\ast_T)$ for $\catratrequestsequence=(R_0,C_1, \ldots, C_T)$,
  consider the corresponding solution $D^\ast$ for $\slotrequestsequence$  that starts in configuration $R^\ast_0 = R_0$,
  then, for each $t\in [T]$, responds to $\forcingseq(\calC\setminus\braced{C_t})^k$ 
  by having its cache in configuration $R^\ast_t\in\calC\setminus\braced{C_t}$
  for all requests in $\forcingseq(\calC\setminus C_t)^k$.  
    % \mareksrevision{
  For each $t\in [T]$,
  the response of $D^\ast$ to $\forcingseq(\calC\setminus \braced{C_t})^k$ costs $0$ if $R^\ast_{t-1} = R^\ast_t$ (the rat didn't move)
 % \reporttag{(R1.17)}                                                                            % <------------ report tag
  and otherwise at most $k$ (to transition the cache from $R^\ast_{t-1}$ to $R^\ast_t$).
  % }
  This proves Claim~\ref{claim: sh paging lower bound opt}.

%%%%%

\begin{claim}\label{claim: sh paging lower bound approximation}
Algorithm~$\algB$ is $O(f(k))$-competitive.
 % \reporttag{(R1.18)}                                                                            % <------------ report tag
\end{claim} 

% \marekscommentintext{The reviewer is right here, the $O(1)$ below not quite right. 
% Perhaps we should just bite the bullet and do the calculation, taking also into account doubling the
% cost of $D$ in the argument above?
% }

% \mareksrevision{
  Since $\algA$ is $f(k)$-competitive, $\cost(D) \le f(k) \opt(\slotrequestsequence) + c_k$, where $c_k\in \mathbb R$ is determined by $k$ alone.
  Whenever the rat moves (i.e., $R_{t-1}\ne R_{t}$),
  by the definition of $\calC$,  the Hamming distance between $R_{t-1}$ and $R_{t}$ is $\Omega(k)$,
 % \reporttag{(R1.19)}                                                                            % <------------ report tag
so $D$ pays
% paid
$\Omega(k)$ to transition from $R_{t-1}$ to $R_t$ (possibly in multiple steps).
  Using Claim~\ref{claim: sh paging lower bound opt},
  we obtain that $\cost(R) = O(\cost(D)/k)=O((f(k) \opt(\slotrequestsequence) + c_k)/k)= O(f(k)\opt(\catratrequestsequence) + c_k/k)$.
  So $R$ is an $O(f(k))$-competitive solution for $\mu$.
  This proves Claim~\ref{claim: sh paging lower bound approximation},
  completing the proof of the lemma.
% }
\end{proof}

%%% Local Variables:
%%% mode: latex
%%% TeX-master: "0_0_main__short"
%%% End:

%% file: 4_page_laminar_paging.tex
Recall that \PageLaminarPaging generalizes \Paging by allowing each request to be a set $\pageset$ of pages.
The request $\pageset$ is satisfiable by having any page $p\in\pageset$ in the cache.
We require $\pageset\in \pagesetfamily$, where $\pagesetfamily$ is a pre-specified laminar collection of sets of pages,
whose height we denote by $h$. (See the example in Figure~\ref{fig: page laminar example}.)
To our knowledge, this problem has not been yet studied in the literature. In particular, we do not know
whether the optimum solution can be computed in polynomial time. 

  %%%%%%%% 

  \begin{figure}[t]
    \begin{center}
      \includegraphics[width = 4in]{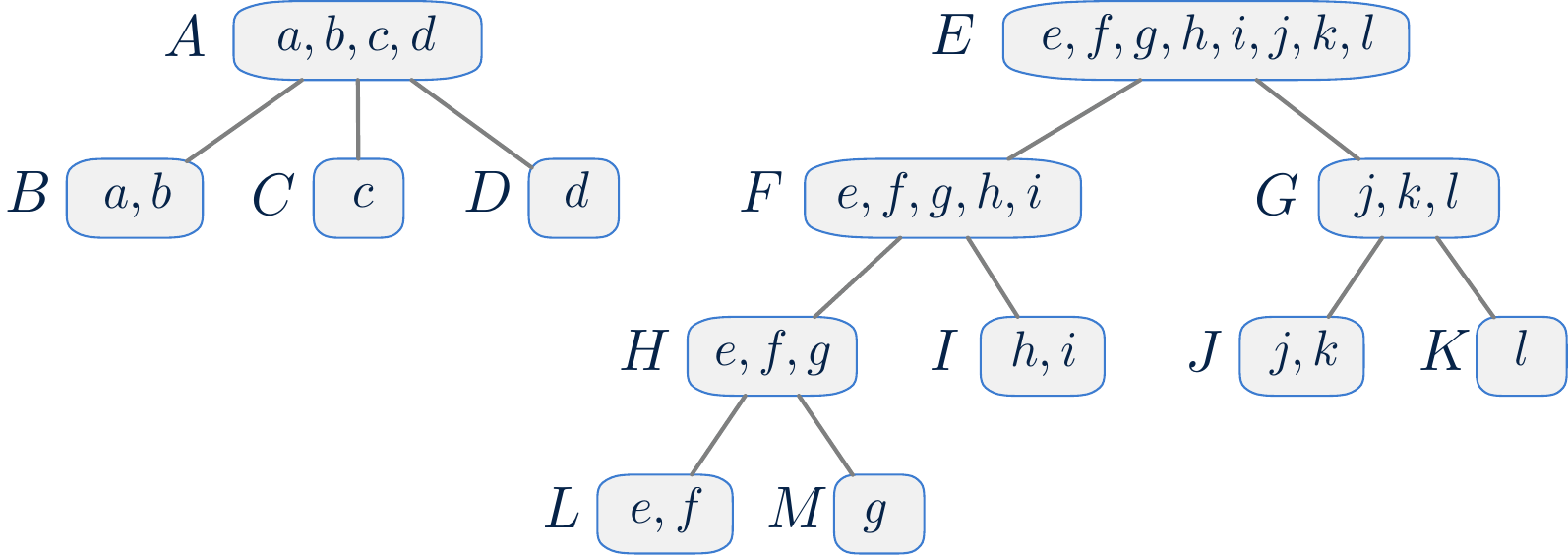}
    \end{center}
    \caption{An example of a laminar family $\pagesetfamily$ of height $4$.}%
    \label{fig: page laminar example}
  \end{figure}

%%%%%%%%%%%%%%%%

\begin{theorem}\label{thm: page laminar}
  \PageLaminarPaging admits the following polynomial-time algorithms:
  an $hk$-competitive deterministic online algorithm, 
  an $h H_k$-competitive randomized online algorithm, and
  an offline $h$-approximation algorithm.
\end{theorem}

The proof is by reduction to standard \Paging. Known polynomial-time algorithms for standard \Paging
include an optimal offline algorithm~\cite{Belady1966Study},
a deterministic $k$-competitive online algorithm~\cite{DBLP:journals/cacm/SleatorT85}
and a randomized $H_k$-competitive online algorithm~\cite{DBLP:journals/tcs/AchlioptasCN00}.
Theorem~\ref{lemma: page laminar} follows directly from composing these known results with 
the following lemma. 

%%%%%%%%%%%%%%%%

\begin{lemma}\label{lemma: page laminar}
  Every $f(k)$-approximation algorithm \algA for \Paging
  can be converted into an $h f(k)$-approximation algorithm \algB for \PageLaminarPaging,
  preserving the following properties: being polynomial-time, online, and/or deterministic.
\end{lemma}

\begin{proof}
  Let \algA be any (possibly online, possibly randomized)  $f(k)$-approximation algorithm for \Paging.
  Let \PageLaminarPaging instance $\pagerequestsequence$ be the input to algorithm \algB.
  For any time step $t$ and set $\pageset\in\pagesetfamily$, let $\prevc \pageset t$ denote the child of $\pageset$ whose subtree
  contains $\pagerequestsequence$'s most recent request to a proper descendant of $\pageset$.
  This is the child $c$ of $\pageset$ such that $\pagerequest_{t'} \subseteq c$,
  where $t'=\max\{i \le t : \pagerequest_i \subset \pageset\}$.
  If there is no such request ($t'$ is undefined %\textcolor{red}
  {or $P$ is a leaf}), then define $\prevc \pageset t = \pageset$.
  Define $\prevp \pageset t$ inductively via $\prevp \pageset t = \prevp {\prevc \pageset t} t$ when $\prevc \pageset t \ne \pageset$,
  and otherwise $\prevp \pageset t$ is an arbitrary (but fixed) page in $\pageset$.
  Call $\prevc \pageset t$ and $\prevp \pageset t$ the \emph{preferred child} and \emph{preferred page} of $\pageset$ at time $t$.
  At any time, $\pageset$'s preferred page can be found by starting at $\pageset$ and tracing the path down through preferred children.

  Define a \Paging instance $\sigma$ from the given instance $\pagerequestsequence$
  by replacing each request $\pagerequest_t$ in $\pagerequestsequence$
  by its preferred page $\prevp {\pagerequest_t} t$  (so $\sigma_t = \prevp {\pagerequest_t} t$).
  Algorithm \algB just simulates \Paging algorithm \algA on input $\sigma$,
  and maintains its cache exactly as \algA does. (Note that $\sigma$ can be computed online, deterministically, in polynomial time.) 
  Algorithm \algB is correct because any solution to $\sigma$ is also a solution to $\pagerequestsequence$
  (because $\sigma_t = \prevp {\pagerequest_t} t \in \pagerequest_t$).
  And $\cost(\algB(\pagerequestsequence)) = \cost(\algA(\sigma))$.
  To finish proving the lemma, we show $\opt(\sigma) \le h\, \opt(\pagerequestsequence)$.

  For any requested set $\pageset$,
  define a \emph{$\pageset$-phase} of $\pagerequestsequence$ to be a maximal contiguous interval $[i, j] \subseteq [1, T]$
  such that 
  $
  %\textcolor{red}
  {\pagerequestsequence_t} \not\subset \pageset$ for $t \in [i+1, j]$.
  The $\pageset$-phases for a given $\pageset$ partition $[1, T]$.
  Each $\pageset$-phase $[i, j]$ (except possibly the first, with $i=1$) starts with a request to a proper descendant of $\pageset$,
  but there are no such requests during $[i+1, j]$.
  It follows that $\prevc \pageset i$ and $\prevp \pageset i$ remain the preferred child and page of $\pageset$ throughout the phase,
  that the preferred child $\prevc \pageset i$ of $\pageset$ also has the same preferred page $\prevp \pageset i$ throughout the phase,
  and that $\pageset$-phase $[i, j]$ is contained in some $\prevc \pageset i$-phase.
  By definition of $\sigma$,  each request to $\pageset$ in the given instance $\pagerequestsequence$ in interval $[i, j]$
  is replaced in $\sigma$ by a request to $\pageset$'s preferred page, $\prevp \pageset i$.

\myparagraph{Example.}
Consider the laminar family $\pagesetfamily$ in Figure~\ref{fig: page laminar example}. Consider also a 
request sequence $\pi$ whose requests $\pi_{61},\pi_{62},\ldots,\pi_{79}$ are shown in the table below:
%

%\bigskip
%\newlength{\collen}
%\setlength{\collen}{0.09in}
%\noindent
%\begin{tabular}{p{0.75in}p{\collen}p{\collen}p{\collen}p{\collen}p{\collen}p{\collen}p{\collen}p{\collen}p{\collen}p{\collen}p{\collen}p{\collen}p{\collen}p{\collen}p{\collen}p{\collen}p{\collen}p{\collen}p{\collen}p{\collen}p{\collen}}
%time step: & \ldots & 61 & 62 & 63 & 64 & 65 & 66 & 67 & 68 & 69 & 70 & 71 & 72 & 73 & 74 & 75 & 76 & 77 & 78 & 79 & \ldots
%\\
%sequence $\pi$:     & \ldots & $A$ & $B$ & $H$  & $E$ & $C$ & $K$  & $A$ & $I$ & $D$ & $E$ & $M$  & $B$ & $H$  & $A$ & $D$ & $E$ & $G$ & $A$ & $F$   & \ldots
%\\ \hline
%$E$-phases:         & \ldots &--  &--  &  $H$ &--  &--  & $K$  &-- & $I$ &--  &--   & $M$ &-- & $H$  &--    &--    &--  & $G$ &--  & $F$   & \ldots
%\\
%$F$-phases:         & \ldots &--  &--  &  $H$ &--  &--  &--   &--  & $I$ &--  &--   & $M$ &-- &  $H$ &--  &--    &--  &--  &--  &--    & \ldots
%\\
%sequence $\sigma$:  & \ldots & $a$ & $a$ &  $e$ & $e$ & $c$ & $l$  & $c$ & $h$ & $d$ & $h$  & $g$ & $a$ &  $g$ &$a$  &$d$  & $g$ & $l$ & $d$ & $g$   & \ldots
%\end{tabular}
%\bigskip
% ------- START REPLACEMENT TEXT FOR PROBLEMATIC EXAMPLE ---

{\small
    \[\arraycolsep=3.1pt
    \begin{array}{r c c c c c c c c c c c c c c c c c c c c c}
        \text{time step}: & \cdots & 61 & 62 & 63 & 64 & 65 & 66 & 67 & 68 & 69 & 70 & 71 & 72 & 73 & 74 & 75 & 76 & 77 & 78 & 79 & \cdots
        \\
        \text{sequence } \pi: & \cdots & A & B & H & E & C & K & A & I & D & E & M & B & H & A & D & E & G & A & F & \cdots
        \\ \hline
        E\text{-phases}: & \cdots & - & - & H & - & - & K & - & I & - & - & M & - & H & - & - & - & G & - & F & \cdots
        \\
        F\text{-phases}: & \cdots & - & - & H & - & - & - & - & I & - & - & M & - & H & - & - & - & - & - & - & \cdots
        \\
        \text{sequence } \sigma: & \cdots & a & a & e & e & c & l & c & h & d & h & g & a & g & a & d & g & l & d & g & \cdots
    \end{array}
    \]
}

% ------- END REPLACEMENT TEXT FOR PROBLEMATIC EXAMPLE ---

  \noindent
The third row in the table shows $E$-phases, marking the beginning of each phase with the request at that step.
The fourth row shows $F$-phases. 
Assume that at time $61$ the preferred page of each set in $\pagesetfamily$ is the leftmost page in the leftmost leaf of its subtree. For example, 
$p_{61}(E) = e$. Then the fifth row shows the sequence of preferred pages that forms the resulting sequence $\sigma$.

%%%%%%%%%%%%%%%%%%%%%%%%%%%%%%%%%%

  \begin{figure}
    \centering\framebox{\parbox{0.97\textwidth}{\setlength{\parindent}{1em}
        \begin{steps}
          \step Initialize the current instance $\pagerequestsequence'$ and current solution $\C'$ to the given instance $\pagerequestsequence$ and its solution $\C$.

          \step Incrementally modify $\pagerequestsequence'$ and $\C'$ by \emph{repairing} each phase, as follows.

          \step
          While there is an unrepaired phase,
          choose any unrepaired $\pageset$-phase $[i, j]$
          such that all proper descendants of $\pageset$ have already been repaired,
          then repair the chosen phase as follows:
% replacing block by steps - Rajmohan
          \begin{steps}
            \step
            Modify the current instance $\pagerequestsequence'$ by
            replacing each request to $\pageset$ during $[i,j]$ in $\pagerequestsequence'$ by a request to $\pageset$'s preferred page $p = \prevp \pageset i$.
            (So, after all phases are repaired, the current instance $\pagerequestsequence'$ will equal $\sigma$.)
            
            \step\label{step: repair solution}
            Modify the current solution $\C'$ during $[i, j]$ accordingly,
            to ensure that $\C'$ continues to satisfy $\pagerequestsequence'$.
            To do that, we will establish a stronger property throughout $[i, j]$, namely:
            \emph{whenever $\C'$ has at least one page in $\pageset$ cached, $\C'$ has $p$ cached.}
			
            Say that time $t\in [i, j]$ \emph{needs repair}
            if, at time $t$, $\C'$ caches at least one page in $\pageset$, but not $p$.
            For $t\gets i, i+1, \ldots, j$, if time $t$ needs repair,  modify what $\C'$ caches at time $t$
            by replacing one of its currently cached pages $q_t\in \pageset$ by $p$,
            where $q_t$ is defined greedily as follows
            \begin{align*}
              q_t  = \begin{cases}
                q_{t-1} & \textit{if \(q_{t-1}\) is defined and still cached at time \(t\)} \\
                \textit{any page in \(\pageset\) cached at time \(t\)}& \textit{otherwise.}
              \end{cases}
            \end{align*}
            This completes the repair of this $\pageset$-phase $[i, j]$.
            The algorithm terminates after it has repaired all phases.
          \end{steps}
        \end{steps}
      }
    }
    \caption{The algorithm that transforms $(\pagerequestsequence, \C)$ into $(\sigma, \C')$,
      by repairing each phase.}\label{fig: repair}
  \end{figure}

%%%%%%%%%%%%%%%%%%%%%%%%%%%%%%%%%%

\smallskip

  To prove that $\opt(\sigma) \le h\, \opt(\pagerequestsequence)$, we will start with an optimal
  solution for $\pagerequestsequence$ and convert it into a solution of $\sigma$ while
  increasing the cost at most by a factor of $h$.
  So let $\C=({\C_1}, \ldots, \C_T)$ be an optimal solution for $\pagerequestsequence$. The
  conversion of $\C$ into a solution for $\sigma$ is given in Figure~\ref{fig: repair}, and is described as 
  an algorithm that incrementally modifies, or ``repairs'', both $\C$ and $\pagerequestsequence$,  phase by phase.
  It  maintains the invariant that the current solution, denoted $\C'$,
  is always correct for the current instance, denoted $\pagerequestsequence'$.
  (In $\pagerequestsequence'$, some requests will be to a page, rather than a set.
  Any such request is satisfied only by having the requested page in the cache.)
  At the end the modified instance $\pagerequestsequence'$ will equal $\sigma$,
  so that the modified solution $\C'$ will be a correct solution for $\sigma$.

  Specifically, we will show the following claim (whose proof we postpone): 

%%%%%%%%%%

  \begin{claim}\label{claim: correctness invariant}
    The repair algorithm maintains the invariant
    that the current solution $\C'$ is correct for the current instance $\pagerequestsequence'$,
    so at termination $\C'$ is a correct solution for $\sigma$.
  \end{claim}
 
  Next we bound the cost, as follows.  Call a phase \emph{costly} if its repair increases the cost of $\C'$,
  and \emph{free} otherwise.  We show that the number of costly phases is at most $(h-1)\cost(\C)$,
  and that the repair of each phase increases the cost of $\C'$ by at most 1.
  This implies that the final cost of $\C'$  is at most $\cost(\C) + (h-1)\cost(\C) = h\,\cost(\C)$, as desired.
  Specifically, we will show the following claims (whose proofs we postpone): 

%%%%%%%%%%

  \begin{claim}\label{claim: costly one}
    For any requested set $\pageset$, the repair of any $\pageset$-phase $[i, j]$ increases the cost of $\C'$ by at most $1$,
    and only if $j\ne T$.
  \end{claim}

%%%%%%%%%%
  
  \begin{claim}\label{claim: costly number}
    For any non-leaf set $\pageset$,
    the number of costly $\pageset$-phases is at most the cost paid by $\C$ for pages in $\pageset$
    (that is, the number of retrievals of pages in $\pageset$ by $\C$).
  \end{claim}
  
  By the definition of $\pageset$-phases, each leaf set $\pageset$ has only one $\pageset$-phase $[1, T]$,
  so by Claim~\ref{claim: costly one} only non-leaf sets have costly phases.
  Each page $p$ is in at most $h-1$ non-leaf sets $\pageset$,  so Claim~\ref{claim: costly number}
  implies that the total number of costly phases is at most $(h-1)\cost(\C)$.
  This and Claim~\ref{claim: costly one} imply that the final cost of $\C'$ is at most $h\, \cost(\C) = h\,\opt(\pagerequestsequence)$,
  proving Lemma~\ref{lemma: page laminar}. 

  It remains only to prove the three claims.
  
  \begin{proof}[Proof of Claim~\ref{claim: correctness invariant}.]
    The invariant holds initially when $\C' = \C$ and $\pagerequestsequence' = \pagerequestsequence$,
    just because by definition $\C$ is an (optimal) solution for $\pagerequestsequence$.
    Suppose the invariant holds just before the repair of some $\pageset$-phase $[i, j]$.
    We will show that it continues to hold after. The repair modifies $\pagerequestsequence'$
    by replacing each request to $\pageset$ during $[i, j]$ by a request to its preferred page $p=\prevp \pageset i$.
	
    Consider any time $t\in [i, j]$. First consider the case that (before the repair)
    $\pagerequestsequence'$ requested $\pageset$ at time $t$.  In this case, $\C'$ cached some page in $\pageset$,
    so (by Step~\ref{step: repair solution} of the repair algorithm) after the repair $\C'$ has the preferred page $p$ cached,
    and thus $\C'$ satisfies the modified request (for $p$).
	
    The other case is when $\pagerequestsequence'$ requested page at time $t$ is not $\pageset$. 
	Then the repair doesn't modify the request in $\pagerequestsequence'$ at time $t$.
    In this case, either the repair doesn't modify the cache at time $t$
    (in which case $\C'$ continues to satisfy $\pagerequestsequence'$),
    or the repair replaces some page $q_t$ in the cache by $p$. If that happens, the page $q_t$ is also in $\pageset$.
    Also, the request in $\pagerequestsequence$ at time $t$ cannot be to a proper descendant of $\pageset$
    (by definition of $\pageset$-phase), so the request in $\pagerequestsequence'$ at time $t$
	is either to an ancestor of $\pageset$, a set disjoint from $\pageset$, or an already repaired page not in $\pageset$.
    (We use here that no proper ancestors of $\pageset$ have yet had their phases repaired.)
    In all three cases, after swapping $p$ for $q_t$ (with $p,q_t\in \pageset$), the request must still be satisfied.
    So the invariant holds after the repair. At termination the invariant holds so $\C'$ is correct for $\sigma$. 
  \end{proof}

  \begin{proof}[Proof of Claim~\ref{claim: costly one}]
    Consider the repair of any $\pageset$-phase $[i, j]$ for any requested set $\pageset$.
    This repair modifies the cache only at times in $[i, j]$.
    Recall that the cost of $\C'$ at time $t$ is the number of retrievals at time $t$,
    where a retrieval is a page that $\C'$ caches at time $t$ but not at time $t-1$.

    At each time $t\in [i, j]$ that needed repair (as defined in Step~\ref{step: repair solution} of the algorithm),
    the repair of the phase replaced some page $q_t$ in the cache at time $t$ by the preferred page $p = \prevp{\pageset}{i}$.
    This can increase the cost only at times in $[i, j+1]$, and by at most $1$ at each such time.

    We claim the cost cannot increase at time $i$. Indeed, in the case that $i=1$,
    the cost at time $i$ equals the number of pages cached at time $i$,  which a repair at time $1$ doesn't change.
    In the case that $i>1$,
    at time $i$ no repair is done, because $\pagerequestsequence$ requests a proper descendant $d$ of $\pageset$,
    and the $d$-phase containing $[i, j]$ has already been repaired,
    so $\pagerequestsequence'$ requests $\prevp d i$ at time $i$, and now our invariant
	implies that $\C'$ already caches $\prevp d i$ at time $i$.
    But by definition $\prevp{\pageset}{i} = \prevp d i$, so $\C'$ already caches $\pageset$'s preferred page $p$ at time $i$.
    Thus the cost cannot increase at time $i$.

    Next, consider any time $t\in [i+1, j]$. To prove the claim, we show that the repair didn't increase the cost at time $t$.
    The repair can introduce up to two new retrievals at time $t$:
    a new retrieval of $p$,  and/or a new retrieval of $q_{t-1}$.
    We show that, for each retrieval that the repair introduced, it removed another one.

    Suppose it introduced a new retrieval of $p$ at time $t$. 
    That is, it replaced $q_t$ by $p$ in the cache at time $t$
    and (after modification) $\C'$ doesn't cache $p$ at time $t-1$.
    The latter property (by inspection of Step~\ref{step: repair solution} of the algorithm)
    implies that $\C'$ caches no pages in $\pageset$ at time $t-1$ (before or after modification).
    It follows that the repair removed one retrieval of $q_t$ at time $t$.

    Now suppose the repair introduced a new retrieval of $q_{t-1}$ at time $t$. 
    That is, it removed $q_{t-1}$ from the cache at time $t-1$ (replacing it by $p$)
    and (after modification) $\C'$ caches $q_{t-1}$ at time $t$.
    The latter property implies that time $t$ did not need repair
    (because if it did the repair would have taken $q_t=q_{t-1}$).
    So $\C'$ cached $p$ at time $t$ (before and after modification).
    Thus the repair removed a retrieval of $p$ at time $t$.

    Overall, for each retrieval introduced at times in $[i,j]$, another was removed, and therefore the cost at
	times in $[i,j]$ didn't increase. The cost increase at time $j+1$, if any, can
	only be caused by the repair at time $t=j$ that replaced 
	$q_j$ by $p$, so this increase is at most $1$. This completes the proof of the claim.
  \end{proof}

  \begin{proof}[Proof of Claim~\ref{claim: costly number}]
    Fix any non-leaf set $\pageset$. First consider the repair of any $\pageset$-phase $[i, j]$ with $i>1$.
    Let $c=\prevc \pageset i \ne \pageset$ and $p=\prevp \pageset i = \prevp c i$
    be the preferred child and page throughout the phase.  When the repair starts,
    the $c$-phase containing $[i, j]$ has already been repaired.
    By inspection, that repair established the following property of $\C'$ throughout $[i, j]$:
    \emph{whenever $\C'$ has at least one page in $c$ cached, $\C'$ has $c$'s preferred page $p$ cached.}
    None of the ancestors of $c$ (including $\pageset$) have had their phases repaired since then,
    so this property still holds just before the repair of $\pageset$.

Since $i>1$, the definition of $\pageset$-phases implies that at time $i$
    the instance $\pagerequestsequence$ requests a descendant of $c$,  so $\C$ caches at least one page $p'$ in $c$.
    Suppose that $\C$ doesn't evict $p'$ during $[i+1, j]$.
    Then, \emph{at every time during $[i, j]$, $\C$ caches at least one page in $c$.}
    Each repair for $c$ or a descendant of $c$ preserves this property
    (because such a repair only replaces cached pages in $c$ by other pages in $c$).
    Throughout $[i, j]$, then, $\C'$ also caches at least one page in $c$
    and, by the previous paragraph, must cache $\pageset$'s preferred page $p$.
    (Recall that $\pageset$ and $c$ have the same preferred page throughout $[i, j]$.)
    In this case, by inspection of the algorithm the repair of this $\pageset$-phase does not change $\C'$,
    and the phase is not costly.  We conclude that, for the phase to be costly, $\C$ must evict $p'$ during $[i+1, j]$.

    Note that $p'\in c\subset \pageset$. Also, by Claim~\ref{claim: costly one}, this is not the final $\pageset$-phase (that is, $j<T$).
    It follows that the number of costly $\pageset$-phases $[i,j]$ with $i>1$
    is at most the number of evictions of pages in $\pageset$ by $\C$,  before the final $\pageset$-phase (with $j=T$).

    Regarding the $\pageset$-phase $[i, j]$ with $i=1$, if it is costly, then $j<T$
    and, by the reasoning in the paragraph before last,
    in the final $\pageset$-phase, there is a page $p'$ in $\pageset$ that $\C$ either evicts  or leaves in the cache at time $T$.

    By the above reasoning, the number of costly $\pageset$-phases  is at most the number of evictions by $\C$ of pages in $\pageset$,
    plus the number of pages left cached by $\C$ at time $T$.
    This sum is the number of retrievals of pages in $\pageset$ by $\C$,  proving Claim~\ref{claim: costly number}.
  \end{proof}

As explained earlier, Claims~\ref{claim: correctness invariant},~\ref{claim: costly one} and~\ref{claim: costly number}
imply the lemma, thus the proof is now complete.
\end{proof}

%%% Local Variables:
%%% mode: latex
%%% TeX-master: "0_0_main__short"
%%% End:

%% file: 5_0_laminar_sc_paging.tex
In this section we prove upper bounds for \SlotLaminarPaging given in
Table~\ref{table: results}.  Recall that in \SlotLaminarPaging the family
$\slotsetfamily$ is assumed to be a laminar family of slot sets whose
height we denote by $h$.  Theorem~\ref{thm: slot laminar} bounds the
optimal ratios by $3h^2k$ (deterministic), $3h^2 H_k$ (randomized) and
$3h^2$ (offline polynomial-time approximation).  The proof of
Theorem~\ref{thm: slot laminar} (Section~\ref{sec: slot laminar}) is
by a reduction of \SlotLaminarPaging to \PageLaminarPaging, 
studied in Section~\ref{sec: page laminar paging}.  Theorem~\ref{thm:
  laminar upper bound}, presented in Section~\ref{sec: laminar
  k-server}, tightens the deterministic upper bound to $2hk$.

%%% Local Variables:
%%% mode: latex
%%% End:

%% file: 5_1_laminar_reduction.tex
{}\label{sec: slot laminar}

%%%%%%%%%%%%%

\begin{theorem}\label{thm: slot laminar}
% \mareksrevision{
  \SlotLaminarPaging admits the following polynomial-time algorithms:
  a deterministic online $3h^2 k$-competitive algorithm, 
  a randomized online $3h^2 H_k$-competitive algorithm, and
  an offline $3h^2$-approximation algorithm.
  % }
\end{theorem}

Our focus here is on uniform treatment of the three variants of
\SlotLaminarPaging in the above theorem.   The ratios in this
theorem have not been optimized. For example, in Section~\ref{sec:
  laminar k-server} we give a better deterministic algorithm.  For the
special case when $h=2$ the problem can be reduced to \AllOrOnePaging,
for which the ratio can be improved even
further~\cite{castenow+fkmd:server}.

% \mareksrevision{
The proof of Theorem~\ref{thm: slot laminar} is by a reduction of
\SlotLaminarPaging to \PageLaminarPaging, presented in Lemma~\ref{lemma: slot laminar} below.
% }
The reduction uses a relaxation of \SlotLaminarPaging that relaxes the
constraint that each slot hold at most one page (but still enforces
the cache-capacity constraint), yielding an
instance of \PageLaminarPaging.  The reduction simulates the given
\PageLaminarPaging algorithm on multiple instances of
\PageLaminarPaging --- one for each set $S\in\slotsetfamily$, obtained
by relaxing the subsequence that contains just those requests
contained in $S$ --- then aggregates the resulting \PageLaminarPaging
solutions to obtain the global \SlotLaminarPaging solution.
Lemma~\ref{lemma: slot laminar} and Theorem~\ref{thm: page laminar}
(for \PageLaminarPaging) immediately imply Theorem~\ref{thm: slot laminar}.

%%%%%%%%%%%%%

\begin{lemma}\label{lemma: slot laminar} 
  Every $f_h( k)$-approximation algorithm \algA for \PageLaminarPaging
  can be converted into a $3h f_h( k)$-approximation algorithm \algB
  for \SlotLaminarPaging, preserving the following properties: being
  polynomial-time, online, and/or deterministic.
\end{lemma}

\begin{proof}
  We first define the \PageLaminarPaging \emph{relaxation} of a given \SlotLaminarPaging instance.
  The idea is to relax the constraint that each slot can hold at most one page,
  while keeping the cache-capacity constraint. The relaxed problem is equivalent to a \PageLaminarPaging instance
  over ``virtual'' pages $v(p, s)$ corresponding to page/slot pairs $(p,s)$.
  This virtual page  can be placed in any slot, although it represents page $p$ being in slot $s$.  

Formally, this relaxation is defined as follows.
% \mareksrevision{
    Fix any $k$-slot \SlotLaminarPaging instance $\slotrequestsequence=(\slotrequest_1, \ldots, \slotrequest_T)$
    with requestable slot-set family $\slotsetfamily$.  
% }
	For any page $p$ and $\slotset\in\slotsetfamily$,
    define $V(p, \slotset) = \{v(p, s) \suchthat s\in \slotset\}$,  where $v(p, s)$ is a \emph{virtual page} for the pair $(p, s)$.
	Define the \emph{relaxation} of $\slotrequestsequence$
    to be the $k$-slot  \PageSubsetPaging  instance $\pagerequestsequence=(\pagerequest_1, \ldots, \pagerequest_T)$
    defined by   $\pagerequest_t = V(p_t, \slotset_t)$
    (where $\slotrequest_t = \slotrequestpair {p_t} {\slotset_t}$, for $t\in [T]$).
    The requestable-set family for $\pagerequestsequence$ is
    $\pagesetfamily = \big\{V(p, \slotset) \suchthat p \text{ is any page and } \slotset\in\slotsetfamily\big\}$.
    Crucially, if $\slotsetfamily$ is slot-laminar with height $h$,
    then $\pagesetfamily$ is page-laminar with the same height $h$.
 
 % \mareksrevision{
  Instance $\pagerequestsequence$ is a relaxation of $\slotrequestsequence$
  in the sense that for any solution $\C$ for $\slotrequestsequence$ there is a solution $\C'$ for $\pagerequestsequence$
  with $\cost(\C') \le \cost(\C$).
  Namely, this $\C'$
 mimics
 $\C$ by keeping in its cache the virtual pages $v(p, s)$ such that $\C$ has page $p$ cached in slot $s$.
 It follows that $\opt(\pagerequestsequence) \le \opt(\slotrequestsequence)$.

\begin{figure}[t]
\centering
\includegraphics[width = 5in]{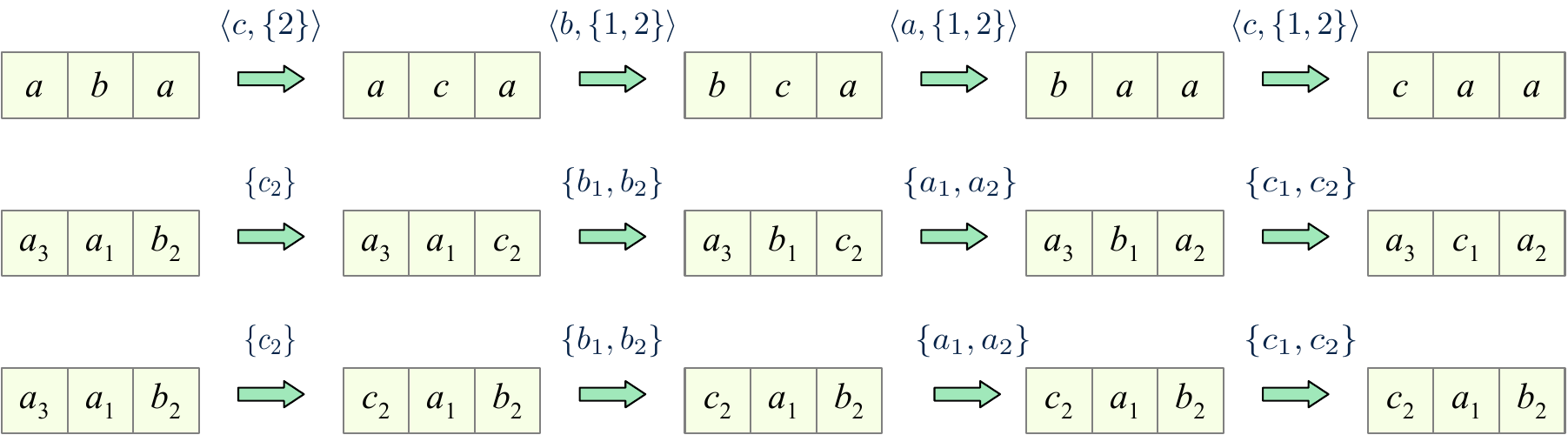}
\caption{At the top, an instance of \SlotLaminarPaging and its solution $\C$.
In the middle row, the corresponding solution $\C'$ of the relaxed instance.
The bottom row shows a cheaper solution of the relaxed instance.
}\label{fig: laminar example}
\end{figure}

  \myparagraph{Example.}
  We illustrate these concepts with a simple example for a cache of size $k=3$
  (see Figure~\ref{fig: laminar example}).
  Let the laminar family be $\slotsetfamily = \braced{\braced{1},\braced{2},\braced{3},\braced{1,2},\braced{1,2,3}}$.
  Suppose that the initial cache has page $a$ in slots $1$ and $3$ and page $b$ in slot $2$.
  To simplify notation for virtual pages, in this example we write $p_s$ instead of $v(p,s)$.
  Suppose also that in the corresponding relaxed instance the cache has initially virtual pages $a_1,b_2,a_3$ (the slot assignment here is not significant). 
  Now, consider the request sequence,
  $\slotrequestsequence = (\slotrequestpair{c}{\braced{2}}, \slotrequestpair{b}{\braced{1,2}}, \slotrequestpair{a}{\braced{1,2}}, \slotrequestpair{c}{\braced{1,2}})$,
  and its solution $C$ shown in the top row in Figure~\ref{fig: laminar example}.
  This solution has cost $3$.
  The sequence $\pagerequestsequence$ that is the relaxation of $\slotrequestsequence$ is
  $\pagerequestsequence = (\braced{c_2}, \braced{b_1,b_2}, \braced{a_1,a_2}, \braced{c_1, c_2})$.
  The second row shows the solution $\C'$ for $\pagerequestsequence$ corresponding to $\C$. Its cost is also $3$.
  The bottom row shows a solution for $\pagerequestsequence$ whose cost is only $1$.
% }

\smallskip

% \mareksrevision{
  For the rest of the proof, we assume that the family $\slotsetfamily$ has just one root $R$ with $|R|\le k$.
  (This is without loss of generality, as multiple roots, being disjoint,
  naturally decouple any \SlotLaminarPaging instance into independent problems, one for each root.)
  
  For each $\slotset\in\slotsetfamily$, define $\slotset$'s \SlotLaminarPaging
  \emph{subinstance} $\slotrequestsequence_\slotset$
  to be obtained from $\slotrequestsequence$ by deleting all requests that are not subsets of $\slotset$.
    Let $\pagerequestsequence_\slotset$ denote the (\PageLaminarPaging) relaxation of $\slotrequestsequence_\slotset$.
% }

\smallskip

  Next we define the algorithm \algB.  Fix an $f_h( k)$-approximation algorithm \algA for \PageLaminarPaging.
  Fix the input $\slotrequestsequence$
  with $\slotrequest_t=\slotrequestpair {p_t} {\slotset_t}$ (for $t\in[T]$)  to \SlotLaminarPaging algorithm \algB.
  % \mareksrevision{
  For ease of presentation we will focus on the online case: 
  we assume that Algorithm \algA is an online algorithm,
  and we convert it into Algorithm \algB that is also online.
 % \reporttag{(R1.20)}                                                                            % <------------ report tag
  If \algA is offline, the reduction we give below can be naturally modified to produce \algB as an offline algorithm instead.
  % }

% \mareksrevision{
	The description of \algB, below, is top-down: we start with an overview that explains the 
	fundamental strategy. We then formulate an invariant that \algB needs to maintain to 
	ensure correctness. Next we provide the details, proving that \algB indeed satisfies
	this invariant. We wrap up the proof by bounding \algB's cost.
% }

% \mareksrevision{
  Algorithm \algB on input $\slotrequestsequence$ executes,
  simultaneously, $\algA(\pagerequestsequence_\slotset)$ for every requestable set $\slotset\in\slotsetfamily$,
  giving each execution $\algA(\pagerequestsequence_\slotset)$ its own independent cache of size $|\slotset|$ composed of copies of the slots in $\slotset$.
  Guided by \algA, for each such $\slotset$, Algorithm \algB will internally build its own solution, denoted $\algB(\slotrequestsequence_\slotset)$, for $\slotrequestsequence_\slotset$,
 % \reporttag{(R1.21)}                                                                            % <------------ report tag
  also using its own designated cache of size $|\slotset|$ composed of copies of the slots in $\slotset$.
  These solutions are not independent---the choices made for some $\slotset$ may affect the solution constructed for its ancestors.
  (This will be captured by Invariant~(I) given below.)
  We stress that all the above actions are internal to  \algB.
  The actual output produced by \algB on input $\slotrequestsequence$ for a cache of size $k$ is 
  $\algB(\slotrequestsequence_R)$. (Recall that $\slotrequestsequence_R = \slotrequestsequence$.)
  All other solutions $\slotrequestsequence_\slotset$ are used by \algB, roughly, only as a way to
  represent information about the past.
% }

% \marekscommentintext{
% The paragraph below is confusing, and even self-contradictory. If all copies of $p$ are functionally
% equivalent, then why to evict a copy if it serves the current request (see last sentence).
% }

% \nymargincomment{TODO: REWORD PARAGRAPH L 688}
  For internal bookkeeping purposes, in presenting Algorithm \algB, for each page $p$,
  we let the virtual pages $v(p, s)$ (for every slot $s$, as defined for \PageLaminarPaging) represent copies of page $p$.
  We have \algB maintain cache configurations that place these virtual pages in specific slots,
  with the understanding that the actual cache configurations are obtained by
  replacing each virtual page $v(p, s)$ (in whatever slot it's in) by a copy of page $p$.
 % \reporttag{(R1.22)}                                                                            % <------------ report tag
% \textcolor{blue}{
  So, when a virtual page $v(p, s)$ is in a slot $s'$ (with, possibly, $s'\ne s$) this satisfies any request $\slotrequestpair{p}{S'}$ with $s'\in S'$.
  We then overcount the cost by considering two copies $v(p, s)$ and $v(p', s')$
  to be distinct unless $(p', s') = (p, s)$.
  In particular, if \algB evicts $v(p, s)$ while retrieving $v(p, s')$ (with $s'\ne s$) in the same slot,
  we still charge 1 to the cost of \algB.
  We will upper bound \algB's cost overcounted in this way.
%}

% \mareksrevision{
Algorithm \algB will maintain the following invariant over time:
 % \reporttag{(R2.11)}
%%%%%%%%%%%%%%%

\begin{description}\setlength{\itemsep}{-0.02in}
\item{\textbf{Invariant (I)}:} For each requestable set $\slotset$, for each virtual page $v(p, s)$ currently cached by $\algA(\pagerequestsequence_\slotset)$:
  \begin{description} 
    \item{(I1)} the solution $\algB(\slotrequestsequence_\slotset)$ caches $v(p, s)$ in some slot in $\slotset$, and
%
 % \reporttag{(R1.23)}                                                                            % <------------ report tag
    \item{(I2)} if $\slotset$ has a child $c$ with $s\in c$, and $\algB(\slotrequestsequence_c)$ has $v(p, s)$ in its cache $c$, 
      then in $\algB(\slotrequestsequence_\slotset)$ copy $v(p, s)$ is in the same slot as in $\algB(\slotrequestsequence_c)$.
    \end{description}
\end{description}
% }

% \mareksrevision{
As explained in Section~\ref{sec: preliminaries}, by a child of $\slotset$ in Condition~(I2) we mean
a child of $\slotset$ in the forest representation of the laminar family $\slotsetfamily$.
% }

% \mareksrevision{
  The invariant above suffices to guarantee correctness of the solution $\algB(\slotrequestsequence_\slotset)$ for each instance $\slotrequestsequence_\slotset$.
  % }
  Indeed, when $\algB(\slotrequestsequence_\slotset)$ receives a request $\requestpair {p_t} {\slotset_t}$, its relaxation $\algA(\pagerequestsequence_\slotset)$
  has just received the request $\{v(p_t, s) : s\in \slotset_t\}$,
  so $\algA(\pagerequestsequence_\slotset)$ is caching a virtual page $v(p_t, s)$ (for some $s\in \slotset_t$) in $\slotset$.
  By 
  % \mareksrevision{
  Condition~(I1), 
  % }
  then, $\algB(\slotrequestsequence_\slotset)$ also has $v(p_t, s)$ in some slot in $\slotset$.
  In the case $\slotset=\slotset_t$, this suffices for $\algB(\slotrequestsequence_\slotset)$ to satisfy the request.
  In the remaining case $\slotset$ has a child $c$ with $\slotset_t\subseteq c$,
  and $\algB(\slotrequestsequence_C)$ just received the same request,
  so (assuming inductively that $\algB(\slotrequestsequence_c)$ is correct for $\slotrequestsequence_c$)
  $\algB(\slotrequestsequence_c)$ has $v(p_t, s)$ in some slot $s'$ in $\slotset_t$,
  so by 
  % \mareksrevision{
  Condition~(I2)
  % }
  of the invariant $\algB(\slotrequestsequence_\slotset)$ has $v(p_t, s)$
  in the same slot $s'$ in $\slotset_t$, as required.
  In particular, $\algB(\slotrequestsequence_R)$ will be correct for $\slotrequestsequence_R$.

  \begin{figure}
    \[
      \begin{array}{@{}r|c|c|c|} 
        \multicolumn 1 r {}& \multicolumn 3 c {\text{before}} \\ \cline{2-4}
        \multicolumn 4 r {} \\[-1em]
        \multicolumn 1 r {}
                             & \multicolumn 1 c {\textit{slot } s_1}
                             & \multicolumn 1 c {\textit{slot } s'}
                             & \multicolumn 1 c {\textit{slot } s_2} \\ \cline{2-4}
        \textit{root } R: 
                             & x_1 & y_1 & z_1 \\ \cline{2-4}
                             & \vdots & \vdots & \vdots \\ \cline{2-4}
                             & x_i & y_i & z_i \\ \cline{2-4} 
        \textit{parent }\algB(\slotrequestsequence_{P}): 
                             & x_{i+1} & y_{i+1} & v(p, s)  \\ \cline{2-4}
        \algB(\slotrequestsequence_\slotset): 
                             & x_{i+2} & v(p, s) & \color{gray} z_{i+2}\\ \cline{2-4} 
        \textit{child } \algB(\slotrequestsequence_c): 
                             & \color{gray} v(p, s) &  \color{gray} y_{i+3} & \color{gray} z_{i+3}\\ \cline{2-4}
      \end{array}
      {~~~~~\Longrightarrow~~~~~}
      \begin{array}{|c|c|c|} 
        \multicolumn 3 c {\text{after}} \\ \hline 
        \multicolumn 3 r {} \\[-1em]
        \multicolumn 1 c {\textit{slot } s_1}
        & \multicolumn 1 c {\textit{slot } s'}
        & \multicolumn 1 c {\textit{slot } s_2} \\ \hline
        z_1 & x_1 & y_1 \\ \hline
        \vdots & \vdots & \vdots \\ \hline
        z_i & x_i & y_i \\ \hline
        v(p, s) & x_{i+1} & y_{i+1}  \\ \hline
        v(p, s) & x_{i+2} & \color{gray} z_{i+2}\\ \hline
        \color{gray} v(p, s) &  \color{gray} y_{i+3} & \color{gray} z_{i+3} \\ \hline
      \end{array}
    \]
    \caption{``Rotating'' slots in $\algB(\slotrequestsequence_\slotset)$ and ancestors to preserve the invariant. Pages in grey are not moved.}\label{fig: rotate}
  \end{figure}

  To maintain 
  % \mareksrevision{
  Invariant~(I)
  % }
  \algB does the following for each requestable set $\slotset$.
  Whenever the relaxed solution $\algA(\pagerequestsequence_\slotset)$ evicts a page $v(p, s)$, the solution $\algB(\slotrequestsequence_\slotset)$ also evicts $v(p, s)$.
  After this eviction both 
  % \mareksrevision{
  Conditions~(I1) and~(I2)
%}
  will be preserved.
  Whenever $\algA(\pagerequestsequence_\slotset)$ retrieves a page $v(p, s)$, the solution $\algB(\slotrequestsequence_\slotset)$ also retrieves $v(p, s)$,
  into any vacant slot in $\slotset$
  (there must be one, because $\algA(\slotrequestsequence_\slotset)$ caches at most $|\slotset|$ pages).
  This retrieval can cause up to two violations of 
  % \mareksrevision{Condition~(I2)} 
  Condition~(I2)
  of the invariant: one at $\algB(\slotrequestsequence_\slotset)$,
  because $v(p, s)$ is already cached by a child $\algB(\slotrequestsequence_c)$ but in some slot $s_1 \ne s'$;
  the other at the parent $\algB(\slotrequestsequence_P)$ of $\algB(\slotrequestsequence_\slotset)$
  (if any), because $v(p, s)$ is already cached by the parent, but in some slot $s_2 \ne s'$.
  In the case that the retrieval does create two violations (and $s_1\ne s_2$), \algB restores the invariant by
  ``rotating'' the contents of the slots $s_1$, $s'$, and $s_2$
  in $\algB(\slotrequestsequence_\slotset)$ \emph{and in each ancestor}, as shown in Figure~\ref{fig: rotate}.
  Note that $y_{i+3}$ and $z_{i+2}$ cannot be $v(p, s)$,
  so moving $v(p, s)$ out of slots $s'$ and $s_2$ doesn't introduce a violation there.
  Thus this rotation indeed restores the invariant, at the expense of three retrievals at the root.
  (The retrievals at other nodes only modify the internal state of \algB.)
  There are three other cases: two violations with $s_1 = s_2$, one violation at $\algB(\slotrequestsequence_\slotset)$,
  or one violation at its parent, but all these three cases
  can be handled similarly, also with at most three retrievals (in fact at most two) at the root.

  \myparagraph{Total cost.}
  Each retrieval by $\algA(\pagerequestsequence_\slotset)$ causes at most 3 retrievals in $\algB(\slotrequestsequence_R)$, so $\cost(\algB(\slotrequestsequence_R))$ is at most
  \begin{align*}
    \le\, 
    \sum_{\slotset\in\slotsetfamily} 3\,\cost(\algA(\pagerequestsequence_\slotset))
    \,\le\, 
    \sum_{\slotset\in\slotsetfamily} 3\, f_h( |\slotset|) \opt(\pagerequestsequence_\slotset)
    \,\le\, 
    3 f_h( k) \sum_{\slotset\in\slotsetfamily} \opt(\slotrequestsequence_\slotset)
    \,\le\, 
    3 h f_h( k)\, \opt(\slotrequestsequence_R).
  \end{align*}
  The second step uses that $\algA(\pagerequestsequence_\slotset)$ is $f_h( |\slotset|)$-competitive for $\pagerequestsequence_\slotset$.
  The third step uses that $\pagerequestsequence_\slotset$ is a relaxation of $\slotrequestsequence_\slotset$
  so $\opt(\pagerequestsequence_\slotset) \le \opt(\slotrequestsequence_\slotset)$,
  and that $|\slotset|\le k$ so $f_h( |\slotset|) \le f_h( k)$.\footnote
  {We assume here that $f_h(k')\le f_h( k)$ for $k'\le k$, which is without loss of generality
    as one can simulate a cache of size $k'$ using a cache of size $k$
    by introducing artificial requests that force $k-k'$ slots to be continuously occupied.}
  The last step uses that the sets within any given level $i\in\{1,2,\ldots, h\}$ of the laminar family
  are disjoint, so $\opt(\slotrequestsequence_R)$ is at least the sum, over the sets $\slotset$ within level $i$,
  of $\opt(\slotrequestsequence_\slotset)$. This shows that $\algB$ is a $3hf_h(k)$-approximation algorithm.
  To finish, we observe that \algB is polynomial-time, online, and/or deterministic if \algA is.
\end{proof}

%%% Local Variables:
%%% mode: latex
%%% TeX-master: "0_0_main__short"
%%% End:

%% file: 5_2_laminar_deterministic_upper_bound.tex
For \SlotLaminarPaging, this section presents a deterministic algorithm with competitive ratio $O(hk)$,
improving upon the bound of $O(h^2k)$ from Theorem~\ref{thm: slot laminar}.  The algorithm, \algRefSearch, 
refines \algExhSearch.  Like \algExhSearch, it is phase-based and maintains a configuration that can satisfy all requests in a phase;
 however, in order to satisfy the next request in the current phase, the particular configuration is chosen by judiciously moving 
 pages in certain slots that are serving requests along a path in the laminar hierarchy.  

%%%%%%%%%%%%%%%%%

\begin{theorem}\label{thm: laminar upper bound}
For \SlotLaminarPaging,
Algorithm \algRefSearch (Fig.~\ref{fig: refsearch})
has competitive ratio
at most $2\cdot \sumofsizes(\slotsetfamily) - k \le (2h-1)k$.
\end{theorem}

\begin{figure}
  \framebox{\parbox{0.98\textwidth}{
      \begin{steps}
      \item[] \textbf{input:} \SlotLaminarPaging instance $(k, \slotsetfamily, \sigma = (\sigma_1, \ldots, \sigma_T))$
        \step for $t\gets 1, 2, \ldots, T$,  respond to the current request $\sigma_t=\requestpair p S$ as follows:
        \begin{steps}
			 \step\label{step: new phase} if $t=1$ or $R_{t-1}\cup \braced{\sigma_t}$ is not satisfiable:   let $R_{t-1} = \emptyset$ and empty the cache  \algcomment{--- start new phase}
			\step let $R_t = R_{t-1} \cup \{\sigma_t\}$
          	\step\label{step: slp-redundant} if $C_{t-1}$ satisfies  $\sigma_t=\requestpair p S$: 
          let $C_t = C_{t-1}$                                          \algcomment{--- redundant request}
          	\step  else:  \algcomment{--- non-redundant request}
         	 \begin{steps}
         		\step\label{step: determine sequences}
				find sequences  $\langle s_1, \ldots, s_m\rangle$,  $\langle S_0 = S, S_1, \ldots, S_{m-1}\rangle$, and  $\langle p_0 = p, p_1, \ldots, p_{m-1}\rangle$ s.t.\
     		   \\ \mbox{\ \ \ \ \ }(i) $S_{i-1} \subsetneq S_{i}$ and slot $s_i \in S_{i-1}$ of $C_{t-1}$ satisfies $\requestpair{p_i}{S_i}\in \rep(R_{t-1})$, for $1 \le i < m$, and
     		   \\ \mbox{\ \ \ \ \ }(ii) slot $s_m\in S_{m-1}$ of $C_{t-1}$ either
			   \\ \mbox{\ \ \ \ \ \ \ \ } (ii.1) does not satisfy any requests in $\rep(R)$, or 
			   \\ \mbox{\ \ \ \ \ \ \ \ } (ii.2) satisfies a request $\requestpair{p}{S'} \in \rep(R_{t-1})$ such that $S' \supsetneq S_{m-2}$
       		 	\step \label{step: serving one request}
					to obtain $C_t$ and satisfy $\requestpair{p_{i-1}}{ S_{i-1}}$, place $p_{i-1}$ in slot $s_i$, for $1\leq i \leq m$ 
        	  \end{steps}
        \end{steps}
      \end{steps}
    }}
  \caption{Deterministic online \SlotLaminarPaging algorithm~\algRefSearch. 
  Note that in Step~\ref{step: determine sequences}
  we have  $m\le k+1 - |S|$, and that in~(ii), if $s_m$ satisfies $\requestpair{p}{S'} \in \rep(R_{t-1})$ then $m\ge 2$
  (because $C_{t-1}$ does not satisfy $\sigma_t$); thus $S_{m-2}$ is well-defined.
  }\label{fig: refsearch}
\end{figure}

We begin by defining the terminology used in the algorithm and the
proof, and establish some useful properties.  Recall that a
configuration $D$ satisfies a request $r = \requestpair{p}{S}$ if
there exists a slot $s$ in $S$ such that $s$ holds $p$ in $D$; in this
case, we also say that slot $s$ satisfies $r$ in $D$.  A configuration
$D$ is said to \emph{satisfy a set $R$} of requests if it satisfies
every request in $R$.  A set $R$ of requests will be called
\emph{satisfiable} if there exists a configuration that satisfies $R$.
To determine if a set $R$ of requests is satisfied by a configuration,
it is sufficient (and necessary) to examine the maximal subset of
``deepest'' requests in the laminar hierarchy\ignore{, as observed in
Lemma~\ref{obs:satisfiable request sets} below}.  Formally, a request
$\requestpair{p}{S}$ is an \emph{ancestor} (resp., \emph{descendant})
of $\requestpair{p}{S'}$ if $S \supseteq S'$ (resp., $S \subseteq
S'$).  For any set $R$ of requests, define $\rep(R)$ as the set of
requests in $R$ that do not have any proper descendants in $R$.  That
is, $\rep(R) = \{\requestpair{p}{S} \in R \suchthat \forall S'
\subsetneq S, \requestpair{p}{S'} \notin R \}$.  For
$r=\requestpair{p}{S}$, define $\dep(r,R) = \{\requestpair{p}{S'} \in
R \suchthat S\subseteq S'\}$.  Lemma~\ref{obs:satisfiable request
  sets} % (proved in Appendix~\ref{app:slot-laminar paging})
establishes
some basic properties of $\rep(R)$.

%%%%%%%%%%

\begin{lemma}%
  \label{obs:satisfiable request sets}
  Let $R$ be a set of requests.  Then,\\
(i) In any configuration, each slot can satisfy at most one request in $\rep(R)$.\\
(ii) A configuration satisfies $R$ if and only if it satisfies $\rep(R)$.\\
(iii) $R$ is satisfiable iff for any requestable set $S$, $\rep(R)$ has at most $|S|$ requests to subsets of $S$.
\eat{
  \begin{description}\renewcommand{\itemsep}{0in}
	  \item{(i)} In any configuration, each slot can satisfy at most one request in $\rep(R)$.
	  \item{(ii)} A configuration satisfies $R$ if and only if it satisfies $\rep(R)$.  
	  \item{(iii)} $R$ is satisfiable iff for any requestable set $S$, $\rep(R)$ has at most $|S|$ requests to subsets of $S$.
  \end{description}
  }
\end{lemma}
\begin{proof}%[Proof of Lemma~\ref{obs:satisfiable request sets}]
(i) This part holds because any two requests in $\rep(R)$ request
  either different pages or disjoint slot sets.  (ii) Since $\rep(R)
  \subseteq R$, if $R$ is satisfiable, so is $\rep(R)$.  On the other
  hand, if a configuration $D$ satisfies $\rep(R)$ then $D$ satisfies
  $R$, because every $r$ in $R$ is an ancestor of some $r'$ in
  $\rep(R)$ and can be satisfied by the slot satisfying $r'$.

(iii) Suppose that $R$ is satisfiable. If $D$ is a configuration that satisfies $R$ then it also satisfies $\rep(R)$, by~(ii). By~(i), for any requestable set $S$,
    all requests in $\rep(R)$ to subsets of $S$ must be satisfied in $D$ by different slots of $S$, so there can be at most $|S|$ such requests.
	To prove the reverse implication, assume that for any requestable set $S$ there are at most
	  			$|S|$ requests in $\rep(R)$ to subsets of $S$. 
    We construct $D$ top-down. Let $T$ be the root of the laminar hierarchy $\slotsetfamily$. (We could assume that $T = [k]$, but it's not necessary.)
	By our assumption, there are at most $|T|$ requests in $\rep(R)$.
	The children of $T$ in $\slotsetfamily$ are disjoint, so we can distribute these requests to the children of $T$
	in such a way that each child $Q$ is assigned at most $|Q|$ requests from $\rep(R)$, and each request assigned
	to $Q$ is to a subset of $Q$. Continuing this recursively down the tree, we will end up with requests assigned
	to leaves. Then, for any leaf $L$ we can satisfy its assigned requests by different slots in $L$.
\end{proof}

Algorithm~{\algRefSearch} is given in Figure~\ref{fig: refsearch}. It consists of phases. The first phase starts in time step~1, and
each phase ends when adding the current request to the request set from this phase makes it unsatisfiable.
Within a phase, redundant requests, that is those satisfied by the current configuration, are ignored (Step~\ref{step: slp-redundant}).
To serve a non-redundant request $\sigma_t=\requestpair p S$, the cache content is rearranged to free a slot in $S$.
This rearrangement involves shifting the content of some slots that serve requests in $\rep(R)$ along the path from $S$ to the root,
to find a slot that is either unused or holds $p$ (Step~\ref{step: serving one request}).

For technical reasons, in the analysis of Algorithm~{\algRefSearch} it will be useful to introduce a slightly refined
concept of configurations. Given a request set $R$, an \emph{$R$-configuration} is a configuration $D$ in which
each request in $\rep(R)$ is served by exactly one slot.
(By Lemma~\ref{obs:satisfiable request sets}(i), each slot can serve only one request in $\rep(R)$, but in general
in a configuration serving $R$ there may be multiple slots that serve the same request in $\rep(R)$.)
Slots in $D$ that do not serve requests in $\rep(R)$ are called \emph{free in $D$}. 
Observe that each configuration $C_t$ of Algorithm~{\algRefSearch} implicitly is an $R_t$-configuration
-- due to the assignment of slots in Step~\ref{step: serving one request}. 
Also, if the slot $s_m$ chosen by the algorithm in Step~\ref{step: determine sequences}
satisfies condition~(ii.1) then $s_m$ is a free slot of $D$, according to our definition.

The following helper claim, which characterizes when a particular
request is not satisfied by a given configuration, follows directly
from Lemma~\ref{obs:satisfiable request sets}(iii).

%%%%%%%%%%

 \begin{claim}%
  \label{claim:server_exists_higher}
  Let $R$ be a set of requests and $D$ be an $R$-configuration.
  Let also $r = \requestpair{p}{S}$ be a request such that $D$ does not satisfy $r$, yet $R \cup \{r\}$ is satisfiable.  
  Then $D$ has a slot $s$ in $S$ that is either free or satisfies a request $\requestpair{p'}{S'}\in \rep(R)$ where $S\subsetneq S'$. 
\end{claim}

The following lemma establishes the validity of Steps~\ref{step: determine sequences} and~\ref{step: serving one request} of Algorithm~\algRefSearch.
% We defer the proof of Lemma~\ref{lem:sequence_exists} and the full
% proof of Theorem~\ref{thm: laminar upper bound} to
% Appendix~\ref{app:slot-laminar paging}.

%%%%%%%%%%%%%%%%%%%

\begin{lemma}%
  \label{lem:sequence_exists}
  Let $R$ be a set of requests and $D$ be an $R$-configuration.
 Let $r = \requestpair{p_0}{S_0}$ be a request such that $r$ is not satisfied by $D$ and $R\cup \{r\}$ is satisfiable. 
  Then there exist sequences $\langle s_1, \ldots, s_m\rangle$,  $\langle S_0, S_1, \ldots, S_{m-1}\rangle$, and $\langle p_0, p_1, \ldots, p_{m-1}\rangle$ such that 
  (i) $S_{i-1} \subsetneq S_{i}$ and $s_i \in S_{i-1}$ is currently satisfying request $\requestpair{p_i}{S_i}\in \rep(R)$, for $1 \le i < m$, and
  (ii) $s_m\in S_{m-1}$ is either a free slot or is currently satisfying $\requestpair{p_0}{S'} \in  \rep(R)$ for some $S' \supsetneq S_{m-2}$. 
  Furthermore, transforming $D$ by moving page $p_{i-1}$ to slot $s_i$ (and modifying the slot assignment in $D$ accordingly), 
  for $1\leq i \leq m$, yields an ($R \cup \{r\}$)-configuration. 
\end{lemma}
\begin{proof}%[Proof of Lemma~\ref{lem:sequence_exists}]
  The proof is by induction on the depth of $S_0$ in the laminar hierarchy. 
  For the induction base, consider $S_0 = [k]$. Since $r$ is not satisfied by $D$, $R \cup \{r\}$ is satisfiable, 
and every requestable slot set is subset of $[k]$, we obtain from Claim~\ref{claim:server_exists_higher} that there is a free slot $s_1\in S_0$.  The desired claim of the lemma holds 
with $m=1$ and sequences $\langle s_1 \rangle$, $\langle S_0 \rangle$ and $\langle p_0\rangle$ which satisfy (i). Since $s_1$ is free, 
 bringing page $p_0$ to slot $s_1$ yields a ($R\cup \{r\}$)-configuration. 

  We now establish the induction step.  Let $R$, $D$, and $r = \requestpair{p_0}{S_0}$ be as given.  By Claim~\ref{claim:server_exists_higher} there are two cases.  
  In the first case, there is a free slot $s_1 \in S_0$ in $D$.  Then the desired claim holds with $m=1$, and sequences $\langle s_1 \rangle$,  $\langle S_0 \rangle$ and $\langle p_0 \rangle$. 
   Furthermore, as in the base case, since $s_1$ is free, bringing page $p_0$ to slot $s_1$ yields an ($R \cup \{r\}$)-configuration.

  The remainder of this proof concerns the second case, in which there is a slot $s_1\in S_0$ currently satisfying a request $r' = (p_1, S_1)$ in $\rep(R)$ with $S_0 \subsetneq S_1$.  
  Let $D'$ denote the configuration that is identical to $D$ except that $D$ has $p_0$ in slot $s_1$.  
  Since $D$ is an $R$-configuration,
	 no other slot satisfies $r'$ in $D$; the same holds in $D'$.  Hence, $D'$ does not satisfy $r'$.  
	 Furthermore, $D'$ satisfies every request in $\rep(R)$ other than $r'$.   Let $R' = R\cup \{r\} \setminus \dep(r',R)$.   
	 In $D'$, $s_1$ satisfies $r$.  Consider any request $x$ in $R \setminus \dep(r', R)$.  
	 By definition of $\rep(R)$, there exists a request $x'$ in $\rep(R)$ that is a descendant of $x$.  
	 Since $R'$ does not include any ancestors of $r'$, $x'$ is not $r'$ and hence is satisfied by some slot in $D'$.  
	 We thus obtain that $D'$ satisfies $R'$ and, in fact $D'$ is an $R'$-configuration. In $D'$ slot $s_1$ is assigned to $r$,
	 and if there is a request $(p,S')$ in $\rep(R)$ then its assigned slot is designated as free in $D'$. 
	 At the same time, $D'$ does not satisfy $r'$. 
	 Further, since $R' \cup \{r'\}$ is a subset of $R \cup \{r\}$, which is satisfiable, $R' \cup \{r'\}$ is also satisfiable. 
   Since $S_1 \supsetneq S_0$, by the induction hypothesis, 
   there are sequences $\langle s_2,\ldots, s_m\rangle$, $\langle S_1,S_2, \ldots S_{m-1}\rangle$ and  $\langle p_1,p_2,\ldots, p_{m-1}\rangle$ 
   such that (i) $S_{i-1}\subsetneq S_i$ and $s_i \in S_{i-1}$ is currently satisfying  
   $(p_i,S_i)\in \rep(R')$, for $2\leq i < m$; 
   and either (ii.1) $s_m$ is a free slot in $D'$ or (ii.2) is currently satisfying a request $(p_1,S')\in \rep(R')$ for some $S' \supsetneq S_1$. 
Note, however, that $s_m$ has to be a free slot in $D'$ since (ii.2) above cannot hold: any request $(p_1, S')$ is in $\dep(r',R)$, 
all requests of which are excluded from $R'$.
    Furthermore, transforming $D'$ to $D''$ by moving page $p_{i-1}$ to $s_i$ for $2\leq i \leq m$, satisfies $R' \cup \{r'\}$.  

  We now establish the desired claim for $D$, $R$, and $r$.  Consider sequences $\langle s_1,\ldots, s_m\rangle$, $\langle S_0, S_1, \ldots S_{m-1}\rangle$ 
  and  $\langle p_0, \ldots, p_{m-1}\rangle$.  The desired condition~(i) follows from~(i) of the induction hypothesis above and the fact 
  that in $D$, $s_1 \in S_0$ is currently satisfying a request $(p_1, S_1)$ in $\rep(R)$ with $S_0 \subsetneq S_1$.
  For (ii), note that since $s_m$ is a free slot in $D'$, either 
  $s_m$ is a free slot in $D$ or $(p_0, S')$ is in $\rep(R)$ for some $S' \supsetneq S_{m-2}$, thus establishing (ii).  
  Finally, transforming $D$ to $D''$ by moving $p_{i-1}$ to $s_i$ for $1\leq i \leq m$, satisfies $R' \cup \{r'\}$. 
  Since any request satisfying $r'$ also satisfies all ancestors of $r'$, we have $\rep(R \cup \{r\}) = \rep(R' \cup \{r'\})$, implying that $D''$ also satisfies $R \cup \{r\}$.  
  This completes the induction step and the proof of the lemma.
\end{proof}

%%%%%%%%%%%%

\begin{proof}[Proof of Theorem~\ref{thm: laminar upper bound}]
We first argue that at any time $t$, configuration $C_t$ of
{\algRefSearch} satisfies the set $R_t$ of requests from the current
phase of the algorithm.  The proof is by induction on the number of
steps within a phase. When the phase is about to start at time $t$
then $R_{t-1}$ is set to $\emptyset$, so the claim holds. For the
induction step, consider a step $t$ within a phase and assume that
$C_{t-1}$ satisfies $R_{t-1}$.  If $C_{t-1}$ satisfies new request
$\sigma_t$, then by Step~\ref{step: redundant}, $C_t$ satisfies $R_t$.
Otherwise, $R_{t-1} \cup \{\sigma_t\}$ is satisfiable but
$C_{t-1}$ does not satisfy $\sigma_t$.  Then, by
Lemma~\ref{lem:sequence_exists}, Steps~\ref{step: determine sequences}
and~\ref{step: serving one request} derive a configuration $C_t$
satisfying $R_t$, completing the induction step and the argument that
at any time $t$, $C_t$ satisfies $R_t$.

 We next analyze the competitive ratio.  We first show that the number
 of page retrievals during a phase of {\algRefSearch} is at most
 $2\cdot \sumofsizes(\slotsetfamily)$.  Let $R$ denote the set of
 requests in the current phase.  We charge the cost in this phase to
 the depths of the requests in $\rep(R)$.  The cost of Step~\ref{step:
   serving one request} is $m$.  If $s_m$ satisfies condition~(ii.1),
 then $\rep(R \cup \{\sigma_t\}) = \rep(R) \cup \{\sigma_t\}$ and the
 depth of $S$ is at least $m$, so the charge per unit depth is at most
 $1$.  Otherwise, condition~(ii.2) holds and $\rep(R \cup
 \{\sigma_t\})= \rep(R) \cup\{\sigma_t\} \setminus
 \{\requestpair{p}{S'}\}$. In this case we have $\sigma_t$ inherit the
 charges to $\requestpair{p}{S'}$, and we charge the cost of $m$ to
 the difference in depths of $S$ and $S'$, which is at least $m-1$
 (because $S_{m-2}\subsetneq S'$), so the charge per unit of depth is
 at most $m/(m-1)\le 2$. (Note that in this case $m\ge 2$.)  When the
 phase ends, a request at depth $d$ was charged at most $d$ times, and
 these charges include at least one unit charge, so its total
 charge is most $2d-1$. Thus the algorithm's cost per phase is at most
 $2\cdot\sumofsizes(\slotsetfamily)-k \le (2h-1)k$. The optimal cost
 in a phase is at least $1$ as no configuration satisfies all
 requests in the phase and the request that starts the next phase. The
 theorem follows.
\end{proof}

%%% Local Variables:
%%% mode: latex
%%% TeX-master: "0_0_main__short"
%%% End:

%% file: 6_0_all_or_one_paging.tex
%\section{\AllOrOnePaging}%
%\label{sec: all or one paging}
%\input{5_0_all or one_paging}

Recall that \AllOrOnePaging is the extension of standard \Paging that allows two types of requests:
 % \reporttag{(R2.12)}                                                                            % <------------ report tag
A general request for a page $p$,  denoted $\requestpair{p}{\ast}$, can be served by having $p$ in any cache slot.
A specific request $\requestpair{p}{j}$, where $j\in [k]$, must be served by having $p$ in slot $j$ of the cache. 
(Section~\ref{sec: preliminaries} gives a formal definition.) It is a restriction of  \SlotLaminarPaging with $h=2$.

% \mareksrevision{
  For \AllOrOnePaging,
  % in Section~\ref{sec: all or one paging improved lower} we prove that the optimal  
  Section~\ref{sec: all or one paging improved lower} proves that the optimal  
randomized ratio is at least $2H_k-O(1)$.  (Compare it to the upper bound of $12 H_k$, that
follows from Theorem~\ref{thm: slot laminar} for $h=2$.)
Then,
% in Section~\ref{sec: np-completeness of all or one paging} we show
Section~\ref{sec: np-completeness of all or one paging} shows
that the offline problem is $\NP$-hard.
% }

%%% Local Variables:
%%% mode: latex
%%% TeX-master: "0_0_main__short"
%%% End:

%% file: 6_1_all_or_one_lower.tex
\begin{theorem}\label{thm: all or one deterministic 2k-1 lower bound}
  Every online randomized algorithm $\algA$ for the \AllOrOnePaging problem has competitive ratio at least $2H_k - 1$.
\end{theorem}

\begin{proof}
 We establish our lower bound by giving a probability distribution on
 % \reporttag{(R2.13)}                                                                            % <------------ report tag
the input sequences for which any deterministic algorithm $\algA$  has expected cost at least $2H_k-1$ times the optimum cost. 
Without loss of generality we can assume that $\algA$ is \emph{lazy}, in the sense that it retrieves a page only when it is
necessary to satisfy a request.

% \mareksrevision{
  We present the proof in terms of a game between algorithm $\algA$ and
  % the adversary that  
  an adversary who
generates the request sequence and its solution. It is then sufficient to show that
the expected cost of $\algA$ is at least $2H_k-1$ times the adversary's cost. 
We use some fixed $k+1$ pages $p_1, p_2, \ldots, p_{k+1}$ and the
random input sequence generated by the adversary will consist of $L$ phases, where $L$ is an integer that can be made arbitrarily large.
% }

Consider any phase.  To ease notation, by symmetry, assume without loss of generality that when the phase starts
the adversary has pages  $p_1,p_2,\ldots,p_k$ in the cache, with each page $p_i$ in slot $i$, for $i = 1,2,\ldots,k$. 
 % \reporttag{(R1.24)}                                                                            % <------------ report tag
To start the phase, the adversary chooses a random permutation $p_{i_1}, p_{i_2}, \ldots, p_{i_k}$ of these $k$ pages,
replaces $p_{i_k}$ in its cache by $p_{k+1}$, at cost $1$,
then makes request $\requestpair{p_{k+1}}{\ast}$, followed by $k-1$ stages.
Each stage $s=1,2,\ldots,k-1$ consists of $L\cdot (2H_{k} -1)$ repetitions of the request sequence 
\begin{equation*}
\requestpair{p_{i_1}}{i_1}  \,,\, \requestpair{p_{i_2}}{i_2}
			 \,,\, \ldots \,,\, \requestpair{p_{i_{s}}}{i_{s}} \,,\, \requestpair{p_{k+1}}{\ast} \,,
\end{equation*}
which costs the adversary nothing.

% \nymargincomment{E = expectation}
% \mareksrevision{
\newcommand{\event}{{\cal E}}
It remains to bound the expected cost of $\algA$.  Let $\event$ denote the event that for every phase and every stage $s$ in the phase, 
the configuration of $\algA$ at the end of the stage has each page $p_{i_r}$, for $r = 1, 2, \ldots, s$, in slot $i_r$ and one of the slots in $[k] \setminus \braced{i_1,\ldots,i_{s}}$
 % \reporttag{(R1.25)}                                                                            % <------------ report tag
contains $p_{k+1}$.  We will separately bound the expected cost of $\algA$ conditioned on $\event$, and the expected cost of $\algA$ conditioned on $\overline{\event}$
(the complement of $\event$).

We first analyze the expected cost of $\algA$ conditioned on $\event$.  Observe that when a phase starts
$\algA$ and the adversary are in the same configuration. For the first phase this holds because both  $\algA$ and the adversary are
in the initial configuration. For any other phase it follows from condition $\event$ applied to stage $k-1$ of the previous phase.

Now, consider a stage $s$ of a phase. At the beginning of this phase
 % \reporttag{(R1.26)}                                                                            % <------------ report tag
the configuration of $\algA$ has
each page $p_{i_r}$, for $r = 1,2,\ldots,s-1$, in slot $i_r$, and one of the slots in $[k] \setminus \braced{i_1,\ldots,i_{s-1}}$
contains $p_{k+1}$.  Indeed, for $s\ge 2$  this follows directly from condition $\event$ applied to stage $s-1$.
For $s=1$ this follows from $\algA$ and the adversary being in the same configuration when the phase starts,
and from the adversary making the request $\requestpair{p_{k+1}}{\ast}$ right before stage $1$.

Since the probability distribution of $i_s$ is uniform in  $[k]\setminus \braced{i_1,i_2,\ldots,i_{s-1}}$,
the probability that $\algA$ has $p_{k+1}$ in slot $i_s$ equals $1/(k-s+1)$.
If it does, the cost of $\algA$ is at least $2$ in stage $s$, because $p_{i_s}$ will need to be fetched into slot $i_s$
and $p_{k+1}$ will need to be moved to a different slot (to preserve condition $\event$).
So the expected cost of $\algA$ in this stage is at least
$2/(k-s+1)$. Summing over all stages $s = 1,2,\ldots,k-1$ and adding $1$ for the first request, the expected cost of $\algA$ for a phase, 
conditioned on event $\event$, will be at least $2(H_k-1)+1 = 2H_k-1$.  
Therefore, the  expected total cost of $\algA$ over $L$ phases, conditioned on event $\event$, 
is at least $L \cdot(2(H_k-1)+1) = L\cdot(2H_k-1)$.

We next analyze the expected cost of $\algA$ conditioned on $\overline{\event}$.  The event $\overline{\event}$ implies that there is a stage $s$ of a phase in which $\algA$ 
does not end with a configuration in which each page $p_{i_r}$, for $r = 1,2,\ldots,s$, is in slot $i_r$, and page $p_{k+1}$ is in one of the slots in $[k] \setminus \braced{i_1,\ldots,i_{s}}$.  
Since such a configuration satisfies all requests in the stage and $\algA$ is lazy, this implies that $\algA$ never reaches such a configuration in the stage.
(Otherwise $\algA$ would have stayed in this configuration through the rest of the phase.)  
 % \reporttag{(R1.27)}                                                                            % <------------ report tag
There are no other configurations that satisfy all requests in this phase at no cost. Therefore, $\algA$ incurs a cost of at least one during each of the $L \cdot (2H_k - 1)$ repetitions of the request sequence, yielding a total cost in this stage alone of at least $L \cdot (2H_k - 1)$. 

Since the adversary pays 1 for each phase, the total cost of the adversary is $L$. We showed that, whether we condition on $\event$ or $\overline{\event}$, the
expected cost of $\algA$ is at least $L \cdot (2H_k - 1)$,
so this will be true also without any conditioning.
Therefore, the competitive ratio of $\algA$ is at least $2H_k-1$.
% }
\end{proof}

%%% Local Variables:
%%% mode: latex
%%% TeX-master: "0_0_main__short"
%%% End:

%% file: 6_2_all_or_one_np_completeness.tex
The off-line version of {\Paging}, where the request sequence is given upfront, can be solved in time $O(n\log n)$ using
the classical algorithm by Belady~\cite{DBLP:journals/ibmsj/Belay66}. \AllOrOnePaging differs from standard \Paging only by inclusion of
specific requests, which appear easy to handle because they don't give the algorithm any choice.
In this section we show that this intuition is not valid:

\begin{theorem}\label{thm: all or one np-hard}
  Offline {\AllOrOnePaging} is $\NP$-complete. 
\end{theorem}
\begin{proof}[Proof of Theorem~\ref{thm: all or one np-hard}]
Let $G = (V,E)$ be a graph 
  with vertex set $V = \braced{0,1,\ldots,n-1}$. Given an integer $k$, $1\le k\le n$, we 
  compute in polynomial time a request sequence $\requestsequence$ and an integer $F$ such that
  the following equivalence holds: $G$ has a vertex cover of size $k$ if and only if there is a solution
  for $\requestsequence$ whose cost with a cache size $k+2$ is at most $F$. 

  At a fundamental level our proof resembles the argument in~\cite{chrobak_etal_caching_is_hard_2012},
  where $\NP$-completeness of an interval-packing problem was proved. The basic idea of the proof is to
  represent the vertices by a collection of intervals with specified endpoints
  that are to be packed into a strip of width $k$. These intervals will be represented by pairs of requests, 
  one at the beginning and one at the end of the interval, and the strip to be packed is the cache.
  Since the strip's capacity is bounded by $k$, only a subset of intervals can be packed, and
  the intervals that are packed correspond to a vertex cover. 

  There will actually be many ``bundles'' of such intervals, with each bundle
  containing $n$ intervals corresponding to the $n$ vertices. 
  If we had $|E|$ bundles and if we forced each bundle's packing (that is, its corresponding set of vertices)
  to be the same, we could add 
  an edge-gadget to each bundle that will verify that all edges are covered. While it does not
  seem possible to design these bundles to force all bundles' packings to be equal,
  there is a way to design them to ensure that the packing of each bundle
  is dominated (in the sense to be defined shortly) by the next one, and this dominance
  relation has polynomial depth. So with polynomially many bundles we can ensure that there
  will be $|E|$ consecutive equally packed bundles, allowing us to verify whether
  the vertex set corresponding to this packing is indeed a correct vertex cover.

  %%%%%%%%% 

  \myparagraph{Set dominance.}
  We consider the family of all $k$-element subsets of $V$.
  For any two $k$-element sets $X,Y\subseteq V$, we say that $Y$ \emph{dominates} $X$, and denote it $X\preceq Y$,
  if there is a 1-to-1 function $\psi :X \to Y$ such that $x \le \psi(x)$ for all $x\in X$.
  We write $X\prec Y$ iff $X\preceq Y$ and $X\neq Y$.
  The dominance relation is a partial order. The following lemma from~\cite{chrobak_etal_caching_is_hard_2012}
  will be useful:

  \begin{lemma}\label{lem: dominance depth}
    Let $X_1, X_2, \ldots, X_p \subseteq V$ be sets of cardinality $k$ such that
    $X_1 \prec X_2 \prec \ldots \prec X_p$. Then $p \le k(n-k)$.
  \end{lemma}

  %%%%%%%%% 

  \myparagraph{Cover chooser.}
  We start by specifying the ``cover chooser'' sequence $\requestsequence'$ of requests. In
  this sequence some time slots will not have assigned requests. 
  Some of these unassigned slots will be used later to insert requests representing edge gadgets.

  Let $m = |E|+1$. (For notation-related reasons, it is convenient to have $m$ be one larger than
  the number of edges.) Let also $P = k(n-k)+1$ and $B = mP$. In $\requestsequence'$
  we will use the following pages and requests:
  \begin{itemize}
  \item
    We have $nB$ pages $x_{b,j}$, for $b = 0,1,\ldots,B-1$ and $j = 0,1,\ldots,n-1$. For each page $x_{b,j}$
    there are two general requests $\requestpair{x_{b,j}}{\ast}$ in $\requestsequence'$ at time steps
    $\tau_{b,j} = 9(bn+j)$ and $\tau'_{b,j} = 9(bn+j)+9n-6$.  These requests are called \emph{vertex requests}.
    They are grouped into \emph{bundles} of requests, where bundle $b$ consists of all $2n$ general requests to pages
    $x_{b,0}$, $x_{b,1}$, \ldots, $x_{b,n-1}$.
    See Figures~\ref{fig: np-hard vertex requests} and~\ref{fig: np-hard vertex requests zoomin} for illustration.
  \item
    We have $B$ pages $y_b$, for $b = 0,1,\ldots,B-1$. For each page $y_b$
    we have two specific requests $\requestpair{y_b}{k+1}$ and $\requestpair{y_b}{k+2}$ in $\requestsequence'$,
    at times $\theta_b = \tau'_{b,0}-2 = \tau_{b,n-1}+1$ and $\theta'_b = \tau'_{b,0}-1 = \tau_{b,n-1}+2$,
    respectively. For each $b$, these requests are called \emph{$b$-blocking requests}, because
    for each page $x_{b,j}$ in bundle $b$ we have $\theta_b, \theta'_b \in [\tau_{b,j},\tau'_{b,j}]$,
    so these two requests make it impossible to have both requests $\requestpair{x_{b,j}}{\ast}$
    served in cache slots $k+1$ or $k+2$ with only one fault.
  \end{itemize}
  Slots $1,2,\ldots,k$ in the cache will be referred to as \emph{vertex} slots. Slot $k+1$
  is called the \emph{edge-gadget} slot, and slot $k+2$ is called the \emph{junkyard} slot.

  %%%%%%%% 

  \begin{figure}[t]
    \begin{center}
      \includegraphics[width = 5.2in]{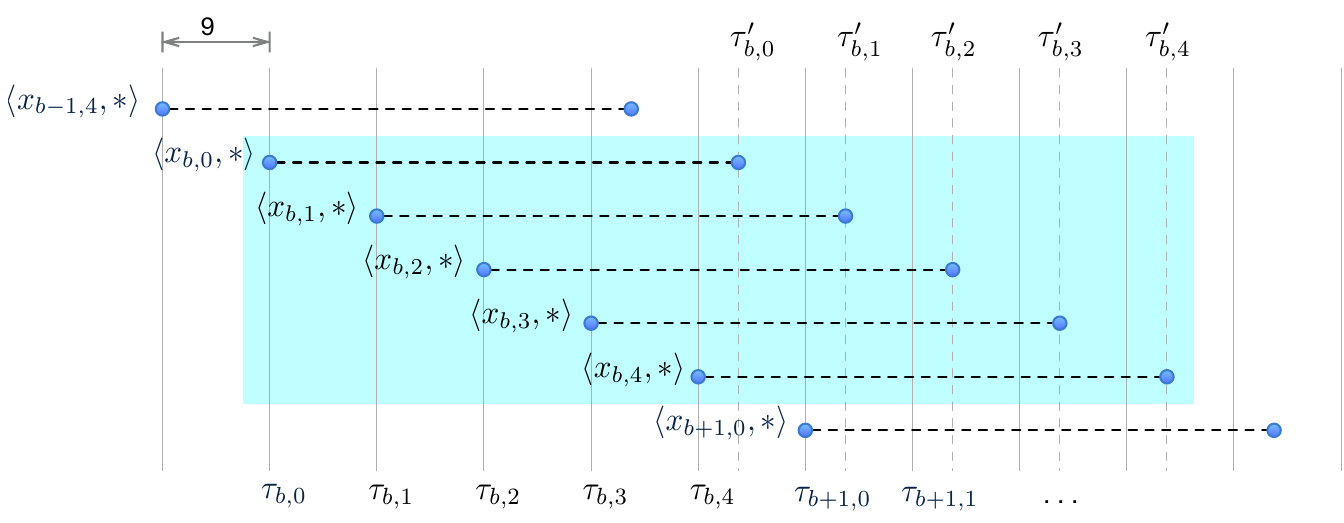}
    \end{center}
    \caption{The sequence of vertex requests, for $n=5$. The shaded region contains requests from bundle $b$.}%
    \label{fig: np-hard vertex requests}
  \end{figure}

  %%%%%%%% 

  %%%%%%%% 

  \begin{figure}[ht]
    \begin{center}
      \includegraphics[width = 4.5in]{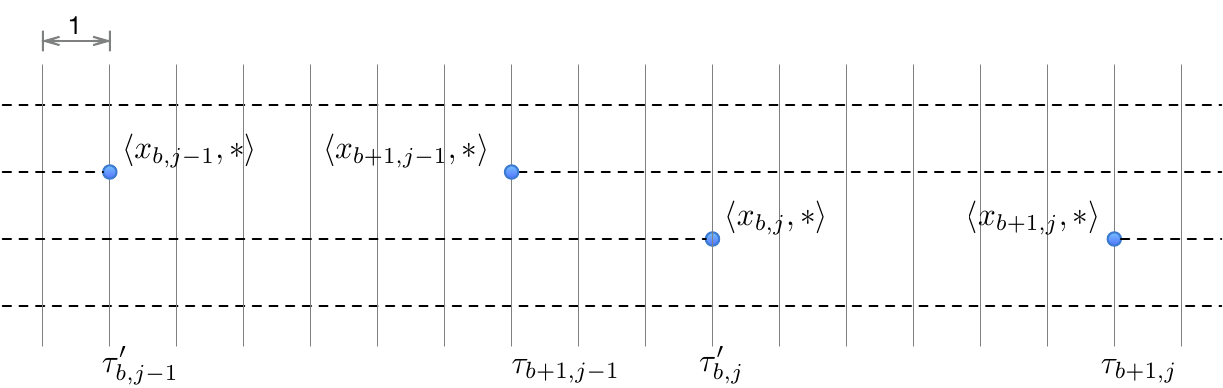}
    \end{center}
    \caption{A more detailed picture showing relations between general requests to pages
      $x_{b,j-1}$, $x_{b,j}$, $x_{b+1,j-1}$, and $x_{b+1,j}$, where $1\le j \le n-1$.}%
    \label{fig: np-hard vertex requests zoomin}
  \end{figure}

  %%%%%%%% 

  Let $F' = (2n-k+2)B$. For any solution $S$ of $\requestsequence'$ and any bundle $b$, denote by $V_{S,b}$ the 
  set of vertices $j\in V$
  for which $S$ does not fault in $\requestsequence'$ on request $\requestpair{x_{b,j}}{\ast}$ at time $\tau'_{b,j}$. 
  In other words, $S$ keeps $x_{b,j}$ in the cache throughout the time interval $[\tau_{b,j}, \tau'_{b,j}]$.

  \begin{lemma}\label{lem: bundle dominance property}
    (a) The minimum number of faults on $\requestsequence'$ in a cache of size $k+2$ is $F'$.
    (b) If $S$ is a solution for $\requestsequence'$ with at most $F'$ faults, then for any $b=0,1,\ldots,B-2$ we have $V_{S,b}\preceq V_{S,b+1}$.
  \end{lemma}
  
  \begin{proof}
    (a) There are $2B$ specific $b$-blocking requests $\requestpair{y_b}{k+1}$ and $\requestpair{y_b}{k+2}$ and all of these are faults.
    Consider a bundle $b$. For this bundle, for each $j$,
    the two requests $\requestpair{x_{b,j}}{\ast}$ at times $\tau_{b,j}$ and $\tau'_{b,j}$ are
    separated by requests $\requestpair{y_b}{k+1}$ and $\requestpair{y_b}{k+2}$. Thus
    if $S$ does not fault at time $\tau'_{b,j}$ then page $x_{b,j}$ must have been stored
    in one of the vertex slots $1,2,\ldots,k$ throughout the time interval $[\tau_{b,j},\tau'_{b,j}]$.
    As there are $k$ vertex slots, $S$ can avoid faulting on at most $k$ requests in bundle $b$.
    So, including the faults at $\requestpair{y_b}{k+1}$ and $\requestpair{y_b}{k+2}$, the number of faults in $S$ associated
    with this bundle $b$
    will be at least $2 + k + 2(n-k) = 2n-k+2$. We thus conclude that the total number of faults is at least $F'$.

    It is also possible to achieve only $F'$ faults on $\requestsequence'$, as follows: for each $b$, and for each vertex
    $j = 0,1,\ldots,k-1$, at time $\tau_{b,j}$ load $x_{b,j}$  into cache slot $j+1$
    and keep it there until time $\tau'_{b,j}$. For $j = k,\ldots,n-1$, load each request to
    $x_{b,j}$ into slot $k+2$. This will give us exactly $F'$ faults.

    (b) If $S$ makes at most $F'$ faults, since there are $2B$ faults on the blocking requests and for
    each bundle $S$ makes at least $2n-k$ faults on vertex requests, $S$ must make \emph{exactly}
    $2n-k$ faults on vertex requests from each bundle, including one request for each vertex $j\in V_{S,b}$
    and two requests for each vertex $j\notin V_{s,b}$.
    If $u \in V_{S,b}$ and $x_{b,u}$ is stored by $S$ in slot $\ell$ of the cache
    throughout its interval $[\tau_{b,u},\tau'_{b,u}]$,
    and if some $x_{b+1,v}$, for
    $v \in V_{S,b+1}$, is stored by $S$ in slot $\ell$ throughout its interval $[\tau_{b+1,v},\tau'_{b+1,v}]$,
    then we must have $v\ge u$. This is because otherwise we would have $\tau_{b+1,v} <  \tau'_{b,u}$, that is
    the intervals of $x_{b,u}$ and $x_{b+1,v}$ would overlap, so we would fault at least three times on the
    requests to these two pages. This implies part~(b).
  \end{proof}

  %%%%%%%% 

  We partition all bundles into \emph{phases}, where phase $p = 0,1,\ldots,P-1$ 
  consists of $m$ bundles $b = pm, pm+1,\ldots, pm+m-1$. (Recall that $m = |E|+1$.)
  The corollary below states that there is a phase $p$ in which all sets $V_{S,b}$ must be equal. It follows
  directly from Lemmas~\ref{lem: dominance depth} and~\ref{lem: bundle dominance property},
  by applying the pigeonhole principle.

  \begin{corollary}\label{cor: bundles stabilize}
    If $S$ is a solution for $\requestsequence'$ with $F'$ faults, then there is index $p$, $0 \le p \le P-1$, for which
    $V_{S,pm} =V_{S,pm+1} = \cdots = V_{S,pm+m-1}$.
  \end{corollary}
  
  %%%%%%%%% 

  \myparagraph{Edge gadget.}
  For each fixed phase $p$, we create $m-1$ edge gadgets, one for each edge.
  Ordering the edges arbitrarily, the gadget for the $e$th edge, where $0 \le e \le m-2$, will be denoted $\edgegadget_{p,e}$, and it
  will consist of $8$ requests between times $\theta'_{pm+e}$ and $\theta_{pm+e+1}$, that is
  in the region where bundles $pm+b$ and $pm+b+1$ overlap.

  Let the $e$th edge be $(u,v)$, where $u < v$. Edge gadget $\edgegadget_{p,e}$ 
  uses six new pages $z_{p,u}$, $z_{p,v}$, $g_{p,u}$, $g_{p,v}$, $h_{p,u}$ and $h_{p,v}$, and 
  consists of the following requests:
  \begin{itemize}
  \item Two specific requests $\requestpair{z_{p,u}}{k+2}$ at times $\tau'_{pm+e,u}+2$ and  $\tau'_{pm+e,u}+4$,
    and two specific requests $\requestpair{z_{p,v}}{k+2}$ at times $\tau'_{pm+e,v}+2$ and  $\tau'_{pm+e,v}+4$.
  \item General requests $\requestpair{g_{p,u}}{\ast}, \requestpair{g_{p,v}}{\ast}$, at times $\tau'_{pm+e,u}+3$ and  $\tau'_{pm+e,v}+3$.
  \item A pair of requests $\requestpair{h_{p,u}}{k+1}, \requestpair{h_{p,u}}{\ast}$, 
    the first one specific and the second one general, at times $\tau'_{pm+e,u}+1$ and  $\tau'_{pm+e,v}+1$, respectively.
  \item A pair of requests $\requestpair{h_{p,v}}{\ast},\requestpair{h_{p,v}}{k+1}$, the first one general and the second one specific,
    at times $\tau'_{pm+e,u}+5$ and  $\tau'_{pm+e,v}+5$, respectively.
  \end{itemize}

  %%%%%%%% 

  \begin{figure}[ht]
    \begin{center}
      \includegraphics[width = 5.5in]{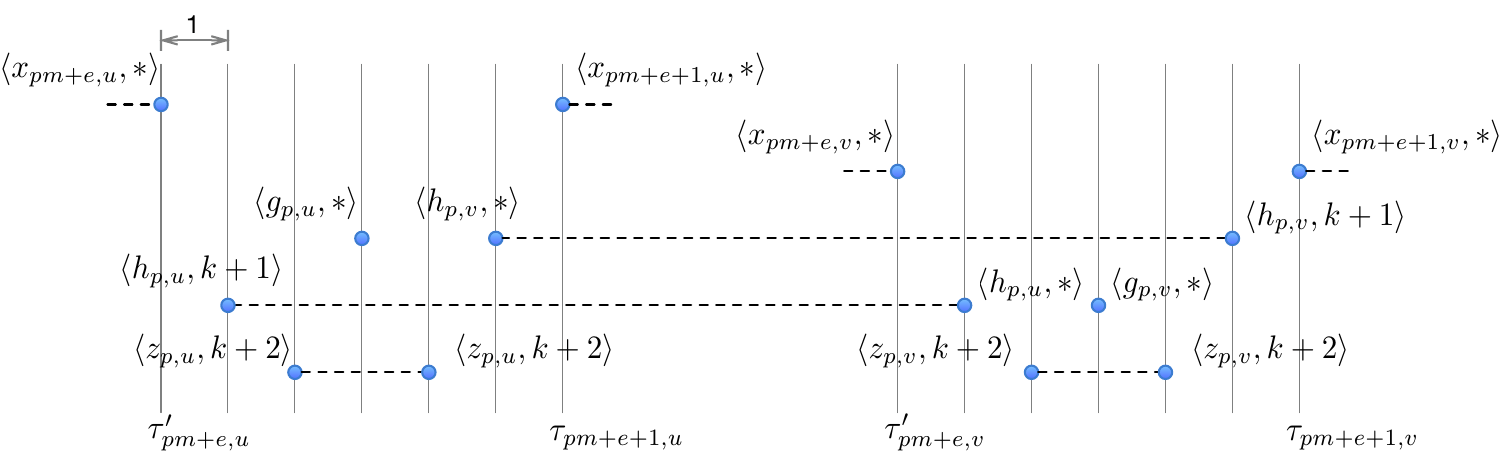}
    \end{center}
    \caption{Gadget $\edgegadget_{p,e}$. Requests $\requestpair{x_{pm+e,u}}{\ast}$ , $\requestpair{x_{pm+e+1,u}}{\ast}$, 
      $\requestpair{x_{pm+e,v}}{\ast}$ , and $\requestpair{x_{pm+e+1,v}}{\ast}$ 
      are not part of this gadget; they are shown only to illustrate how gadget $\edgegadget_{p,e}$ fits into the overall 
      request sequence. }
    \label{fig: edge-e gadget}
  \end{figure}

  %%%%%%%% 

  Consider now possible solutions of gadget $\edgegadget_{p,e}$. Notice that this gadget will require
  $7$ faults regardless of all other requests, since we need to make 
  two faults on requests $\requestpair{g_{p,u}}{\ast}$, $\requestpair{g_{p,v}}{\ast}$,
  at least two faults on requests to pages $\requestpair{z_{p,u}}{k+2}$, $\requestpair{z_{p,v}}{k+2}$, 
  and at least three faults on requests $\requestpair{h_{p,u}}{k+1}$, $\requestpair{h_{p,u}}{\ast}$, 
  $\requestpair{h_{p,v}}{\ast}$, and $\requestpair{h_{p,v}}{k+1}$.
  (This is because if we retain page $h_{p,u}$ in slot $k+1$ until time $\tau'_{pm+e,v}+1$, so that
  we do not fault on $\requestpair{h_{p,u}}{\ast}$,  then we will fault on both requests 
  $\requestpair{h_{p,v}}{\ast}$, and $\requestpair{h_{p,v}}{k+1}$.)
  Another important observation is that if we fault only $7$ times on
  $\edgegadget_{p,e}$ then one of requests $\requestpair{g_{p,u}}{\ast}$, $\requestpair{g_{p,v}}{\ast}$ must be
  put in a vertex cache slot (that is, one of slots $1,2,\ldots,k$).
  A solution that puts $\requestpair{g_{p,u}}{\ast}$ in a vertex slot is called
  a  \emph{$u$-solution} of $\edgegadget_{p,e}$ and a solution that puts $\requestpair{g_{p,v}}{\ast}$ in a vertex
  slot is called a \emph{$v$-solution} of $\edgegadget_{p,e}$. (A solution of $\edgegadget_{p,e}$
  can be both a $u$-solution and a $v$-solution.)

  %%%%%%%%% 

  \myparagraph{Complete reduction.}
  Let $F = F' + 7P(m-1)$. Our request sequence $\requestsequence$ constructed for $G$ 
  consists of $\requestsequence'$ and of all $P(m-1)$ edge gadgets $\edgegadget_{p,e}$ defined above inserted into $\requestsequence$
  at their specified time steps. (At some time steps there will not be any requests.)
  To complete the proof it is now sufficient to show the following claim.

  \begin{claim}\label{cla: reduction works}
    $G$ has a vertex cover of size $k$ if and only if $\requestsequence$ has a solution 
    with at most $F$ faults in a cache of size $k+2$.
  \end{claim}
  
  $(\Rightarrow)$
  Suppose that $G$ has a vertex cover $U$ of size $k$. We construct a solution for $\requestsequence$
  as follows. Each vertex $j\in U$ is assigned to some uniqe vertex cache slot 
  and all $2B$ requests $\requestpair{x_{b,j}}{\ast}$ associated with vertex $j$ are served in this slot. This
  will create $kB$ faults. For $j\notin U$, all requests $\requestpair{x_{b,j}}{\ast}$
  are served in the junkyard slot $k+2$ at cost $2(n-k)B$.
  Together with the $2B$ blocking requests $\requestpair{y_b}{k+1}$ and $\requestpair{y_b}{k+2}$,
  this will give us $F' = (2n-k+2)B$ faults.
  For each $p = 0,1,\ldots,P-1$, and for each $e = 0,1,\ldots,m-2$,
  we do this: Let the $e$th edge of $G$ be $(u,v)$. Since $U$ is a vertex cover, we either have $u\in U$ or
  $v\in U$. If $u\in U$, we use the $u$-solution for gadget $\edgegadget_{p,e}$, with
  $g^\ast_{p,u}$ served in the cache slot associated with $u$.
  If $v\in U$, we use the $v$-solution for gadget $\edgegadget_{p,e}$, with
  $\requestpair{g_{p,v}}{\ast}$ served in the cache slot associated with $v$.
  This will give us $7$ faults for this gadget, adding up to $7P(m-1)$ faults on all edge gadgets.
  Then the total number of faults on $\requestsequence$ will be $F' + 7P(m-1) = F$.

  $(\Leftarrow)$
  Now suppose that there is a solution $S$ for $\requestsequence$ with at most $F$ faults.
  By the earlier observations, we know that $S$ must have exactly $F$ faults, including
  exactly $F'$ faults on the request in $\requestsequence'$ and exactly $7$ faults per each edge gadget.
  As there are $F'$ faults on $\requestsequence'$, we can find some $p$, $0\le p\le P-1$,
  such that $V_{S,pm}= V_{S,pm+1} = \cdots = V_{S,pm+m-1}$, per Corollary~\ref{cor: bundles stabilize}.
  Let $U = V_{S,pm}$. For each $j\in U$, all requests $\requestpair{x_{b,j}}{\ast}$, for $b = pm, pm+1,\ldots,pm+m-1$,
  must be in the same slot, that we refer to as the slot associated with vertex $j$.
  The size of $U$ is $k$, and we claim that $U$ must be
  a vertex cover. To show this, let $(u,v)$ be an edge, and let $e$ be its index.
  Solution $S$ makes $7$ faults on gadget $\edgegadget_{p,e}$, so for this gadget it
  must be either a $u$-solution or a $v$-solution.
  If it is a $u$-solution then $\requestpair{g_{p,u}}{\ast}$ is served in some vertex
  cache slot. But the only vertex slot available in that time step is the
  slot associated with vertex $u$. This means that $u$ must be in $U$.
  The case of a $v$-solution is symmetric. Thus we obtain that either $u\in U$ or $v\in U$.
  This holds for each edge, implying that $U$ is a vertex cover.
  This proves the claim, and completes the proof of the theorem.
\end{proof}

%%% Local Variables:
%%% mode: latex
%%% TeX-master: "0_0_main"
%%% End:

%% file: 7_weighted_all_or_one.tex
% \mareksrevision{
  % In Section~\ref{sec: introduction} we 
  Section~\ref{sec: introduction}
  introduced the generalization of the standard \kServer problem
called \HeteroServer, where each request, in addition to the request point, specifies also
a subset of servers that are allowed to serve the request.
The previous section focussed on the special case of \HeteroServer in uniform metrics.
Extending this work beyond uniform metrics,
% we now address
this section addresses
 % \reporttag{(R1.28)}                                                                            % <------------ report tag
\WeightedAllOrOnePaging,  the natural weighted variant of \AllOrOnePaging
(allowing general and specific requests) in which the pages have weights
and the cost of retrieving a page is its weight.
% }
This is equivalent to  \HeteroServer
in star metrics with requestable-set family $\slotsetfamily=\braced{[k]} \cup \braced{\braced{s} \suchthat s\in [k]}$.
This section proves the following theorem:

%%%%%%%%%%%%%%%

\begin{theorem}\label{thm: wtd alg}
  \WeightedAllOrOnePaging has  a deterministic $O(k$)-competitive online algorithm.
\end{theorem}

%%%%%%%%%%%%%%%

\begin{figure}[t]
  \begin{mdframed}[userdefinedwidth=\linewidth]\hspace*{-0.02\linewidth}
    \begin{minipage}{\linewidth}

      \begin{steps}

      \item[] \textbf{input:} \WeightedAllOrOnePaging instance $(k, \slotrequestsequence)$,
        where $\slotrequest_t=\slotrequestpair {p_t} {s_t}$ for $t\in[T]$

        \step initialize $\capacity[t] \gets \credit[t] \gets 0$ for each $t\in [T]$
	
        \step\label{step: simplifying assumption}
		% \mareksrevision{
		Assume that $\slotrequest_i = \slotrequestpair{\artificialpage}{i}$ for $i\in [k]$
        \algcomment{--- $k$ specific requests to artificial weight-0 page $\artificialpage$}
		% }
		
        \step\label{step: loop}
        for $t\gets k+1,k+2, \ldots, T$:
        \begin{steps}
			
          \step if $\slotrequestpair {p_t} {s_t}$ is a specific request
          with no equivalent request $t'$
          (s.t.~$\slotrequestpair {p_{t'}} {s_{t'}} = \slotrequestpair {p_{t}} {s_{t}}$) in the cache:
  
          \begin{steps}

            \step\label{step: clear request}\label{step: clear slot}
            evict  any cached general request to page $p_{t}$, and any cached request in slot $s_t$

            \step put $t$ in slot $s_t$
            \algcomment{--- note $\capacity[t] = \credit[t] = 0$}
          \end{steps}

          \step\label{step: cache general}
		  % \mareksrevision{
          else if $\slotrequestpair {p_t} {s_t} = \slotrequestpair {p_t} {\ast}$ is a general request
          not satisfied by any cached request $t'$ (s.t.~$p_{t'} = p_{t}$):
		  % }

          \begin{steps} \smallskip
            \step\label{step: define}
            define~~~{}
            $\begin{cases}
              & \ell_t(s) \coloneqq \max \{ t' \le t \suchthat s_{t'} = s\}\ \textrm{for}\ s\in [k]
              ~~~\algcomment{--- most recent specific request to slot $s$} \smallskip
              \\
              &A \coloneqq \{s\in [k]: \capacity[\ell_t(s)] \ge \half \razy \wt{p_t} \textit{ and \(s\) does not hold a specific request} \}  \smallskip
              \\
              & B \coloneqq \{s\in [k] : \textit{slot \(s\) holds a general request of weight at least $\half \razy \wt{p_t}$} \}
            \end{cases}$\smallskip
            
            \step\label{step: loop condition}
            while ${|A| \le |B|}$:
            \begin{steps}
				
              \step\label{step: raise}
              continuously raise $\capacity[\ell_t(s)]$ for $s\in[k]$  and $\credit[t']$ for each cached request $t'$, at unit rate,
			  
              \step\label{step: raise evict}
              evicting each request $t'$ such that $\credit[t'] = \wt{p_{t'}}$, and updating $A$ and $B$ continuously
			  
            \end{steps}
            
            \step\label{step: well defined}\label{step: place evict}
            choose a slot $s \in A\setminus B$;  evict the request $t'$ currently in slot $s$ (if any)

            \step put $t$ in slot $s$ %            initialize $\credit[t] \gets \wt {p_t}$
            \algcomment{--- note $\credit[t] = 0$}
          \end{steps}
          \step\label{step: redundant}
          else: classify the (already satisfied) request as \emph{redundant} and ignore it\footnote
          {And elsewhere (e.g.~the definition of $\ell_t(s)$ in Step~\ref{step:  define}) restrict to previous times $t'$ that were not ignored in this way.}
        \end{steps}
      \end{steps}
    \end{minipage}
  \end{mdframed}
  \caption{An $O(k)$-competitive online algorithm for \WeightedAllOrOnePaging.
    % \mareksrevision{
	Following our convention, we present the algorithm as caching request times
    rather than pages, with the understanding that request $t$ actually represents page $p_t$.
	% }
  }\label{fig: alg}
\end{figure}

% \mareksrevision{
The bound is optimal up to a constant factor, as the optimal ratio for standard \WeightedPaging is $k$.
% }
 % \reporttag{(R2.14)}                                                                            % <------------ report tag
Figure~\ref{fig: alg} shows the algorithm. It is implicitly a linear-programming primal-dual algorithm.
 % \reporttag{(R1.29,30)}                                                                            % <------------ report tag
Note that the standard linear program for standard \WeightedPaging
doesn't have constraints that force pages into specific slots---indeed, those constraints make even the unweighted problem an $\NP$-hard
special case of Multicommodity Flow. As a small example that illustrates the challenge, consider a cache of size $k=2$, and
 % \reporttag{(R1.31)}                                                                            % <------------ report tag
repeatedly make three requests: a general request to a weight-1 page,
and specific requests to different weight-zero pages in slots 1 and 2.
The weight-zero requests force the weight-1 page to be evicted with each round,
so the optimal cost is the number of rounds. But the solution of the classical linear-program relaxation
will have value $1$. Thus this linear program cannot be used to bound the competitive ratio.

% \mareksrevision{
Let $\slotrequestsequence = (\slotrequest_t=\slotrequestpair {p_t} {s_t})_{t=1}^T$ be the request sequence.
For convenience, in the proof we % will
identify requests with the time when they are issued;
that is, by ``request $t$'' we % will
mean request $\slotrequest_t$.
Also, our algorithm needs to keep track not only of the cache content, but also of
the times when each page in the cache was retrieved. For this reasons, it's convenient to think about the
 % \reporttag{(R1.33)}                                                                            % <------------ report tag
algorithm and the optimal solution as caching \emph{requests} or \emph{request times}
rather than pages, with the understanding that ``request $t$'' in the cache actually represents page $p_t$.  
We adopt this convention throughout the proof.
% }

Here is a sketch of the proof of Theorem~\ref{thm: wtd alg}, then the detailed proof.
  Fix an optimal solution $\solutionC$, that is
  $\opt(\slotrequestsequence) = \cost(\solutionC)$.  For each $t\in
  [T]$, let $x_t\in\{0,1\}$ be an indicator variable for the event
  that $\solutionC$ evicts request $t$ before satisfying another
  request $t'>t$ with the same page/slot pair that satisfied $t$.  Let
  $R\subseteq [T]$ be the set of all specific requests, and for each
  $t\in R$, let $y_{t}$ be the amount $\solutionC$ pays to retrieve
  pages into slot $s_t$ before the next specific request to slot $s_t$
  (if any).  Define the \emph{pseudo-cost} of the optimal solution to
  be $\sum_{t=1}^T \wt{p_t} x_t + \sum_{t\in R} y_t$.  The pseudo-cost
  is at most $\two \opt(\slotrequestsequence)$.  As the algorithm
  proceeds, define the \emph{residual cost} to be
  $\sum_{t=1}^T \max(0, \wt{p_t} x_t - \credit[t]) + \sum_{t\in R} \max(0, y_{t} - \capacity[t])$.
\eat{  \[\textstyle\sum_{t=1}^T \max(0, \wt{p_t} x_t - \credit[t])
    + \sum_{t\in R} \max(0, y_{t} - \capacity[t]).\]}
  The residual cost is initially the pseudo-cost (at most $\two\opt(\slotrequestsequence)$),
  and remains non-negative throughout,
  so the total decrease in the residual cost is at most $\two\opt(\slotrequestsequence)$.
  One can show (Lemma~\ref{lemma: k opt charge}) that
  whenever the algorithm is raising credits and capacities at time $t$,
  there is either a cached request $t'$ with $x_{t'} = 1$ and $\credit[t'] < \wt{p_{t'}}$,
  or there is a slot $s$ with $y_{t'} > \capacity[t']$, where $t'=\ell_t(s) \in R$.
  It follows that the residual cost is decreasing at least at unit rate in Step~\ref{step: raise}.

  On the other hand, the algorithm is raising $k$ capacities and at
  most $k$ credits, so the value of \(\phi = \sum_{t=1}^T \credit[t] +
  \sum_{t\in R} \capacity[t]\) is increasing at rate at most $2k$.
  So, the final value of $\phi$ is at most
  $4k\,\opt(\slotrequestsequence)$.  To finish, we show by a charging
  argument that the algorithm's cost is at most $\six \phi + \three
  \opt(\slotrequestsequence) \le (24k+3)\,\opt(\slotrequestsequence)$.
  
  \smallskip
  
Here is the detailed proof.
Consider any execution of the algorithm on a $k$-slot instance $\slotrequestsequence$.
To ease notation and streamline the analysis, without loss of generality we % will
make the following assumptions:
%
% \mareksrevision{
\begin{itemize}[leftmargin=*]\renewcommand{\itemsep}{-0.02in}
\item The first $k$ requests are specific requests for an artificial weight-zero page $\artificialpage$
		in each of the $k$ slots.
\item 
Each request is not redundant (per Step~\ref{step: redundant}).
\item The last $k$ requests are specific requests for an artificial weight-zero page $\artificialpage$
 % \reporttag{(R1.32)}                                                                            % <------------ report tag
		in each of the $k$ slots.
\end{itemize}
These assumptions can indeed be made without loss of generality, as
the zero-weight requests do not have any cost, the algorithm ignores
redundant requests, and removing redundant requests doesn't increase
the optimum cost. 
% }

\smallskip
We first prove a key lemma used in the proof of the theorem.

\begin{lemma}\label{lemma: k opt charge}
  Suppose that, while responding to a general request $t$, the algorithm is executing Step~\ref{step: raise}
  (that is, the loop condition in Step~\ref{step: loop condition} is satisfied).
  Then, just after $\solutionC$ has responded to request $t$, either 
  %
% \mareksrevision{
\begin{description}\setlength{\itemsep}{-0.03in}
   \item{(i)} some request $t'$ currently cached by the algorithm is not in $\solutionC$'s cache, or
 % \reporttag{(R1.34)}                                                                            % <------------ report tag
	  %
   \item{(ii)} for some slot $s\in [k]$, after the most recent specific request $\ell_t(s)$ to slot $s$
  		solution $\solutionC$ has incurred cost more than $\capacity[\ell_t(s)]$ for retrievals into $s$.
\end{description}
% }
\end{lemma}

\begin{proof}
If $\solutionC$ satisfies property~(i), we are done. So assume that it doesn't.  We will show that
then property~(ii) holds.  If~(i) doesn't hold then, just after responding to request $t$, in addition to the current general request $p_t$, 
solution $\solutionC$ caches every request $t'$ that is cached by the algorithm. This,
together with the loop condition, implies that $\solutionC$ has at least
$|B|+1 \ge |A|+1$ generally requested pages of weight at least $\half\razy\wt {p_t}$ in its cache.
Thus one of these pages, say $p_{t'}$, is in a slot $s\notin A$.
The choice of $p_{t'}$ and the definition of $A$ imply then that
the cost of $\solutionC$ for retrievals into $s$ after time $\ell_t(s)$ 
is at least $\wt{p_{t'}} \ge \half\razy \wt {p_t} > \capacity[\ell_t(s)]$, so property~(ii) holds.
\end{proof}

\begin{proof}[Proof of Theorem~\ref{thm: wtd alg}]
  Fix an optimal solution $\solutionC$, that is $\opt(\slotrequestsequence) = \cost(\solutionC)$.
  For each $t\in [T]$,
  let $x_t\in\{0,1\}$ be an indicator variable for the event
  that $\solutionC$ evicts request $t$ before satisfying another request $t'>t$
  with the same page/slot pair that satisfied $t$.
  Let $R\subseteq [T]$ be the set of all specific requests, and for each $t\in R$,
  let $y_{t}$ be the amount $\solutionC$ pays to retrieve
  pages into slot $s_t$ before the next specific request to slot $s_t$ (if any).
  Define the \emph{pseudo-cost} of the optimal solution
  to be $\sum_{t=1}^T \wt{p_t} x_t + \sum_{t\in R} y_t$.  
  The pseudo-cost is at most $\two \opt(\slotrequestsequence)$.
  As the algorithm proceeds, define the \emph{residual cost} to be
  $\sum_{t=1}^T \max(0, \wt{p_t} x_t - \credit[t])
    + \sum_{t\in R} \max(0, y_{t} - \capacity[t])$.
\eat{  \[\textstyle\sum_{t=1}^T \max(0, \wt{p_t} x_t - \credit[t])
    + \sum_{t\in R} \max(0, y_{t} - \capacity[t]).\]}
  The residual cost is initially the pseudo-cost (at most $\two\opt(\slotrequestsequence)$),
  and remains non-negative throughout,
  so the total decrease in the residual cost is at most $\two\opt(\slotrequestsequence)$.
  By Lemma~\ref{lemma: k opt charge},\footnote
  {This interpretation of the problem via covering constraints handled via residual costs
    follows~\cite{Koufogiannakis13Greedy}.  It can be recast as
    a linear-programming primal-dual argument, or as (a generalization of) the local-ratio method~\cite[\S 5 \& \S 6]{Koufogiannakis13Greedy}.
  }
  whenever the algorithm is raising credits and capacities at time $t$,
  there is either a cached request $t'$ with $x_{t'} = 1$ and $\credit[t'] < \wt{p_{t'}}$,
  or there is a slot $s$ with $y_{t'} > \capacity[t']$, where $t'=\ell_t(s) \in R$.
  It follows that the residual cost is decreasing at least at unit rate in Step~\ref{step: raise}.

  On the other hand, the algorithm is raising $k$ capacities and at most $k$ credits, so the value of
  \(\phi = \sum_{t=1}^T \credit[t] + \sum_{t\in R} \capacity[t]\)
  is increasing at rate at most $2k$.  So, the final value of $\phi$ is at most $4k\,\opt(\slotrequestsequence)$.
  To finish, we show that the algorithm's cost is at most  $\six \phi + \three \opt(\slotrequestsequence) \le (24k+3)\,\opt(\slotrequestsequence)$.\,\footnote{%
  This constant can be reduced with more careful analysis.}
  Count the costs that the algorithm pays as follows:
  \begin{enumerate}
  \item \emph{Requests remaining in the cache at the end (time $T$).}
    By the assumption on the last $k$ requests, these cost nothing to bring in.
    All other requests are evicted.  
  \item \emph{Requests evicted in Line~\ref{step: raise evict}.}
    Each such request $t'$ is evicted only after $\credit[t']$ reaches $\wt{p_{t'}}$.
    So these have total weight at most $\sum_{t'=1}^T \credit[t']$.
  \item \emph{Specific requests $t'$ evicted from slot $s_t$ in Line~\ref{step: clear slot}.}
    Throughout the time interval $\interval {t'} {t-1}$,   the algorithm has $p_{t'}$ in slot $s_{t'}=s_t$,  and
    $\slotrequestsequence$ has neither an equivalent specific request
    nor a general request to $p_t$ (by our non-redundancy assumption).
    The optimal solution $\solutionC$ has $p_{t'}$ in slot $s_{t'}$ at time $t'$,
    but not at time $t$, so evicts it during $\interval {t'+1} {t}$.
    So the total cost of such requests is at most the total weight of specific requests evicted by $\solutionC$,
    and thus at most $\opt(\slotrequestsequence)$.
  \item \emph{General requests evicted from slot $s_t$ in  Line~\ref{step:  clear slot}.}
    By Line~\ref{step: well defined},
    any general request in slot $s_t$ at time $t$ has weight at most $\two\capacity[\ell_{t-1}(s_t)]$.
    So the total weight of such requests is at most $\two \sum_{t'\in R}\capacity[t']$.
  \item \emph{General requests to page $p_t$ evicted in Line~\ref{step: clear request}.}
    The algorithm replaces each such general request $t'$
    by a specific request $t$ (which it later evicts, unless the weight is zero) to the same page.
    Have general request $t'$ \emph{charge} its cost $\wt{p_{t'}} = \wt{p_t}$,
    and any amount charged to $t'$ (in Item 6 below), to specific request $t$.
    (We analyze the charging scheme for Items 5 and 6 below.) 
  \item \emph{{General} requests $t'$ evicted in Line~\ref{step: place
    evict}.}  Have request $t'$ charge the cost of its eviction, and
    any amount charged to $t'$ to
    % general
    request $t$.  Since the slot
    holding $p_{t'}$ is not in $B$, $\wt{p_{t'}}
    < \half\razy \wt{p_{t}}$.
  \end{enumerate}

  Each general request $t$ receives at most one charge in Item~6, from a request $t'$ of at most half the weight of $t$; 
  this general request $t'$ may also receive such charges, forming a chain of charges, but since the weights of the
  requests in this chain decrease geometrically, $t$ is charged at most its weight.
  In Item~5, each specific request $t$ is charged by at most one general request $t'$ of the same weight, that may also
  carry the chain charge not exceeding its weight. So this specific request is charged at most twice its weight.
  Overall, the charge of each request from Items~5 and~6 is at most twice its weight.
  
  The total weight of evictions considered in Items~1, 2, 3, and~4 is at most  $\two\phi + \opt(\slotrequestsequence)$.
  Adding also the charges to these items by evictions considered in Items~5 and~6, we obtain that 
the total cost of the algorithm is bounded by  $\three (\two\phi + \opt(\slotrequestsequence)) = \six \phi + \three \opt(\slotrequestsequence)$.
\end{proof}

%%%%%%%%%%%%%%

%%%%%%%%%%%%%%

%%% Local Variables:
%%% mode: latex
%%% TeX-master: "0_0_main__short"
%%% End:

%% file: 8_open_problems.tex
The results here suggest many open problems and avenues for further research. 
Closing or tightening gaps left by our upper and lower bounds would be of interest.
In particular:
\begin{itemize} 
\item For \SlotHeteroPaging, is the upper bound in Theorem~\ref{thm: subset server upper bound} tight
for every $\slotsetfamily\subseteq 2^{[k]}\setminus\{\emptyset\}$, within $\poly(k)$ factors?  
\item
  % \mareksrevision{
What is the best randomized competitive ratio for \SlotHeteroPaging, 
for arbitrary $\slotsetfamily\subseteq 2^{[k]}\setminus\{\emptyset\}$?
Is it possible to achieve ratio that is a poly-logarithmic function of $\slotsetfamily$ and $k$?
% }
%
\item
For \PageLaminarPaging it is easy to show a lower bound of $\Omega(h)$, even for $k=1$ and
for randomized algorithms. But it still may be possible to eliminate or reduce the multiplicative
dependence on $h$. For example, is it possible to achieve ratio
$O(h+k)$ with a deterministic algorithm and $O(h+H_k)$ with a randomized algorithm?
Similarly, does \SlotLaminarPaging (where $h\le k$) admit an
$O(k)$ deterministic ratio and $O(\log k)$ randomized ratio?
\item
For deterministic \AllOrOnePaging, we conjecture that the optimal
ratio is $2k-1$. (For $k=2$ we can show an upper bound of $3$.)  In
the randomized case, can ratio $2H_k-1$ be achieved?
\ignore{For the deterministic case of \AllOrOnePaging the gap is already quite small,
between $2k-1$ and $2k+14$ (see Section~\ref{sec: all or one paging improved upper} and~\cite{castenow+fkmd:server}).
We conjecture that the optimal ratio is $2k-1$. (For $k=2$ we can show an upper bound of $3$.)
In the randomized case, can ratio $2H_k-1$ be achieved?}
\item
For \WeightedAllOrOnePaging, is the optimal randomized ratio $O(\polylog(k))$?
\item
The status of \HeteroServer in arbitrary metric spaces is wide open.
 % \reporttag{(R1.35)}                                                                            % <------------ report tag
Can ratio dependent only on $k$ be achieved? This question, while challenging,
could still be easier to resolve for \HeteroServer than for \GenerServer. 
\end{itemize}

%%% Local Variables:
%%% mode: latex
%%% TeX-master: "0_0_main__short"
%%% End: